\documentclass[sigconf, nonacm,  pdfa]{acmart}

\usepackage[a-3b,mathxmp]{pdfx}
\usepackage{colorprofiles}

\newcommand\vldbdoi{10.14778/3659437.3659444}

\newcommand{\dom}[1]{\mathbb{#1}}
\newcommand{\mech}[1]{\mathcal{#1}}

\newcommand{\ex}[2]{{\ifx&#1& \mathbb{E} \else \underset{#1}{\mathbb{E}} \fi \left[#2\right]}}
\newcommand{\pr}[2]{{\ifx&#1& \mathbb{P} \else \underset{#1}{\mathbb{P}} \fi \left[#2\right]}}

\newcommand{\new}[1]{{#1}}

\usepackage{ragged2e}
\usepackage{stmaryrd}
\usepackage[linesnumbered]{algorithm2e}
\RestyleAlgo{boxruled}
\usepackage{multirow}

\usepackage[many]{tcolorbox}    
\newtcolorbox{mybox}[1]{%
    tikznode boxed title,
    enhanced,
    arc=0mm,
    interior style={white},
    attach boxed title to top center= {yshift=-\tcboxedtitleheight/2},
    fonttitle=\bfseries,
    colbacktitle=white,coltitle=black,
    boxed title style={size=normal,colframe=white,boxrule=0pt},
    title={#1}}

\begin{document}
\title{Privacy Amplification via Shuffling: Unified, Simplified, and Tightened}

\author{Shaowei Wang}
\affiliation{%
  \institution{Guangzhou University}
  \city{Guangzhou}
  \country{China}
}
\email{wangsw@gzhu.edu.cn}

\author{Yun Peng}
\affiliation{%
  \institution{Guangzhou University}
  \city{Guangzhou}
  \country{China}
}
\email{yunpeng@gzhu.edu.cn}

\author{Jin Li}
\affiliation{%
  \institution{Guangzhou University}
  \city{Guangzhou}
  \country{China}
}
\email{lijin@gzhu.edu.cn}

\author{Zikai Wen}
\affiliation{%
  \institution{Virginia Tech}
  \city{Virginia}
  \country{U.S.A}
}
\email{zkwen@vt.edu}

\author{Zhipeng Li}
\affiliation{%
  \institution{Guangzhou University}
  \city{Guangzhou}
  \country{China}
}
\email{2006200002@e.gzhu.edu.cn}

\author{Shiyu Yu}
\affiliation{%
  \institution{Guangzhou University}
  \city{Guangzhou}
  \country{China}
}
\email{2112106255@e.gzhu.edu.cn}

\author{Di Wang}
\affiliation{%
  \institution{King Abdullah University of Science and Technology}
  \city{Thuwal}
  \country{Saudi Arabia}
}
\email{di.wang@kaust.edu.sa}

\author{Wei Yang}
\affiliation{%
  \institution{University of Science and Technology of China}
  \city{Suzhou}
  \country{China}
}
\email{qubit@ustc.edu.cn}

\begin{abstract}
The shuffle model of differential privacy provides promising privacy-utility balances in decentralized, privacy-preserving data analysis. However, the current analyses of privacy amplification via shuffling lack both tightness and generality. To address this issue, we propose the \emph{variation-ratio reduction} as a comprehensive framework for privacy amplification in both single-message and multi-message shuffle protocols. It leverages two new parameterizations: the total variation bounds of local messages and the probability ratio bounds of blanket messages, to determine indistinguishability levels. Our theoretical results demonstrate that our framework provides tighter bounds, especially for local randomizers with extremal probability design, where our bounds are exactly tight. Additionally, variation-ratio reduction complements parallel composition in the shuffle model, yielding enhanced privacy accounting for popular sampling-based randomizers employed in statistical queries (e.g., range queries, marginal queries, and frequent itemset mining). Empirical findings demonstrate that our numerical amplification bounds surpass existing ones, conserving up to $30\%$ of the budget for single-message protocols, $75\%$ for multi-message ones, and a striking $75\%$-$95\%$ for parallel composition. Our bounds also result in a remarkably efficient $\tilde{O}(n)$ algorithm that numerically amplifies privacy in less than $10$ seconds for $n=10^8$ users.
\end{abstract}

\maketitle



\section{Introduction}

\new{In our increasingly digital world, safeguarding data privacy has become eminent, especially in sensitive sectors like census \cite{hirschman2000meaning}, healthcare \cite{raghupathi2014big}, and e-commerce \cite{bennett2007netflix}. As we gather data that contain references to identifiable personal information, the risk of privacy breaches emerges as a substantial threat to both individuals and organizations \cite{raghupathi2014big,liu2015data}. In response to this growing concern, and in light of stringent data privacy regulations~\cite{voigt2017eu,calzada2022citizens}, there is an urgent need for robust privacy protection measures, particularly in data analysis and learning. A promising answer lies in the concept of differential privacy (DP)~\cite{dwork2006differential}, a mathematical framework ensuring data privacy while maintaining data utility.}

\new{DP traditionally operates in two models: the curator model and the local model, each serving different types of data analysis needs. The curator model, widely used for tasks like count queries \cite{chen2011publishing,chen2014correlated,li2014data,li2015matrix}, data mining \cite{hu2015differential,bonomi2013mining,liu2021projected}, and general SQL engines \cite{johnson2018towards,kotsogiannis2019privatesql}, involves a trusted curator who applies privacy-preserving mechanisms before releasing data. The local model, on the other hand, offers privacy at the individual data source level, useful for queries like histogram query \cite{kairouz2016discrete,bao2021cgm}, range query \cite{2019answering,du2021ahead,ren2022ldp}, marginal query \cite{cormode2018marginal,xu2019dpsaas,xu2020collecting}, frequent itemset mining \cite{wang2018locally,wang2020set} and machine learning \cite{duchi2013local,wang2018empirical,wang2019sparse,truex2020ldp}. Despite their extensive applications in database, both models encounter difficulties in striking the balance between privacy and utility, particularly in decentralized settings where data control and processing are distributed among multiple participants.}


\new{The shuffle model of differential privacy \cite{bittau2017prochlo,erlingsson2019amplification} emerges as a better solution that provides more effective balance between privacy protection and data utility in decentralized environments (e.g., federated analytics~\cite{balle2020private,ghazi2020fewer,ghazi2021differentially,feldman2022hiding,cheu2022differentially} and federated machine learning~\cite{girgis2021shuffled,girgis2021shuffledtradeoffs,tenenbaum2021differentially,garcelon2022privacy}). In this model, taking the local $\epsilon_0$-DP mechanisms~\cite{wang2017locally,wang2019local,wang2022analyzing} as an instance, each user first applies a randomizer to transform their data into one or more messages. These messages are then sent to and mixed up by a ``shuffler'' before they being sent to the statistician (i.e., the server or the analyzer).  The shuffler is the key player who obscures the origins of each piece of data. This process, known as \emph{privacy amplification via shuffling}, significantly enhances the effectiveness of differential privacy by blocking the statistician from tracing data back to individual users.}


\new{Analyzing the level of privacy amplification is a critical step in maintaining a balanced trade-off between privacy protection and data utility. This analysis, however, faces significant challenges due to the expansive message space involved and the complex, non-linear nature of the shuffling operations on $n$ messages. This complexity stands in contrast to simpler operations like summation in traditional DP approaches \cite{mironov2009computational,bonawitz2017practical}. }


\new{To address these challenges, there has been considerable research effort. A notable research by Erlingsson \textit{et al.}~\cite{erlingsson2019amplification} demonstrated when each of the $n$ shuffled messages is randomized using $\epsilon_0$-LDP, they collectively maintain differential privacy at a level of $O({\epsilon_0}\sqrt{{\log(1/\delta)}/{n}})$, with a high probability of $1-\delta$, given that $\epsilon_0$ is of order $O(1)$. For the more practical case of $\epsilon_0=\Theta(1)$, another research \cite{balle2019privacy} proposed viewing messages from other users as a ``privacy blanket'' that conceals each user's sensitive information. This approach showed that the shuffled messages preserve $(O(\min\{\epsilon_0,1\}e^{\epsilon_0}\sqrt{\log(1/\delta)/{n}}),\delta)$-DP. Recent research~\cite{feldman2022hiding,feldman2023stronger} introduced a ``clone'' concept to examine the indistinguishability offered by messages from other users, achieving an amplification bound of $(O(e^{\epsilon_0/2}\sqrt{\log(1/\delta)/{n}}),\delta)$-DP for $\epsilon_0=\Theta(1)$.}



\new{While these analyses are significant, they have limitations in tightness and generality. The ``privacy blanket'' framework \cite{balle2019privacy}, for example, provides tighter parameters for specific $\epsilon_0$-LDP randomizers. However, it is less precise in tail-bounding privacy loss variables, even when using their ${\Omega}(n)$-complexity numerical bounds. Similarly, the recent ``clone'' reduction \cite{feldman2023stronger} achieves nearly optimal amplification bounds for general $\epsilon$-LDP randomizers. Yet, it falls short in offering precise amplification for common LDP mechanisms, such as Hadamard response \cite{acharya2018hadamard}, subset selection \cite{wang2019local}, and PrivUnit \cite{bonawitz2017practical}. Consequently, these approaches do not fully assist practitioners in setting optimal privacy parameters to maximize data utility. Additionally, their results are only applicable to (single-message) shuffle protocols with LDP randomizers, which face intrinsic utility drawbacks \cite{ghazi2021power}. Such drawbacks are not as pronounced in protocols that utilize the metric DP notion \cite{chatzikokolakis2013broadening}, or multi-message protocols where each user contributes multiple messages \cite{balle2020private}.}

\new{In the context of multi-message shuffle model, current works leverage random numerical shares \cite{ghazi2020private,ghazi2021differentially} or categorical shares \cite{balle2020private,luo2022frequency,li2023privacy} to collaboratively preserve privacy post-shuffling. Nevertheless, these protocols, including those in \cite{cheu2022differentially,luo2022frequency,wang2023shuffle,li2023privacy}, analyze DP guarantees on a case-by-case basis, offering only approximate/loose analytic bounds. This hampers their generalizability across different protocols and leads to undesirable practical utility.}


\subsection{Our Contributions}
\new{Recognizing the aforementioned limitations, our research introduces two advancements in analyzing the privacy amplification of differentially private data processing.}

First, we designed a unified, tight, and scalable framework for analyzing privacy amplification within the shuffle model that supports simutaneously LDP randomizers, metric LDP randomizers, and various multi-message protocols. This framework, which we term the \emph{variation-ratio} reduction, innovatively links DP levels to two novel yet intuitive parametrizations: the pairwise total variation and the one-direction probability ratio of local randomizers. It demonstrates exact tightness for an extensive family of randomizers that exhibit extremal properties. This includes a variety of state-of-the-art $\epsilon_0$-LDP randomizers like local hash \cite{wang2017locally}, Hadamard response \cite{acharya2018hadamard}, the PrivUnit mechanism \cite{bhowmick2018protection}, and several multi-message protocols \cite{balcer2021connecting,cheu2022differentially, luo2022frequency,li2023privacy}. Moreover, our framework offers precise bounds with a computational complexity of \(\tilde{O}(n)\), making it highly effective for large-scale data analyses, such as \new{telemetry data} involving millions or billions of users.

Second, we derived tight privacy amplification bounds for parallel local randomizers (i.e., sampling-based composite randomizers). These types of randomizers are widely deployed by differentially private protocols in \new{user data collecting/queries}, such as count queries \cite{ren2018lopub,cormode2019answering,wang2019answering}, heavy hitter identification \cite{qin2016heavy,bassily2017practical,wang2019locally}, frequent itemset mining \cite{wang2018locally}, and machine learning \cite{jagannathan2009practical,fletcher2019decision}. \new{Within the variation-ratio framwork, our analysis reveals that the pairwise total variation bound of a parallel randomizer does not exceed the expected bound of its individual base randomizers. This insight forms the basis of our \emph{advanced parallel composition theorem} in the shuffle model, enabling much tighter amplification bounds compared to straightforward approaches.}



\subsection{Organization}
\new{The remaining paper is organized as follows. Section \ref{sec:related} reviews related works. Section \ref{sec:pre} provides background knowledge. Section \ref{sec:framework} presents main theoretical upper bounding results, the corresponding proof sketch, and an efficient numerical bounding algorithm. Section \ref{sec:lowerbounds} establishes amplification lower bounds. Section \ref{sec:parallel} improves bounding results for parallel randomizers. Section \ref{sec:exp} shows experimental results. Finally, Section \ref{sec:conclusion} concludes the paper.
}

\section{Related Work}\label{sec:related}

\textbf{Privacy amplification of LDP randomizers. } The seminal work \cite{erlingsson2019amplification} utilizes the privacy amplification via subsampling \cite{kasiviswanathan2011can,balle2018subsampling} to analyze the shuffling amplification, and shows $n$ shuffled messages satisfies $(\epsilon_0\sqrt{144\log(1/\delta)/n}, \delta)$-DP. Latterly, the privacy blanket \cite{balle2019privacy} proposes extracting an input-independent part from the output distribution, to work as a ``blanket'' to amplify privacy. For general LDP randomizers, the \cite{balle2019privacy} shows the input-independent part is sampled by each user with probability at least $e^{-\epsilon_0}$ (they term it as total variation similarity). For specific LDP randomizers (e.g., Laplace \cite{dwork2006calibrating} and generalized randomized response \cite{warner1965randomized}), the total variation similarity can be larger. Recently, the works \cite{feldman2022hiding} and \cite{feldman2023stronger} decompose output distributions into mixture distributions with $3$ options, and interprets messages from other users as clones of victim user. This clone reduction is near-optimal w.r.t. the dependence on $\epsilon_0$ for general LDP randomizers. Meanwhile, in order to derive tighter amplification bounds for specific randomizers, it needs to explicitly construct mixture distributions and determine the clone probabilities, which however is intractable for commonly deployed randomizers (e.g., Laplace mechanism, subset selection mechanism \cite{wang2019local,ye2018optimal}, PrivUnit mechanism \cite{bonawitz2017practical}, and sampling-based composite randomizers). As comparision, our framework using a new parametrization (i.e., pairwise total variation) of randomizers to implicitly derive tighter clone probabilities. Besides, the clone reductions in research~\cite{feldman2022hiding,feldman2023stronger} are limited to LDP randomizers, while our framework generalizes to metric-based LDP randomizers and multi-message protocols. Even for LDP randomizers in the shuffle model, our asymptotic amplification bounds are provably tighter than existing results (see Table \ref{tab:bounds}).




\textbf{Privacy amplification of metric-based LDP randomizers. } The metric-based LDP \cite{chatzikokolakis2013broadening} permits a flexible indistinguishable level between any pair of inputs $a,b\in \dom{X}$ (denoted as $d_\dom{X}(a,b)$), and is widely used for user data with a large domain range such as location data \cite{andres2013geo} and unstructed data \cite{zhao2022survey}. Previous work \cite{wang2023SMDP} initiated the study of metric DP in the shuffle model and provided several amplification bounds, which depend on $\max_{c\in \dom{X}} e^{d_\dom{X}(a,c)}+e^{d_\dom{X}(b,c)}$. In our work, we apply the proposed framework to metric LDP randomizers and improve the dependence to $(1-e^{-d_\dom{X}(a,b)})\cdot\max_{c\in \dom{X}} \max[e^{d_\dom{X}(a,c)},e^{d_\dom{X}(b,c)}]$ (a lower value indicates stronger amplification effects, see Section \ref{subsec:params} for detail).

 \textbf{Privacy amplification of multi-message randomizers. } For multi-message protocols in the shuffle model (e.g., \cite{cheu2022differentially,luo2022frequency,li2023privacy}), the local randomizer may neither satisfy LDP nor metric-based LDP. Hence, designated analyses of privacy guarantees are often conducted. Our proposed framework yields much tighter privacy guarantees and reduces the budget by $70\%$-$85\%$ compared to these designated analyses (see detail in Section \ref{sec:exp}).

 \textbf{Privacy amplification via shuffling under composition. } To analyze accumulated privacy loss under sequential composition for the shuffle model~\cite{girgis2021renyi,girgis2021shuffled,girgis2021subsampled,koskela2021tight,feldman2023stronger}, it is common to use analytic/numerical tools from the centralized model of differential privacy (e.g., strong composition
theorem \cite{dwork2008differential}, R\'enyi differential privacy \cite{mironov2017renyi}, and Fourier accountant \cite{koskela2020computing}). The \cite{girgis2021renyi} analyzes R\'enyi differential privacy of shuffled messages, \cite{girgis2021subsampled} further composes privacy for the shuffle model with subsamplings. To the best of our knowledge, we are the first to provide tight privacy amplification results for prevalent parallel composition (see detail in Section \ref{sec:parallel}).



\section{Preliminaries}\label{sec:pre}
\begin{definition}[Hockey-stick divergence]\label{def:hsd} The Hockey-stick divergence (with parameter $e^{\epsilon}$)  between two random variables $P$ and $Q$ is defined by:
\[D_{e^{\epsilon}}(P \| Q)=\int\max\{0,P(x)-e^\epsilon Q(x)\} \mathrm{d}x,\]
where we use the notation $P$ and $Q$ to refer to both the random variables and their probability density functions.
\end{definition}


Two variables $P$ and $Q$ are $(\epsilon, \delta)$-indistinguishable iff $$\max\{D_{e^\epsilon}(P\| Q), D_{e^\epsilon}(Q\| P)\}\leq \delta.$$ If two datasets have the same size and differ only by the data of a single individual, they are referred to as neighboring datasets. Differential privacy ensures that the divergence of query outputs on neighboring datasets is constrained (see Definition \ref{def:dp}). Similarly, in the local setting where a single individual's data is taken as input, we present the definition of local $(\epsilon,\delta)$-DP in Definition \ref{def:ldp}. \new{The notation $\epsilon$-LDP is used when the failure probability $\delta=0$}.


\begin{definition}[Differential privacy \cite{dwork2006calibrating}]\label{def:dp}
An algorithm $\mech{R}:\dom{X}^n\mapsto \dom{Z}$ satisfies $(\epsilon,\delta)$-differential privacy iff for all neighboring datasets $X, X'\in \dom{X}^n$, $\mech{R}(X)$ and $\mech{R}(X')$ are $(\epsilon,\delta)$-indistinguishable.
\end{definition}

\begin{definition}[Local differential privacy \cite{kasiviswanathan2011can}]\label{def:ldp}
An algorithm $\mech{R}:\dom{X}\mapsto \dom{Y}$ satisfies local $(\epsilon,\delta)$-differential privacy iff for all $x,x'\in \dom{X}$, $\mech{R}(x)$ and $\mech{R}(x')$ are $(\epsilon,\delta)$-indistinguishable.
\end{definition}

\textbf{Metric differential privacy. } A distance measure $d_{\dom{X}}: \dom{X}\times\dom{X}\mapsto \mathbb{R}$ over $\dom{X}$ characterizes the distances between any pair of elements $a,b\in \dom{X}$. For this measure to be deemed a metric, it must adhere to the subsequent properties:
\begin{enumerate}
\setlength{\itemsep}{0pt}
\setlength{\parsep}{0pt}
\setlength{\parskip}{0pt}
    \item[i.]\textit{Self-identity: } {For any $a\in \dom{X}$, $d_{\dom{X}}(a,a)=0$.}
    \item[ii.]\textit{Positivity: } {For any $a,b\in \dom{X}$ that $a\neq b$, $d_{\dom{X}}(a,b)>0$.}
    \item[iii.]\textit{Symmetry: } {For any $a,b\in \dom{X}$, $d_{\dom{X}}(a,b)=d_{\dom{X}}(b,a)$.}
    \item[iv.]\textit{Triangle inequality: } {For any $a,b,c\in \dom{X}$, $d_{\dom{X}}(a,b)+d_{\dom{X}}(b,c)\geq d_{\dom{X}}(a,c)$.}
\end{enumerate}

The concept of metric differential privacy expands upon the traditional differential privacy, enabling the calibration of indistinguishability constraints to element dissimilarity.

\begin{definition}[Metric $(d_{\dom{X}},\delta)$-differential privacy \cite{chatzikokolakis2013broadening}]\label{def:mcdp}
A randomized mechanism $\mech{R}$ satisfies $(d_{\dom{X}}, \delta)$-differential privacy iff for any neighboring datasets $X, X'\in \dom{X}^n$ (i.e., $x_i=a$ in $X$ and $x_i=b$ in $X'$), $\mech{R}(X)$ and $\mech{R}(X')$ are $(d_{\dom{X}}(a,b), \delta)$-indistinguishable.
\end{definition}

\new{In the local setting, each user sanitizes their data on their own locally. Substituting the dataset $S$ in Definition \ref{def:mcdp} with an individual data item, the definition of metric LDP is as follows:}
\begin{definition}[Local metric $d_{\dom{X}}$-differential privacy \cite{andres2013geo,alvim2018local}]\label{def:mldp}
Let $\dom{D}_{\mech{R}}$ denote the output domain, a randomized mechanism $\mech{R}$ satisfies local metric $d_{\dom{X}}$-privacy iff for any $a,b \in \mathcal{X}$, $\mech{R}(a)$ and $\mech{R}(b)$ are $(d_{\dom{X}}(a,b), 0)$-indistinguishable.
\end{definition}


\subsection{The Shuffle Model of Differential Privacy}
We use $[n]$ to denote $\{1,...,n\}$ and $[i:j]$ to denote $\{i,i+1,...,j\}$. Follow conventions in the shuffle model based on randomize-then-shuffle \cite{cheu2019distributed,balle2019privacy}, we define a single-message protocol $\mech{P}$ to be a list of algorithms $\mech{P} = (\{\mech{R}_i\}_{i\in [n]}, \mech{A})$, where $\mech{R}_i: \dom{X} \to \dom{Y}$ is the local randomizer of user $i$, and $\mech{A}: \dom{Y}^n \to \dom{Z}$ the analyzer in the data collector's side. 
The overall protocol implements a mechanism $\mech{P} : \dom{X}^n \to \dom{Z}$ as follows.
Each user $i$ holds a data record $x_i$ and a local randomizer $\mech{R}_{i}$, then computes a message $y_i = \mech{R}_i(x_i)$.
The messages $y_1,...,y_n$ are then shuffled and submitted to the analyzer. We write $\mech{S}(y_1,\ldots,y_n)$ to denote the random shuffling step, where $\mech{S} : \dom{Y}^n \to \dom{Y}^n$ is a \emph{shuffler} that applies a uniform-random permutation to its inputs.
In summary, the output of $\mech{P}(x_1, \ldots, x_n)$ is denoted by $\mech{A} \circ \mech{S} \circ \mech{R}_{[n]}(X) = \mech{A}(\mech{S}(\mech{R}_1(x_1), \ldots, \mech{R}_n(x_n)))$.

The shuffle model assumed that all parties involved in the protocol follow it faithfully and there is no collusion between them. From a privacy perspective, the goal is to ensure the differential privacy of the output $\mech{P}(x_1, \ldots, x_n)$ for any analyzer $\mech{A}$. By leveraging the post-processing property of the Hockey-stick divergence, it suffices to ensure that the shuffled messages $\mech{S} \circ \mech{R}_{[n]}(X)=\mech{S}(\mech{R}_1(x_1), \ldots, \mech{R}_n(x_n))$ are differentially private. We formally define differential privacy in the shuffle model in Definition \ref{def:sdp}.

\begin{definition}[Differential privacy in the shuffle model]\label{def:sdp}
A protocol $\mech{P} = (\{\mech{R}_i\}_{i\in [n]}, \mech{A})$ satisfies $(\epsilon,\delta)$-differential privacy in the shuffle model iff for all neighboring datasets $X, X'\in \dom{X}^n$, the $\mech{S} \circ \mech{R}_{[n]}(X)$ and $\mech{S} \circ \mech{R}_{[n]}(X')$ are $(\epsilon,\delta)$-indistinguishable.
\end{definition}

In multi-message shuffle protocols, which offer greater utility potential than single-message counterparts \cite{ghazi2021power}, each user can submit multiple messages. These messages from all users are then randomly shuffled. The differential privacy applied in this context is determined w.r.t. a single alteration in the input dataset, analogous to the approach in single-message protocols.

\section{The Variation-ratio Framework}\label{sec:framework}
In this section, we introduce our variation-ratio reduction framework for privacy amplification via shuffling. Our approach leverages the mixture decomposition of local randomizers, as introduced in the privacy blanket \cite{balle2019privacy} and clone reduction \cite{feldman2022hiding,feldman2023stronger} literature. \new{The mixture decomposition of local randomizers reveal two important facts: (1) randomized and anonymously shuffled messages from other users can mimic as the message from a particular (victim) user with some probability, which are termed as blanket messages in \cite{balle2019privacy} and clones in \cite{feldman2022hiding,feldman2023stronger} and can thus amplify privacy; (2) for the local randomizer of the victim user that pertain to some properties (e.g., $(\epsilon,0)$-LDP), the probability distribution of outputs given arbitrary inputs share some resemblance. The key to analyzing tighter shuffle privacy amplification would be characterizing more concise mimic probability from other users and the probability resemblance of the victim's outputs.} \new{Compared to existing results}, our framework achieves a highly precise mixture decomposition by utilizing both $(0,\delta)$-LDP and $(\epsilon_0, 0)$-LDP properties of local randomizers \new{to capture the the probability resemblance of victim's outputs (over differed inputs), and introduces the probability ratio as simple yet effective indicates to the mimic power from other users.}

To establish our approach, without loss of generality, we consider two neighboring datasets $X$ and $X'$ that differ only in the first user's data \new{(i.e. the first user is considered as the victim in DP)}, i.e., $X=\{x_1=x_1^0,x_2,\ldots,x_n\}$ and $X'=\{x_1=x_1^1,x_2,\ldots,x_n\}$, where $x_1^0,x_1^1,x_2,\ldots,x_n\in \dom{X}$. We define the following properties on (independent) local randomizers $\{\mech{R}_i\}_{i\in [n]}$ with some parameters $p>1$, $\beta\in [0, \frac{p-1}{p+1}]$, and $q\geq 1$:
\begin{itemize}
\item[I.] \textit{$(p, \beta)$-variation property}: we say that the $(p, \beta)$-variation property holds if $D_{p}(\mech{R}_1(x_1^0)\| 
\mech{R}_1(x_1^1))=0$ and $D_{e^0}(\mech{R}_1(x_1^0)\| \\
\mech{R}_1(x_1^1))\leq \beta$ for all possible $x_1^0,x_1^1\in \dom{X}$.
\item[II.] \textit{$q$-ratio property}: we say that the $q$-ratio property holds if $D_{q}(\mech{R}_1(x_1)\| \mech{R}_i(x_i))=0$ for all possible $x_1,\dots,x_n\in \dom{X}$ and all $\{\mech{R}_i\}_{i\in [2:n]}$.
\end{itemize}

The parameter $\beta$ corresponds to the pairwise total variation distance and serves as an indicator of the degree to which $\mech{R}_1$ satisfies $(0,\beta)$-LDP. Conversely, the parameter $p$ relates to the divergence ratio and serves as an indicator of the $(\log p,0)$-LDP property satisfied by $\mech{R}_1$. Meanwhile, the parameter $q$ captures the ability of $\mech{R}_i(x_i)$ to mimic $\mech{R}_1(x_1)$. For any randomizer $\mech{R}_1$, there exists a dominating pair of distributions \cite[Proposition 8]{zhu2022optimal}, such that the divergence upper bound (including the $(p,\beta)$ parameters) can be directly derived upon. Specificlly, in case the randomizer $\mech{R}_1$ satisfy $\epsilon_0$-LDP, we always have $p\leq \exp(\epsilon_0)$ and $\beta \leq \frac{e^{\epsilon_0}-1}{e^{\epsilon_0}+1}$ \cite{kairouz2014extremal}. Most commonly-used $\epsilon_0$-LDP randomizers have a lower total variation bound $\beta$ than the worst-case $\frac{e^{\epsilon_0}-1}{e^{\epsilon_0}+1}$ (see Tables \ref{tab:parametersldp}), thus permits a tighter mixture decomposition of their shuffled outputs.


\subsection{Main Results}
Our objective is to establish an upper bound on the divergence between the shuffled messages resulting from two independent protocol runs: $\mech{S}(\mech{R}_1(x_1^0),\ldots,\mech{R}_n(x_n))$ and $\mech{S}(\mech{R}_1(x_1^1),\ldots,\mech{R}_n(x_n))$, where the local randomizers $\{\mech{R}_i\}_{i\in [n]}$ satisfy the properties mentioned earlier. We accomplish this through an implicit analysis of the mixture decomposition of these local randomizers, \new{then} utilizing the data processing inequality to obtain a dominating pair of binomial counts. Subsequently, we represent the Hockey-stick divergence as an expectation of cumulative probabilities over a binomial variable $c$ (see Theorem \ref{the:dbound}), exploiting the monotone of the integral in Hockey-stick divergence. The notation $\mathsf{CDF}_{c,1/2}[c_1, c_2]$ represents the cumulative probability $\sum_{i\in [c_1,c_2]}{c\choose i}/2^{c}$, which can be quickly computed using two calls to the regularized incomplete beta function \cite{paris2010incomplete}. Therefore, Theorem \ref{the:dbound} directly implies an efficient algorithm for numerically compute Hockey-stick divergence within $\tilde{O}(n)$ complexity (see Section \ref{subsec:numericalupper} for detail).  

\begin{theorem}[Divergence upper bound]\label{the:dbound}
For $p> 1, \beta\in [0, \frac{p-1}{p+1}],\\q\geq 1$, if randomizers $\{\mech{R}_i\}_{i\in [n]}$ satisfy the $(p, \beta)$-variation property and the $q$-ratio property, then for any $x_1^0,x_1^1,x_2,...,x_n\in \dom{X}$:
\begin{alignat*}{2}
&D_{e^\epsilon}(\mech{S}(\mech{R}_1(x_1^0),\ldots,\mech{R}_n(x_n))\ \|\ \mech{S}(\mech{R}_1(x_1^1),\ldots,\mech{R}_n(x_n))) \\
\leq & \mathop{\mathbb{E}}\limits_{c\sim Binom(n-1,2r)}\Big[(p-e^{\epsilon})\alpha\cdot \mathop{\mathsf{CDF}}\limits_{c,1/2}[\lceil low_{c+1}-1\rceil, c] \\
&\ \ \ \ \ \ \ \ \ \ \ \ \ \ \ \ \ \ \ \ \ \ \ \ \ \  +(1-p e^\epsilon)\alpha\cdot\mathop{\mathsf{CDF}}\limits_{c,1/2}[\lceil low_{c+1}\rceil, c] \\
&\ \ \ \ \ \ \ \ \ \ \ \ \ \ \ \ \ \ \ \ \ \ \ \ \ \  +(1-e^\epsilon)(1-\alpha-p\alpha)\cdot \mathop{\mathsf{CDF}}\limits_{c,1/2}[\lceil low_c\rceil, c]\Big],
\end{alignat*}\normalsize
where $\alpha=\frac{\beta}{p-1}$, $r=\frac{\alpha p}{q}$, $low_c=\frac{(e^{\epsilon}p-1)\alpha c+(e^{\epsilon}-1)(1-\alpha-\alpha p)\cdot \frac{(n-c)r}{1-2r}}{\alpha(e^{\epsilon}+1)(p-1)}$.
\end{theorem}


Asymptotically, when $n$ is sufficiently large, we present two analytical amplification bounds: one that is sophisticated, and the other that is more concise but less tight. Specifically, we derive the former bound in Theorem \ref{the:tightlevel}. To achieve this, we tail bound the variable $c$ (i.e. number of clones in the former theorem), and exploit the fact that the statistical divergence monotonically non-increases with $c$. The resulting bound is satisfactorily tight for a wide range of variation-ratio parameters (see Section \ref{subsec:multimessage} for comparison with numerical upper bounds). To tail bound $c$, we use multiplicative Chernoff bounds and Hoeffding's inequality, which induce tighter bounds for small clone probability $2r$ (in local $d_{\dom{X}}$-DP randomizers) and for large $2r$ (in multi-message protocols), respectively.

\begin{theorem}[Analytic privacy amplification bounds]\label{the:tightlevel}

For $p > 1$, $\beta \in [0, \frac{p-1}{p+1}]$, and $q \geq 1$, the $P^{q}_{p,\beta}$ and $Q^{q}_{p,\beta}$ are $(\epsilon, \delta)$-indistinguishable with 
\small
\begin{align*}
\epsilon = \log\Bigg(1+\frac{\beta(2\sqrt{\Omega \log(\frac{4}{\delta})/2}+1)+\beta(\Omega/2-\sqrt{\Omega\log(\frac{4}{\delta})/2})}{\alpha\Omega} \\
+ \frac{(1-\alpha-\alpha p)(n-1-\Omega)r}{(1-2r)\alpha\Omega}\Bigg)
\end{align*}\normalsize
when $(p+1)\alpha/2-(1-\alpha-\alpha p)r/(1-2r) \geq 0$ and 
$\Omega = 2r(n-1)-\sqrt{\min(6r,\frac{1}{2})(n-1)\log(\frac{4}{\delta})} \geq \frac{2p(\beta+1+(\beta-1)p)(n-1)+\beta}{q+p(\beta-1+(\beta+1)p)-pq}$, where $\alpha = {\beta}/({p-1})$ and $r = {p\alpha}/{q}$.
\end{theorem}

We also derive a succinct formula in Theorem \ref{the:level} by further tail bounding the number of clones $c$. 

\begin{theorem}[Asymptotic privacy amplification bounds]\label{the:level}
For $p> 1, \beta\in [0, \frac{p-1}{p+1}],  q\geq 1$, when $n\geq \frac{8\log(2/\delta)(p-1)q}{\beta p}$, the $P^{q}_{p,\beta}$ and $Q^{q}_{p,\beta}$ are $(\epsilon,\delta)$-indistinguishable with 
\begin{align*}
\epsilon = \log\Bigg(1+ & \frac{\beta }{(1-v)(1+p)\beta/(p-1)+v}\Big(\sqrt{\frac{32\log(\frac{4}{\delta})}{r(n-1)}}+\frac{4}{rn}\Big)\Bigg),
\end{align*}
where $r=\frac{p\beta}{(p-1)q}$ and $v=\max\{0, \frac{4}{9}\frac{1-3r}{1-2r}\}$.
\end{theorem}

When $n$ is sufficiently large and a mild condition $r=\frac{p\beta}{(p-1)q}\leq 1/4$ holds, we have the $v$ in Theorem \ref{the:level} no less than a constant value $2/9$, thus the privacy can be amplified to: $$\tilde{O}\big(\sqrt{\beta(p-1)q/(p n)}\big).$$ Therefore, when $\beta$ and $q$ are relatively small, there are strong privacy amplification effects. This corresponds to qualitative interpretation about $\beta$ and $q$: the parameter $\beta$ indicates the total variation distance between $R_1(x_1^0)$ and $R_1(x^1_1)$, so when $\beta$ is relatively low, the statistical distance of shuffled messages would also be low; a lower $q$ indicates that messages from other users (i.e. users $[2:n]$) can better mimic $R_1(x^0_1)$ and $R_1(x^1_1)$, thus reduces the chance an adversary can disguise between $x^0_1$ and $x^1_1$ from shuffled messages.

Specifically, for randomizers satisfying $\epsilon_0$-LDP, we have  $q=p=\exp(\epsilon_0)$, thus the asymptotic amplified DP level becomes: $$\tilde{O}(\sqrt{\beta (\exp(\epsilon_0)-1)/n}).$$ In Table \ref{tab:bounds}, we compare our asymptotic bound with existing asymptotic bounds found in research~\cite{erlingsson2019amplification,balle2019privacy,feldman2022hiding,feldman2023stronger}, under the practical setting where $\epsilon_0=\Theta(1)$. Since the total variation bound $\beta$ is often significantly lower than the worst-case value of $\frac{e^{\epsilon_0}-1}{e^{\epsilon_0}+1}$, our amplification bound is provably tighter than previously known bounds.

\begin{table}[h]
\setlength{\tabcolsep}{0.1em}
\renewcommand{\arraystretch}{1.3}
\caption{Comparison of privacy amplification bound on shuffled $\epsilon_0$-LDP messages when $\epsilon_0=\Theta(1)$ and $n\geq \tilde{\Omega}(\exp(\epsilon_0))$.}
\vspace*{-0.7em}
\centering
\label{tab:bounds}
\begin{tabular}{|c|c|}
\hline
\multicolumn{1}{|c|}{\textbf{Method}}
& 
\multicolumn{1}{c|}{\textbf{Asymptotic Amplification Bound}} \\
\hline
EFMRTT19 \cite{erlingsson2019amplification}
& 
${O}(\sqrt{\exp(3\epsilon_0)\log(1/\delta)/n})$ \\ 
\hline
privacy blanket \cite{balle2019privacy}
&
${O}(\sqrt{\exp(2\epsilon_0)\log(1/\delta)/n})$
\\ 
\hline
clone \cite{feldman2022hiding}
&
$O\Big(\frac{\exp(\epsilon_0)-1}{\exp(\epsilon_0)+1}\sqrt{\frac{\exp(\epsilon_0)\log(1/\delta)}{n}}\Big)$ \\
\hline
stronger clone \cite{feldman2023stronger}
& 
$O\Big(\sqrt{\frac{(\exp(\epsilon_0)-1)^2\log(1/\delta)}{n (\exp(\epsilon_0)+1)}}\Big)$
\\
\hline
variation ratio [\textbf{this work}] 
& 
$O\Big(\sqrt{\frac{\beta(\exp(\epsilon_0)-1)\log(1/\delta)}{n}}\Big)$
\\
\hline
\end{tabular}
\end{table}


\textbf{Discussion on differences with  stronger clone reduction. } For $\epsilon_0$-LDP randomizers where $p\equiv q=\exp(\epsilon_0)$, our intermediate result presented in Theorem \ref{the:reduction} yields a divergence upper bound analogous to the one found in the state-of-the-art stronger clone by \cite{feldman2023stronger}. Our framework offers three primary advantages:
(i) Our bound is tighter than \cite{feldman2023stronger} as long as the newly introduced total variation parameter $\beta$ is not the worst-case $(p-1)/(p+1)$. Notably, commonly-used randomizers exhibit much lower $\beta$ values, and computing $\beta$ is straightforward.
(ii) While \cite{feldman2023stronger} handles only LDP randomizers that $q\equiv p$, our reduction is versatile, accommodating also cases where $p<q$ (e.g., for metric local randomizers) and $p>q$ (e.g., for multi-message randomizers).
(iii) Our method entails a $\tilde{O}(n)$-time algorithm for numerically determining the Hockey-stick divergence, offering a much more efficient solution compared to prevailing algorithms \cite{feldman2022hiding,koskela2021tight}, being about $10$ times faster than their documented execution times.

\subsection{Proof Sketch}
Our proof for Theorem \ref{the:dbound} follows three steps: (i) utilizing the $(p,\beta)$-variation and $q$-ratio properties to obtain mixture decomposition of local randomizers; (ii)    
motivated by the clone reduction \cite{feldman2022hiding,feldman2023stronger}, we succinctly reduce the mixture decomposition as several binomial counts, which serve as the dominating pair of distribution for privacy analyses \cite{zhu2022optimal}; (iii) delving into the Hockey-stick divergence of the dominating binomial counts, we utilize the monotonicness of the Hockey-stick divergence to derive the final results of expectation on cumulative probabilities.

We first consider fixed $x_1^0$ and $x_1^1$, and show that $\mech{R}_1(x_1^0), \mech{R}_1(x_1^1),$ and $\mech{R}_i(x_i)$ can be interpreted as mixture distributions (as elaborated in Lemma \ref{lemma:mixture}) with parameters related to $p$, $q$, and the total variation $\beta'=D_{1}(\mech{R}_1(x_1^0)\| \mech{R}_1(x_1^1))$. The core idea is to concentrate on elements where the distributions $\mech{R}_1(x_1^0)$ and $\mech{R}_1(x_1^1)$ differ, and bound the total differing probability using parameters $p$ and $\beta$. For messages from other users, the parameter $q$ is used to compute their likelihood of appearing as a message from user $1$.

\begin{lemma}[Mixture decompositions]\label{lemma:mixture}
Given $x_1^0,x_1^1,...,x_n\in \dom{X}$, 
if algorithms $\{\mech{R}_i\}_{i\in [n]}$ satisfy the $(p, \beta')$-variation property and the $q$-ratio property with some $p> 1, q\geq 1$ and $\beta'=D_{1}(\mech{R}_1(x_1^0)\| \mech{R}_1(x_1^1))$, then there exists distributions $\mathcal{Q}_1^0, \mathcal{Q}_1^1, \mathcal{Q}_1, \mathcal{Q}_2, ..., \mathcal{Q}_n$ such that
\begin{align}
\label{eq:mix1}
&\mech{R}_1(x_1^0) = p \alpha \mathcal{Q}_1^0 + \alpha \mathcal{Q}_1^1+(1-\alpha-p \alpha)\mathcal{Q}_1 \\
\label{eq:mix2}
&\mech{R}_1(x_1^1) = \alpha \mathcal{Q}_1^0 + p \alpha  \mathcal{Q}_1^1+(1-\alpha-p\alpha)\mathcal{Q}_1 \\
\label{eq:mix3}
&\forall i\in [2,n],\  \mech{R}_i(x_i) = r \mathcal{Q}_1^0 + r \mathcal{Q}_1^1+(1-2r)\mathcal{Q}_i
\end{align}
where $\alpha=\frac{\beta'}{p-1}$ and $r=\frac{\alpha p}{q}$.
\end{lemma}

Next, we generalize the clone reduction technique \cite{feldman2023stronger} by leveraging the data processing inequality to relate mixture parameters to the level of divergence (as demonstrated in Lemma \ref{lemma:clone}). Specifically, this generalized reduction lemma implies that the statistical distance between shuffled messages $$D(\mech{S}(\mech{R}_1(x_1^0),\ldots,\mech{R}_n(x_n))\ \|\ \mech{S}(\mech{R}_1(x_1^1),\ldots,\mech{R}_n(x_n)))$$ is bounded by the one between binomial counts: $D(P^{q}_{p,\beta'}\| Q^{q}_{p,\beta'})$.

\begin{lemma}[Divergence reduction via mixture]\label{lemma:clone}
Given any $n+1$ inputs $x_1^0, x_1^1, x_2,..., x_n\in \dom{X}$, consider algorithms $\{\mech{R}_i\}_{i\in [n]}$ such that the output domain is finite and
\begin{align*}
& \mech{R}_1(x_1^0)=p \alpha \mathcal{Q}_1^0 + \alpha \mathcal{Q}_1^1+(1-\alpha-p \alpha)\mathcal{Q}_1, \\
& \mech{R}_1(x_1^1)=\alpha \mathcal{Q}_1^0 + p \alpha  \mathcal{Q}_1^1+(1-\alpha-p\alpha)\mathcal{Q}_1, \\
& \forall i\in [2,n],\  \mech{R}_i(x_i)=r \mathcal{Q}_1^0 + r \mathcal{Q}_1^1+(1-2r)\mathcal{Q}_i
\end{align*}
\new{holds for some $p\geq 1, q>1, \alpha=\beta'/(p-1)\in [0,1/(p+1)], r\in [0,1/2]$} and some probability distributions $\mathcal{Q}_1^0,\mathcal{Q}_1^1,\mathcal{Q}_1,\mathcal{Q}_2,...,\mathcal{Q}_n$. Let $C\sim Binom(n-1, 2r)$, $A\sim Binom(C, 1/2)$, and $\Delta_1=Bernoulli(p \alpha)$ and $\Delta_2=Bernoulli(1-\Delta_1, \alpha/(1-p \alpha))$; let $P^{q}_{p,\beta'}=(A+\Delta_1, C-A+\Delta_2)$ and $Q^{q}_{p,\beta'}=(A+\Delta_2, C-A+\Delta_1)$. Then for any distance measure $D$ that satisfies the data processing inequality,
\begin{alignat*}{2}
&D(\mech{S}(\mech{R}_1(x_1^0),..,\mech{R}_n(x_n)) \|\mech{S}(\mech{R}_1(x_1^1),..,\mech{R}_n(x_n)))
\leq\ &D(P^{q}_{p,\beta'}\| Q^{q}_{p,\beta'}).
\end{alignat*}
\end{lemma}

Then, \new{we utilize the fact that the statistical distance $D(P^{q}_{p,\beta'}\| Q^{q}_{p,\beta'})$ is monotonically non-decreasing with $\beta'$}. For any arbitrary $x_1^0$ and $x_1^1$, the total variation is upper bounded by $\beta$. Thus, we conclude that for all possible neighboring datasets, the statistical distance between the shuffled messages is upper bounded by $D(P^{q}_{p,\beta}\| Q^{q}_{p,\beta})$.

\begin{table*}[ht]
\caption{A summary of variation-ratio parameters of $\epsilon$-LDP randomizers. A lower $\beta$ means a stronger amplification effect.}
\label{tab:parametersldp}
\vspace*{-0.7em}
\centering
\begin{tabular}{|c|c|c|c|}
\hline
\bfseries randomizer & \bfseries parameter $p$ & \bfseries parameter $\beta$ & \bfseries parameter $q$ \\  
\hline
general mechanisms & $e^{\epsilon}$ & $\frac{e^{\epsilon}-1}{e^{\epsilon}+1}$ & $e^{\epsilon}$ \\
\hline

Laplace mechanism for $[0,1]$ \cite{dwork2006calibrating} & $e^{\epsilon}$ & $1-e^{-\epsilon/2}$ & $e^{\epsilon}$ \\

PrivUnit mechanism with cap area $c$ \cite{bonawitz2017practical} & $e^{\epsilon}$ & $\frac{c\cdot (e^{\epsilon}-1)}{c\cdot e^{\epsilon}+1-c}$ & $e^{\epsilon}$ \\


\hline

general randomized response (GRR) on $d$ options \cite{kairouz2016discrete} & $e^{\epsilon}$ & $\frac{e^{\epsilon}-1}{e^{\epsilon}+d-1}$ & $e^{\epsilon}$ \\

binary RR on $d$ options \cite{duchi2013local} & $e^{\epsilon}$ & $\frac{e^{\epsilon/2}-1}{e^{\epsilon/2}+1}$ & $e^{\epsilon}$ \\

$k$-subset on $d$ options \cite{wang2019local,ye2018optimal} & $e^{\epsilon}$ & $\frac{(e^{\epsilon}-1)({d-1\choose k-1}-{d-2\choose k-2})}{e^{\epsilon}{d-1\choose k-1}+{d-1\choose k}}$ & $e^{\epsilon}$ \\

local hash with length $l$ \cite{wang2017locally} & $e^{\epsilon}$ & $\frac{e^{\epsilon}-1}{e^{\epsilon}+l-1}$ & $e^{\epsilon}$ \\

Hadamard response $(K,s,B=1)$ \cite{acharya2018hadamard} & $e^{\epsilon}$ & $\frac{s(e^{\epsilon}-1)/2}{s e^{\epsilon}+K-s}$ & $e^{\epsilon}$ \\

Hadamard response $(K,s,B>1)$ \cite{acharya2018hadamard} & $e^{\epsilon}$ & $\frac{s(e^{\epsilon}-1)}{s e^{\epsilon}+K-s}$ & $e^{\epsilon}$ \\

\hline

sampling RAPPOR on $s$ in $d$ options \cite{qin2016heavy} & $e^{\epsilon}$ & $\frac{s(e^{\epsilon/2}-1)}{d(e^{\epsilon/2}+1)}$ & $e^{\epsilon}$ \\


Wheel on $s$ in $d$ options with length $p$ \cite{wang2020set} & $e^{\epsilon}$ & $\frac{s p(e^{\epsilon}-1)}{ s p e^{\epsilon}+(1-s p)}$ & $e^{\epsilon}$\\
\hline
\end{tabular}

\end{table*}

\begin{table*}
\setlength{\tabcolsep}{0.30em}
\renewcommand{\arraystretch}{1.2}
\caption{A summary of amplification parameters of $\mech{S}(\mech{R}(x_1^0),..., \mech{R}(x_n))$ and $\mech{S}(\mech{R}(x_1^1),..., \mech{R}(x_n)))$ for local $d_\dom{X}$-DP randomizers, $d_{01}=d_\dom{X}(x_1^0,x_1^1)$ and $d_{max}=\max_{x\in \dom{X}} \max\{d_\dom{X}(x_1^0,x), d_\dom{X}(x_1^1,x)\}$.}
\vspace*{-0.7em}
\label{tab:parameterslmdp}
\centering
\begin{tabular}{|c|c|c|c|}
\hline
\bfseries randomizer & \bfseries parameter $p$ & \bfseries parameter $\beta$ & \bfseries parameter $q$ \\  
\hline
general mechanisms & $e^{d_{01}}$ & $\frac{e^{d_{01}}-1}{e^{d_{01}}+1}$ & $e^{d_{max}}$ \\
\hline
Laplace mechanism \cite{alvim2018local}, $\ell_1$-norm on $\mathbb{R}$  & $e^{d_{01}}$ & $1-e^{-d_{01}/2}$ & $e^{d_{max}}$ \\

planar Laplace \cite{andres2013geo}, $\ell_2$-norm on $\mathbb{R}^2$  & $e^{d_{01}}$ & $2\int_{0}^{\frac{d_{01}}{2}}\int_{-\infty}^{+\infty}\frac{e^{-\sqrt{(x-\frac{d_{01}}{2})^2+y^2}}}{2\pi}\mathrm{d} y \mathrm{d}x$ & $e^{d_{max}}$ \\



\hline
\end{tabular}
\end{table*}

\begin{lemma}[Non-decreasing of divergence]\label{lemma:nondecrease}
For any $p> 1, q\geq 1$ and $\beta,\beta'\in [0, \frac{p-1}{p+1}]$, if $\beta>\beta'$, then for any distance measure $D$ that satisfies data processing inequality,
$$D(P^{q}_{p,\beta}\| Q^{q}_{p,\beta})\geq D(P^{q}_{p,\beta'}\| Q^{q}_{p,\beta'}).$$
\end{lemma}

Use the fact that $\beta'$ is upper bounded by $\beta$, the results of previous steps can be summarized in in Theorem \ref{the:reduction}. \new{It is applicable to any divergence measures that satisfy the data-processing inequality, such as R\'enyi divergences \cite{mironov2017renyi}}.

\begin{theorem}[Variation-ratio reduction]\label{the:reduction}
For $p> 1, \beta\in [0, \frac{p-1}{p+1}],  q\geq 1$, let $C\sim Binom(n-1, \frac{2\beta p}{(p-1)q})$, $A\sim Binom(C, 1/2)$ and $\Delta_1\sim Bernoulli(\frac{\beta p}{p-1})$ and $\Delta_2\sim Bernoulli(1-\Delta_1, \frac{\beta}{p-1-\beta p})$; let $P^{q}_{p,\beta}$ denote $(A+\Delta_1, C-A+\Delta_2)$ and $Q^{q}_{p,\beta}$ denote $(A+\Delta_2, C-A+\Delta_1)$.
If randomizers $\{\mech{R}_i\}_{i\in [n]}$ satisfy the $(p, \beta)$-variation property and the $q$-ratio property, then for any $x_1^0,x_1^1,x_2,...,x_n\in \dom{X}$ and any measurement $D$ satisfying the data-processing inequality:
\small
\begin{alignat*}{2}
&D(\mech{S}(\mech{R}_1(x_1^0),\ldots,\mech{R}_n(x_n))\|\mech{S}(\mech{R}_1(x_1^1),\ldots,\mech{R}_n(x_n))) 
\leq &D(P^{q}_{p,\beta}\|Q^{q}_{p,\beta}).
\end{alignat*}
\normalsize
\end{theorem}

\begin{theorem}[Divergence bound as an expectation]\label{the:hsdnumerical}
For $p> 1, \beta\in [0, \frac{p-1}{p+1}],  q\geq 1$, let $\alpha=\beta/(p-1)$ and $r=\alpha p/q$, then for any $\epsilon \in \mathbb{R}$, the following two equations hold:
\small\begin{alignat*}{2}
D_{e^\epsilon}(P_{p,\beta}^{q} \| Q_{p,\beta}^{q}) =&\mathop{\mathbb{E}}\limits_{c\sim Binom(n-1,2r)}
\Big[(p-e^{\epsilon})\alpha\cdot\mathop{\mathsf{CDF}}\limits_{c,1/2}[\lceil low_{c+1}-1\rceil, c] \\
&+(1-p e^\epsilon)\alpha\cdot\mathop{\mathsf{CDF}}\limits_{c,1/2}[\lceil low_{c+1}\rceil, c] \\
&+(1-e^\epsilon)(1-\alpha-p\alpha)\cdot \mathop{\mathsf{CDF}}\limits_{c,1/2}[\lceil low_c\rceil, c]\Big],
\end{alignat*}\normalsize
where $low_c=\frac{(e^{\epsilon'}p-1)\alpha c+(e^{\epsilon'}-1)(1-\alpha-\alpha p)(n-c)\cdot \frac{r}{1-2r}}{\alpha(e^{\epsilon'}+1)(p-1)}$.
\end{theorem}

Finally, we show an efficient way to  computing (the upper bound of) $D_{e^{\epsilon'}}(P^{q}_{p,\beta}\| Q^{q}_{p,\beta})$. A straight-forward approach is enumerating the output space $(a,b)\in [0,n]^2$ using the definition of the divergence. However, this method has a high computational cost of $O(n^2)$ and may suffer from numerical underflow issues. To avoid these issues, we exploit the monotonicity of the probability ratio $\frac{\mathbb{P}[P^{q}_{p,\beta}=(a,b)]}{\mathbb{P}[Q^{q}_{p,\beta}=(a,b)]}$ with respect to $a$ when $a+b$ is fixed. Thus, the maximum operation in the Hockey-stick divergence can be safely removed by carefully tracking the range of $a$, and the overall divergence can be expressed as an expectation over $c\sim Binom(n-1,2r)$ that follows binomial distribution (see Theorem \ref{the:hsdnumerical}). To track the range of $a$, we introduce the notation $\mathsf{CDF}_{c,1/2}[c_1, c_2]$ to represent the cumulative probability $\sum_{i\in [c_1,c_2]}{c\choose i}/2^{c}$, which can be computed using two calls to the regularized incomplete beta function \cite{paris2010incomplete}. As a result, the computational complexity is reduced to $\tilde{O}(n)$.

\subsection{Amplification Parameters of Randomizers}\label{subsec:params}
To demonstrate the broad applicability of our proposed framework for privacy amplification analysis in the shuffle model, we provide amplification parameters for a wide range of local randomizers, including prevalent $\epsilon_0$-LDP randomizers, local metric $d_\dom{X}$-DP randomizers, and multi-message randomizers.

\textbf{LDP randomizers. } We summarize the variation-ratio parameters of $\epsilon$-LDP randomizers in Table \ref{tab:parametersldp}. The worst-case total variation bound $\beta = \frac{e^\epsilon - 1}{e^\epsilon + 1}$ for general mechanisms is derived from \cite{kairouz2014extremal}, which proves that randomized response maximizes it. This worst-case bound is equivalent to the asymptotically optimal stronger clone reduction in research~\cite{feldman2023stronger}. For commonly used mechanisms, such as mean estimation mechanisms \cite{dwork2006calibrating, duchi2013local,nguyen2016collecting, wang2019collecting}, distribution estimation mechanisms \cite{kairouz2016discrete, duchi2013local,erlingsson2014rappor, wang2019local,ye2018optimal, wang2017locally,acharya2018hadamard}, and mechanisms for complicated data \cite{qin2016heavy,wang2018privset,gu2019pckv,wang2021hiding}, we exploit their specific structures to derive tighter total variation bounds, which can lead to stronger amplification effects (see Lemma \ref{lemma:nondecrease} for qualitative analyses and Section \ref{sec:exp} for numerical results). 


\begin{table*}
\caption{A summary of amplification parameters of some multi-message shuffle private protocols.}
\vspace*{-0.7em}
\label{tab:parametersmulti}
\centering
\begin{tabular}{|c|c|c|c|}
\hline
\bfseries randomizer & \bfseries parameter $p$ & \bfseries parameter $\beta$ & \bfseries parameter $q$ \\  
\hline

Balcer \textit{et al.} \cite{balcer2020separating} with coin prob. $p$ for binary summation & $+\infty$ & $1$ & $\max\{\frac{1}{p},\frac{1}{1-p}\}$ \\

Balcer \textit{et al.} \cite{balcer2021connecting} with uniform coin for binary summation & $+\infty$ & $1$ & $2$ \\

\hline
Cheu \textit{et al.} \cite{cheu2022differentially} with flip prob. $f\in [0,0.5]$ on $\{0,1\}^d$ & $\frac{(1-f)^2}{f^2}$ & $1-2f$ & $\frac{1-f}{f}$ \\

Balls-into-bins \cite{luo2022frequency} with $d$ bins (and $s$ special bins) & $+\infty$ & $1$ & $\frac{d}{s}$ \\

mixDUMP \cite{li2023privacy} with flip prob. $f\in [0,\frac{d-1}{d}]$ on $d$ bins & $\frac{(1-f)(d-1)}{f}$ & $\frac{(1-f)(d-1)-f}{d-1}$ & $(1-f)d$ \\

\hline
\end{tabular}
\end{table*}

\textbf{Local metric DP radomizers. } Regarding local metric DP mechanisms, we study the indistinguishability level between shuffled messages generated by a metric locally $d_\dom{X}$-private mechanism $\mech{R}$, where $D_{\exp(d_\dom{X}(x,x'))}(\mech{R}(x)\|\mech{R}(x'))=0$ holds for all $x,x'\in \dom{X}$. Specifically, we aim to analyze the indistinguishable level between $\mech{S}(\mech{R}(x_1^0),..., \mech{R}(x_n))$ and $\mech{S}(\mech{R}(x_1^1),..., \mech{R}(x_n)))$, which is also captured by Theorem \ref{the:reduction} with $p\leq \exp(d_{01})$ and $q\leq \exp(d_{max})$. Here, $d_{01}$ denotes the local indistinguishable level $d_\dom{X}(x_1^0, x_1^1)$, and $d_{max}$ denotes the maximum indistinguishable level
$$\max_{x\in \dom{X}} \max\{d_\dom{X}(x, x_1^0), d_\dom{X}(x, x_1^1)\}$$
w.r.t. $x_1^0$ and $x_1^1$. We summarize the variation-ratio parameters for $\mech{S}(\mech{R}(x_1^0),..., \mech{R}(x_n))$ and $\mech{S}(\mech{R}(x_1^1),..., \mech{R}(x_n)))$ in Table \ref{tab:parameterslmdp}. For general local $d_{\dom{X}}$-DP mechanisms, the amplification parameters $p=e^{d_{01}}, \beta=\frac{e^{d_{01}}-1}{e^{d_{01}}+1}$ come directly from $d_{01}$-LDP properties of $\mech{R}$ on $x_1\in \{x_1^0,x_1^1\}$, the parameter $q=e^{d_{max}}$ comes from the metric privacy constraints that $\mech{R}(x_i)$ and $\mech{R}(x_1)$ are $(d_{\dom{X}}(x_i,x_1))$-indistinguishable and $d_{\dom{X}}(x_i,x_1)\leq d_{max}$. Compared with the amplification upper bound for general mechanisms proved by \cite{wang2023SMDP}, where the clone probability $2r=2/(\max_{x\in \dom{X}} e^{d_\dom{X}(x, x_1^0)}+ e^{d_\dom{X}(x, x_1^1)})$, our clone probability $2\frac{\beta p}{(p-1)q}=2/(e^{d_{max}}+e^{d_{max}-d_{01}})$ is at least not smaller than theirs (due to the triangle inequality property of $d_\dom{X}$). For Laplace mechanism with $\ell_1$-norm metric privacy on one-dimensional numerical values, the total variation bound $\beta$ is: $$D_1(Laplace(0,1)\| Laplace(d_{01},1))=1-e^{-d_{01}/2}.$$ For the location randomization mechanism with $\ell_2$-norm metric privacy over two-dimensional domains, such as the planar Laplace mechanism \cite{andres2013geo}, the probability density $\mathbb{P}[PlanarLaplace(u,1)=x]=\frac{e^{-\|x-u\|_2}}{2\pi}$ for $u,x\in \mathbb{R}^2$, and the total variation bound $\beta$ is: 
\begin{alignat*}{2}
\small
2\int_{0}^{{d_{01}}/{2}}\int_{-\infty}^{+\infty}(e^{-\sqrt{(x-{d_{01}}/{2})^2+y^2}})/(2\pi)\mathrm{d} y \mathrm{d}x.
\normalsize
\end{alignat*}


\textbf{Multi-message protocols. } For multi-message protocols in the literature, we revisit their actual variation-ratio parameters in Table \ref{tab:parametersmulti}. Typically, in these protocols, messages sent by a user can be classified as input-dependent or input-independent, with only one of them being input-dependent. For example, for the user $1$ among $n'$ users in \cite{cheu2022differentially}, $\mech{R}_1(x_1)$ is produced by binary randomized response on $x_i$, while other $m-1$ messages (i.e. $\mech{R}_{n'+1}(*),\ldots, \mech{R}_{(m-1)\cdot n'+1}(*)$) are independently produced by binary randomized response on a zero vector $\{0\}^d$ with a flip probability of $f \in [0,0.5]$. Similarly, in protocols for binary summation ($x_i\in \{0,1\}$), the input-independent/blanket message can be a Bernoulli variable with biased coin \cite{balcer2020separating} or uniform coin \cite{balcer2021connecting}. In balls-into-bins \cite{luo2022frequency}, pureDUMP, and mixDUMP \cite{li2023privacy}, each blanket message is a uniform-random category in $[d]$. The parameters $p$ and $\beta$ are computed from the input-dependent message $\mech{R}_1(x_1)$, while the parameter $q$ is derived from input-independent messages contributed by both user $1$ and other users, referred to as blanket messages in \cite{balle2019privacy}. It is important to emphasize that, within these multi-message shuffle protocols, input-independent blanket messages can be viewed as if they were generated by dummy users. Consequently, the number of messages in which an input-dependent message can blend is equivalent to the total count of input-independent messages. For instance, assuming $n'$ users with each user contributing $m-1$ input-independent messages, the term $n-1$ in Theorem \ref{the:reduction} effectively becomes $n'\cdot(m-1)$. Compared to the privacy guarantees in the original works, our amplification bound for these protocols offers a more than $70\%$ budget savings, as demonstrated by the numerical results in Section \ref{subsec:multimessage}.


\textbf{Discussion on the  generality of our framework for multi-message protocols. } \new{We note that our framework applies to all single-message shuffle protocols using local $\epsilon$-DP or metric DP randomizers. While for
multi-message protocols in the shuffle model, based on the correlation between outputted messages, they can be primarily categorized into three main types}:
\begin{itemize}
    \item[I.] The first type of multi-message shuffle protocols is achieved by invoking multiple single-message protocols. This is seen as in the Recursive protocol in \cite{balle2020private} and the utility-complexity balanced protocol in \cite{girgis2023distributed}.
    \item[II.] The second type involves sending one input-dependent message and multiple independent blanket messages, as seen in \cite{balcer2020separating,balcer2021connecting,cheu2022differentially,luo2022frequency,li2023privacy}.
    \item[III.] \new{The third type involves sending multiple \emph{correlated} messages per user}. Protocols falling in this category include \cite{ghazi2020private,ghazi2021differentially} and the secret-sharing-style IKOS protocol in \cite{balle2020private}.
\end{itemize}
We also envision that there can be multi-message protocols that combines Types I and II, by invoking several protocols each is a multi-message one. Our variation-ratio framework is compatible with protocols categorized under Type I and/or Type II. However, Type III protocols, wherein messages from a single user are correlated, breach the independence assumption of each randomizer $\mech{R}_i$ within our framework.

\textbf{Discussion on variation-ratio parameters.} 
The research on parameters \(p\), \(\beta\), and \(q\) shows patterns that help design protocols for the shuffle model.
Single-message shuffle protocols that utilize $\epsilon_0$-LDP randomizers are limited to $q=p=e^{\epsilon_0}$. Therefore, the clone probability $2r=2\beta p/((p-1)q)$ decreases significantly with $\epsilon_0$, and larger population are required to achieve global $(\epsilon,\delta)$-DP. Randomization mechanisms with better utility in the local model often fully exploit privacy constraints and have larger total variation bounds $\beta$, resulting in weaker privacy amplification effects (e.g., comparing randomized response with binary RR on $2$ options). Multi-message shuffle protocols typically set $q<p$ or even $q\ll p$ to increase the clone probability and reduce the number of messages transmitted per user. For local metric $d_{\dom{X}}$-DP randomizers that handle large data domains, $d_{max}$ is often set relatively high, and thus $q\geq p$. When a global privacy target is given, they need to strike a balance between local data utility (which increases with $d_{max}$) and privacy amplification effects (which decrease with $d_{max}$).


\subsection{Numerical Method for Upper Bounds}\label{subsec:numericalupper}
This part is dedicated to numerically computing the indistinguishable level between two shuffled message sets $\mech{S}(\mech{R}_1(x_1^0),\ldots,\mech{R}_n(x_n))$ and $\mech{S}(\mech{R}_1(x_1^1),\ldots,\mech{R}_n(x_n))$, within a given privacy failure probability $\delta \in [0,1]$. According to Theorem \ref{the:reduction}, the indistinguishable level is upper bounded by the one between binomial counts $P^{q}_{p,\beta}$ and $Q^{q}_{p,\beta}$. Our aim is then to find the smallest $\epsilon\in [0,\log p]$ that the corresponding Hockey-stick divergence is no more than the privacy failure parameter:
\begin{align}\label{pb:epsilon}
\arg\min_{\epsilon\in [0,\log p]} \max\Big[D_{e^\epsilon}(P^{q}_{p,\beta} \| Q^{q}_{p,\beta}), D_{e^\epsilon}(Q^{q}_{p,\beta} \| P^{q}_{p,\beta})\Big]\leq \delta.
\end{align}
Solving this problem directly is infeasible. Fortunately, the divergence monotonically non-increases with $\epsilon$, so one can binary search on $\epsilon$. Plugging in the results from Theorems \ref{the:dbound} and \ref{the:hsdnumerical}, we compute the Hockey-stick divergence given the temporary level  $\epsilon$ as an expectation of cumulative probabilities. The complete implementation for solving problem (\ref{pb:epsilon}) is provided Algorithm \ref{alg:upperbound}.




\begin{algorithm}[t]
\small
    \caption{Efficient search of  indistinguishable upper bound of $P^{q}_{p,\beta}$ and $Q^{q}_{p,\beta}$}
    \label{alg:upperbound}
    \KwIn{parameter $\delta\in [0,1]$, number of users $n$, parameters $p>1, q\geq 1, \beta\in [0, \frac{p-1}{p+1}]$, number of iterations $T$.}
    \KwOut{An upper bound of $\epsilon'_c$ that  $\max\big[D_{\epsilon'_c}(P^{q}_{p,\beta}\| Q^{q}_{p,\beta}), D_{\epsilon'_c}(Q^{q}_{p,\beta}\| P^{q}_{p,\beta})\big]\leq \delta$ holds.}

    {$\alpha=\frac{\beta}{p-1}$,\ \ $r=\frac{\alpha p}{q}$}

    \SetKwProg{myproc}{Procedure}{}{}
    \SetKwFunction{proc}{{Delta}}
    \myproc{\proc{$\epsilon'$}}{
        {$\delta' \leftarrow 0$}

        {$w_c=\binom{n-1}{c}(2r)^c(1-2r)^{n-1-c}$}

        {$low_c = \frac{(e^{\epsilon'}p-1)\alpha c+(e^{\epsilon'}-1)f}{\alpha(e^{\epsilon'}+1)(p-1)}$}

        
        \For{$c \in [0,n]$}{
            
            {$\delta' \leftarrow \delta'+w_c((p-e^{\epsilon})\alpha\cdot \mathop{\mathsf{CDF}}\limits_{c,1/2}[\lceil low_{c+1}-1\rceil, c]+(1-p e^\epsilon)\alpha\cdot\mathop{\mathsf{CDF}}\limits_{c,1/2}[\lceil low_{c+1}\rceil, c]+(1-e^\epsilon)(1-\alpha-p\alpha)\cdot \mathop{\mathsf{CDF}}\limits_{c,1/2}[\lceil low_c\rceil, c])$}
        }
        \KwRet{$\delta'$}
    }
    
    $\epsilon_{L}\leftarrow 0$,\ \ \ \  $\epsilon_{H}\leftarrow \log(p)$
    
    \For{$t \in [T]$}{
    {$\epsilon_t\leftarrow \frac{\epsilon_{L}+\epsilon_{H}}{2}$}
    
        \eIf{$\proc(\epsilon_t)> \delta$}
        {$\epsilon_{L}\leftarrow \epsilon_t$}
        {$\epsilon_{H}\leftarrow \epsilon_t$}
    }
    
    \KwRet{$\epsilon_H$}
    
\end{algorithm}

\new{Due to symmetry, the divergence $D_{e^{\epsilon}}(P^{q}_{p,\beta}\| Q^{q}_{p,\beta})$ is always equal to $D_{e^{\epsilon}}(Q^{q}_{p,\beta}\| P^{q}_{p,\beta})$}. Thus, Algorithm \ref{alg:upperbound} only needs to compute one of them. We set the binary search upper bound, $\epsilon_H$, to $\log(p)$, given that the differing data item $x_1$ is at least protected by $\mech{R}_1$ with an indistinguishability level of $\log(p)$. The algorithm includes a sub-procedure, $\mathsf{Delta(\epsilon)}$, computing $D_{e^{\epsilon}}(P^{q}_{p,\beta}\| Q^{q}_{p,\beta})$ in $\tilde{O}(n)$ time. Hence, the total computational complexity is $\tilde{O}(n\cdot T)$, with $T$ being the number of binary search iterations, influencing the precision of the numerical indistinguishability level.

\section{Amplification Lower Bounds}\label{sec:lowerbounds}
In this section, we aim to establish lower bounds for privacy amplification via shuffling. These bounds serve to demonstrate the tightness of the upper bounds derived in the preceding sections. \new{Aiming at maximizing the divergence in observable shuffled messages given two neighboring datasets}, our approach involves selecting elements $y\in \dom{Y}$ in the message space where the probabilities $\mathbb{P}[\mech{R}(x_1^0)=y]$ and $\mathbb{P}[\mech{R}(x_1^1)=y]$ differ. \new{We then choose the worst-case data $x^*\in \dom{X}$ for other users to maximize the \emph{expected} probability ratio over differed elements}, and finally, we observe the number of occurrences of these differing elements in shuffled messages. \new{For succinctness, we summarize the occurrences of these elements and summarize them into two binomial counts}, similar to the technique used in Theorem \ref{the:reduction}. The resulting amplification lower bound is presented in Theorem \ref{the:lowerbound}, indicating the (worst-case) divergence of shuffled messages is lower bounded by the divergence between two binomial counts $P_{p_0,\beta}^{q_0,q_1}$ and $Q_{p_0,\beta}^{q_0,q_1}$.

\begin{theorem}[Privacy amplification lower bounds]\label{the:lowerbound}
Given $x_1^0, x_1^1\in \dom{X}$ and local randomizers $\{\mech{R}_1, \mech{R}_2\}$ that have finite output domain $\dom{Y}$, let $p_0$ denote $\frac{\sum_{y\in \dom{Y}}\llbracket \mathbb{P}[\mech{R}_1(x_1^1)=y]>\mathbb{P}[\mech{R}_1(x_1^0)=y]\rrbracket\cdot \mathbb{P}[\mech{R}_1(x_1^1)=y]}{\sum_{y\in \dom{Y}}\llbracket \mathbb{P}[\mech{R}_1(x_1^1)=y]>\mathbb{P}[\mech{R}_1(x_1^0)=y]\rrbracket\cdot \mathbb{P}[\mech{R}_1(x_1^0)=y]}$, let $\beta$ denote $D_{\epsilon^0}(\mech{R}_1(x_1^1)\| \mech{R}_1(x_1^0))$. Find an $x^*\in \dom{X}$ such that: 
\small
\begin{align*}
x^* &= \arg \max\limits_{x\in \mathcal{X}} \min \\
\Bigg\{&\frac{\sum_{y\in \mathcal{Y}}\left[\mathbb{P}[\mech{R}_1(x_1^1)=y]>\mathbb{P}[\mech{R}_1(x_1^0)=y]\right]\cdot \mathbb{P}[\mech{R}_1(x_1^1)=y]}{\sum_{y\in \mathcal{Y}}\left[\mathbb{P}[\mech{R}_1(x_1^1)=y]>\mathbb{P}[\mech{R}_1(x_1^0)=y]\right]\cdot \mathbb{P}[\mech{R}_2(x)=y]}, \\
&\frac{\sum_{y\in \mathcal{Y}}\left[\mathbb{P}[\mech{R}_1(x_1^1)=y]<\mathbb{P}[\mech{R}_1(x_1^0)=y]\right]\cdot \mathbb{P}[\mech{R}_1(x_1^0)=y]}{\sum_{y\in \mathcal{Y}}\left[\mathbb{P}[\mech{R}_1(x_1^1)=y]<\mathbb{P}[\mech{R}_1(x_1^0)=y]\right]\cdot \mathbb{P}[\mech{R}_2(x)=y]}\Bigg\},
\end{align*}
\normalsize
let $q_1$ denote $\frac{\sum_{y\in \dom{Y}}\llbracket \mathbb{P}[\mech{R}_1(x_1^1)=y]>\mathbb{P}[\mech{R}_1(x_1^0)=y]\rrbracket\cdot \mathbb{P}[\mech{R}_1(x_1^1)=y]}{\sum_{y\in \dom{Y}}\llbracket \mathbb{P}[\mech{R}_1(x_1^1)=y]>\mathbb{P}[\mech{R}_1(x_1^0)=y]\rrbracket\cdot \mathbb{P}[\mech{R}_2(x^*)=y]}$ and $q_0$ denote $\frac{\sum_{y\in \dom{Y}}\llbracket \mathbb{P}[\mech{R}_1(x_1^1)=y]<\mathbb{P}[\mech{R}_1(x_1^0)=y]\rrbracket\cdot \mathbb{P}[\mech{R}_1(x_1^0)=y]}{\sum_{y\in \dom{Y}}\llbracket \mathbb{P}[\mech{R}_1(x_1^1)=y]<\mathbb{P}[\mech{R}_1(x_1^0)=y]\rrbracket\cdot \mathbb{P}[\mech{R}_2(x^*)=y]}$. Let $\alpha=\frac{\beta'}{(p-1)}$, $r_0=\frac{\alpha p_0}{q_0}$ and $r_1=\frac{\alpha p_0}{q_1}$, and $C\sim Binom(n-1, r_0+r_1)$, $A\sim Binom(C, r_0/(r_0+r_1))$, and $\Delta_1=Bernoulli(p_0 \alpha)$, $\Delta_2=Bernoulli(1-\Delta_1, \alpha/(1-p_0 \alpha))$. Let $P_{p_0,\beta}^{q_0,q_1}$ denote $(A+\Delta_1, C-A+\Delta_2)$, $Q_{p_0,\beta}^{q_0,q_1}$ denote $(A+\Delta_2, C-A+\Delta_1)$, then there exists $x_2,...,x_n\in \dom{X}$ such that:
\small\begin{alignat*}{2}
&D(\mech{S}(\mech{R}_1(x_1^0),...,\mech{R}_2(x_n))\|\mech{S}(\mech{R}_1(x_1^1),...,\mech{R}_2(x_n)))
\geq&D(P_{p_0,\beta}^{q_0,q_1} \| Q_{p_0,\beta}^{q_0,q_1}).
\end{alignat*}\normalsize
\end{theorem}

According to the definition of binomial counts $P_{p_0,\beta}^{q_0,q_1}, Q_{p_0,\beta}^{q_0,q_1}$ in the lower bounding Theorem \ref{the:lowerbound} and $P_{p,\beta}^{q}, Q_{p,\beta}^{q}$ in the upper bounding Theorem \ref{the:reduction}, when $p_0=p$ and $q_0=q_1=q$, the upper bound matches lower bound and is thus exactly tight. This requires that the expected probability ratio $p_0$ 
equals to the maximum ratio $p$. This can be met by randomizers that employ an extremal probability design  \cite{kairouz2014extremal} that probability ratio ${\mathbb{P}[\mech{R}(x_i)=y]}/{\mathbb{P}[\mech{R}(x_{i'})=y]}$ must belongs to $\{1, e^{\epsilon}, e^{-\epsilon}\}$ for all $i,i'\in [n]$, $x_i,x_{i'}\in \dom{X}$, and $y\in \dom{Y}$. Besides, \new{the matching condition requires that the expected probability ratios $q_0$ and $q_1$ equal to the maximum ratio $q$}. 



Most state-of-the-art $\epsilon$-LDP randomizers adhere to the mentioned criteria. Notable examples include the generalized randomized response \cite{kairouz2016discrete} for more than $2$ options, the $k$-subset mechanism \cite{wang2019local} with cardinality $k\leq 2$, local hash \cite{wang2017locally} with length $l\geq 3$, Hardamard response \cite{acharya2018hadamard}, PrivUnit \cite{bonawitz2017practical} with cap area $c\leq 1/2$, PCKV-GRR \cite{gu2019pckv}, and Wheel mechanism \cite{wang2020set} with length $p\geq 1/(2s)$. As a result, the amplification upper bounds for these randomizers match the lower bounds and are precisely tight. In a similar vein, the privacy amplification upper bounds for several recent multi-message protocols are also exactly tight. Examples include Cheu \textit{et al.} \cite{cheu2022differentially} with $d\geq 3$, Balls-into-bins \cite{luo2022frequency} with $d\geq 3s$, pureDUMP and mixDUMP \cite{li2023privacy} with $d\geq 3$ (listed in Table \ref{tab:parametersmulti}). \new{Similar to the upper bound, there exists $\tilde{O}(n)$-complexity algorithms to compute the numerical value of the lower bound.} For other randomizers that fail to meet the aforementioned criteria, such as the Laplace and $k$-subset mechanism \cite{wang2019local} with cardinality $k\geq 3$, the expected ratios $p_0$ and $q_0, q_1$ are close to the maximum ratio $p$ and $q$ respectively. Consequently, the upper bound remains roughly tight.

\section{Parallel Composition in the Shuffle Model}\label{sec:parallel}
Most data analysis tasks involve multiple estimation targets. In the local setting of DP, it is common practice to partition the entire user population into multiple non-overlapping subsets and allow each subset to handle one estimation query with full budget $\epsilon_0$. This technique is used to achieve better utility compared to dividing the privacy budget $\epsilon_0$. Examples of tasks that utilize this approach include heavy hitter estimation \cite{qin2016heavy, bassily2017practical, wang2019locally}, multi-dimensional data publication \cite{ren2018lopub, wang2019collecting, wang2019answering}, frequent itemset mining \cite{wang2018locally}, range queries \cite{cormode2019answering}, marginal queries \cite{cormode2018marginal}, data synthesis \cite{zhang2018calm, wangT2019locally}, and machine learning \cite{jagannathan2009practical,xin2019differentially,fletcher2019decision}. This technique aligns with the parallel composition theorem of DP in the centralized setting \cite{dwork2006differential}.

In the shuffle model of DP, a naive approach to handling multiple queries is to divide the population into $K$ cohorts and separately amplify privacy \new{with about $n/K$ users} for each query. \new{A more effective approach is to allow each user to randomly select one query from all $K$ queries using a common distribution $P_k\in \Delta_K$, and to contribute to the selected query in parallel}. We illustrate this parallel approach in Algorithm \ref{alg:localparallel}. Since all base mechanisms $\mech{M}_k$ (for $k\in [K]$) satisfy $\epsilon_0$-LDP, the overall algorithm satisfies $\epsilon_0$-LDP.

\begin{algorithm}[htbp]
\SetKwInOut{Parameter}{Params}
    \caption{Parallel local randomizer}
    \label{alg:localparallel}
    \Parameter{Number of queries $K$, a probability distribution $P_{k}:[K]\mapsto [0,1]$, local base randomizers $\{\mech{M}_k: \dom{X}\mapsto \dom{Y}_k\}_{k\in [K]}$ each satisfies $\epsilon_0$-LDP and corresponds to one query.}
    \KwIn{An input $x\in \dom{X}$.}
    \KwOut{An output $z$ satisfies $\epsilon_0$-LDP.}
    
    {sample $k \sim P_{k}$}
    
    {$y \leftarrow \mech{M}_k(x)$}
    
    \KwRet{$y$}
\end{algorithm}

In light of the fact that each user draws a sample $k$ from the identical query distribution $P_k$, all users adopt the same randomization algorithm (Algorithm \ref{alg:localparallel}) that satisfies $\epsilon_0$-LDP. As a result, the privacy amplification via shuffling discussed in the previous sections remains valid with $n$ users amplify privacy together. We refer to Algorithm \ref{alg:localparallel} as $\mech{R}$. Drawing on the $(e^{\epsilon_0}, \frac{e^{\epsilon_0}-1}{e^{\epsilon_0}+1})$-variation property and the $e^{\epsilon_0}$-ratio property of arbitrary identical $\epsilon_0$-LDP randomizers, we can straightforwardly conclude that:
\small
\begin{alignat*}{2}
&D(\mech{S}(\mech{R}(x_1^0), ..., \mech{R}(x_n))\ \|\ \mech{S}(\mech{R}(x_1^1), ..., \mech{R}(x_n)))\\
\leq & D\Big(P^{e^{\epsilon_0}}_{e^{\epsilon_0},\frac{e^{\epsilon_0}-1}{e^{\epsilon_0}+1}}\| Q^{e^{\epsilon_0}}_{e^{\epsilon_0},\frac{e^{\epsilon_0}-1}{e^{\epsilon_0}+1}}\Big),
\end{alignat*}
\normalsize
which we refer to as the basic parallel composition theorem.

By exploiting the connection between total variation bounds and indistinguishable levels established in this study, we can derive a stronger privacy guarantee for Algorithm \ref{alg:localparallel}. Specifically, the total variation bound of the parallel local randomizer is bounded by the expectation of the base randomizers' total variation bound. We present this improved result in Theorem \ref{the:strongparallel} and refer to it as the advanced parallel composition in the shuffle model.

\begin{theorem}[Advanced parallel composition in the shuffle model]\label{the:strongparallel}
Assume $\mech{M}_k$ satisfy the $(e^{\epsilon_0}, \beta_k)$-variation property, and let $\mech{R}$ denote the Algorithm \ref{alg:localparallel}, then for any inputs $x_1^0, x_1^1, x_2, ..., x_n\in \dom{X}$:
\begin{alignat*}{2}
&D(\mech{S}(\mech{R}(x_1^0), ..., \mech{R}(x_n))\|\mech{S}(\mech{R}(x_1^1), ..., \mech{R}(x_n)))
\leq\ & D\big(P^{e^{\epsilon_0}}_{e^{\epsilon_0},\bar{\beta}}\| Q^{e^{\epsilon_0}}_{ e^{\epsilon_0},\bar{\beta}}\big)
\end{alignat*}
where $\bar{\beta}=\sum_{k'\in [K]} \beta_{k'}\cdot \mathbb{P}[P_k=k']$.
\end{theorem}
\begin{proof}
To establish the theorem, it is sufficient to demonstrate that $\mech{R}$ satisfies the $(e^{\epsilon_0}, \bar{\beta})$-variation property. The fact that 
$\mech{R}$ satisfies $D_{e^{\epsilon_0}}(\mech{R}(x_1^0)\|\mech{R}(x_1^1))=0$
is a direct consequence of the $\epsilon_0$-LDP guarantee of $\mech{R}$. Moreover, applying the definition of the total variation (or Hockey-stick divergence), we obtain:
\small
\begin{alignat*}{2}
D_{1}(\mech{R}(x_1^0)\| \mech{R}(x_1^1))&\leq \sum\nolimits_{k'\in [K]}\mathbb{P}[P_k=k']\cdot D_{1}(\mech{M}_{k'}(x_1^0)\| \mech{M}_{k'}(x_1^1))\\
&\leq \sum\nolimits_{k'\in [K]} \mathbb{P}[P_k=k']\cdot \beta_{k'}.
\end{alignat*}
\normalsize
\end{proof}

\vspace*{-1em}
\section{Numerical Results}\label{sec:exp}
In this section, we present evaluations of the proposed variation-ratio framework for both single-message and multi-message protocols in the shuffle model. Our main focus is to demonstrate the effectiveness and efficiency of numerical upper bounds. We also validate the effectiveness of the closed-form bounds presented in Theorems \ref{the:tightlevel} and \ref{the:level}. Additionally, we demonstrate the performance improvements achieved through the advanced parallel composition presented in Theorem \ref{the:strongparallel}, taking private range queries in the shuffle model as a concrete example.


\subsection{On Single-message Protocols}
To evaluate the effectiveness of the proposed variation-ratio framework for privacy amplification on LDP randomizers, we compare it with existing amplification bounds such as the privacy blanket \cite{balle2019privacy}, clone reduction \cite{feldman2022hiding}, and the stronger clone reduction \cite{feldman2023stronger}. We specifically consider two state-of-the-art LDP randomizers for discrete distribution estimation: the subset selection mechanism \cite{wang2019local,ye2018optimal} and the optimal local hash \cite{wang2017locally}. For the privacy blanket method, we present the tighter bound between the upper bounds given by ``Hoeffding, Generic'' and ``Bennett, Generic'' for general LDP randomizers, denoted as $privacy$-$blanket, general$. We also present the tighter bound between the upper bounds given by ``Hoeffding'' and ``Bennett'' with randomizer-specific parameters, i.e., with total variation similarity $\gamma={{d\choose k}}/({e^{\epsilon_0} {d-1\choose k-1}+{d-1\choose k}})$ for the subset mechanism and total variation similarity $\gamma=\frac{l}{e^{\epsilon_0} +l-1}$ for the optimal local hash, denoted as $privacy$-$blanket, specific$. All presented results are numerical amplification upper bounds (except the classical $\mathrm{EFMRTT19}$ \cite{erlingsson2019amplification} providing only closed-form bounds).

We use \emph{amplification ratio} to measure the effectiveness of privacy amplification. It is defined as the ratio of the local budget $\epsilon_0$ to the amplified privacy budget $\epsilon$:
\[\text{amplification ratio}={\epsilon_0}/{\epsilon}.\]
We present the amplification ratio results for the subset and optimal local hash mechanisms in Figures \ref{fig:subsetn46} and \ref{fig:olhn46}, respectively. The results demonstrate that our variation-ratio analyses can save about $30\%$ of the privacy budget for both mechanisms when compared with the best existing bounds. Notably, in the optimal local hash mechanism, when $l=e^{\epsilon_0}+1$ is greater than $2$, the amplification upper bounds obtained from our variation-ratio framework are tight, matching the lower bounds presented in Section \ref{sec:lowerbounds}.

\begin{figure*}[ht]
\begin{center}
\centerline{\includegraphics[width=135mm]{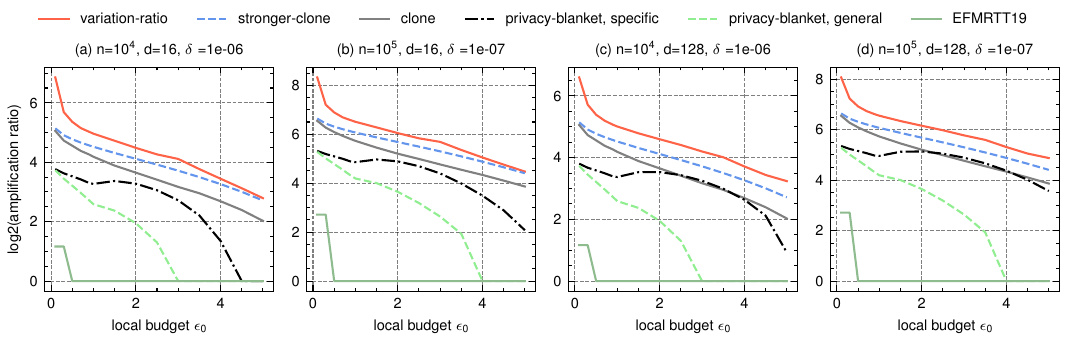}}
\vskip -0.18in
\caption{Numerical comparison of amplification effects (base $2$ logarithm of amplification ratio) of subset selection mechanism with $n=10^4$ or $10^5$, domain size $d=16$ or $128$, and varying local budget $\epsilon_0\in [0.1, 5.0]$.}
\label{fig:subsetn46}
\end{center}
\vspace*{-1.1em}
\end{figure*}

\begin{figure*}[ht]
\begin{center}
\centerline{\includegraphics[width=135mm]{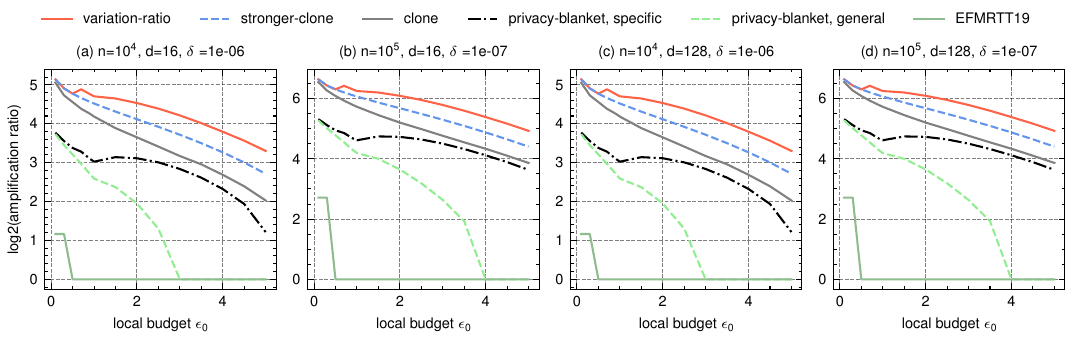}}
\vskip -0.18in
\caption{Numerical comparison of amplification effects (base $2$ logarithm of amplification ratio) of optimal local hash mechanism with $n=10^4$ or $10^5$, domain size $d=16$ or $128$, and local budget $\epsilon_0\in [0.1, 5.0]$.}
\label{fig:olhn46}
\end{center}
\vspace*{-1.0em}
\end{figure*}

\subsection{On Multi-message Protocols}\label{subsec:multimessage}

\begin{figure*}
\begin{center}
\centerline{\includegraphics[width=137mm]{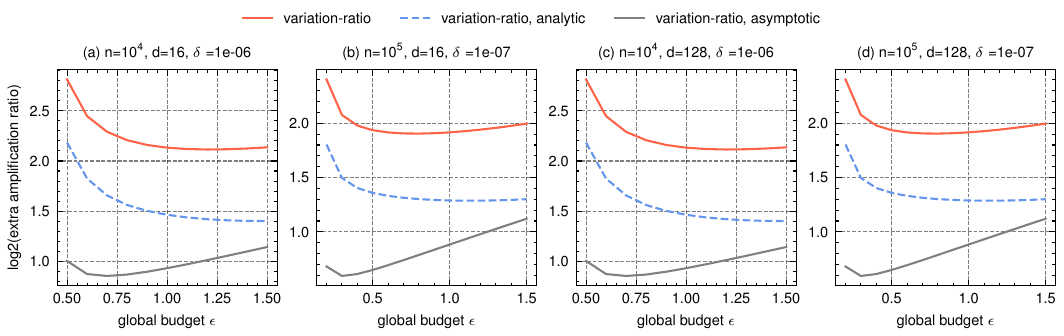}}
\vskip -0.18in
\caption{Numerical comparison of amplification effects (base $2$ logarithm of extra amplification ratio) of the Cheu \textit{et al.} \cite{cheu2022differentially} multi-message protocol with $n=10^4$ or $10^5$, domain size $d=16$ or $128$, and varying global budget $\epsilon'\in [0.01, 1.5]$.}
\label{fig:cheun46}
\end{center}
\vspace*{-0.9em}
\end{figure*}

\begin{figure*}
\begin{center}
\centerline{\includegraphics[width=135mm]{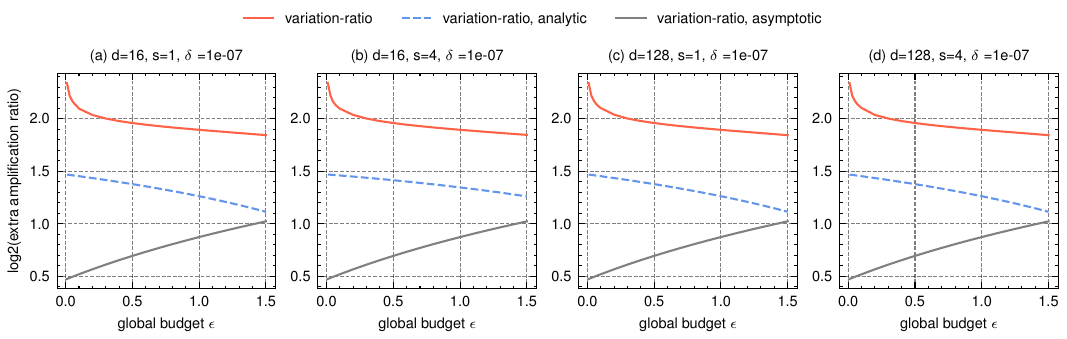}}
\vskip -0.18in
\caption{Numerical comparison of amplification effects (base $2$ logarithm of extra amplification ratio) of the balls-into-bins  multi-message protocol with $n=\frac{32\log(2/\delta)d}{\epsilon'^2 s}$ \cite{luo2022frequency} and varying global budget $\epsilon'\in [0.01, 1.5]$.}
\label{fig:b2bn46}
\end{center}
\vspace*{-1.3em}
\end{figure*}

To assess the efficacy of variation-ratio analyses for multi-message protocols, we apply it to two state-of-the-art histogram estimation protocols, namely Cheu \textit{et al.} \cite{cheu2022differentially} and balls-into-bins \cite{luo2022frequency}. We compare the amplified privacy provided by the original works, denoted as $\epsilon'$, with that given by variation-ratio analyses, denoted as $\epsilon$. To measure the additional privacy amplification offered by our analyses, we use the \textit{extra amplification ratio}, defined as:
\[\text{extra amplification ratio}={\epsilon'}/{\epsilon}.\]
As depicted in Figures \ref{fig:cheun46} and \ref{fig:b2bn46}, our numerical bounds, denoted as $variation$-$ratio$, are significantly tighter and result in a privacy budget savings of approximately $75\%$. Moreover, we present closed-form bounds from Theorem \ref{the:tightlevel} and Theorem \ref{the:level}, denoted as $variation$-$ratio, analytic$ and $variation$-$ratio, asymptotic$, respectively. Both bounds are tighter than those obtained from the original works, and the closed-form bounds from Theorem \ref{the:tightlevel} yield more than $50\%$ budget savings.

\subsection{On Advanced Parallel Composition}
This part evaluates the privacy amplification effects provided by the advanced parallel composition in Theorem \ref{the:strongparallel}. To illustrate the effectiveness of this technique, we consider range queries over a categorical domain $[1:d]$, a well-studied problem in the literature \cite{cormode2019answering,erlingsson2019amplification}. To avoid $\Theta(d)$ errors in estimators, a common practice is to represent categories in a hierarchical form and let each user report one hierarchy level. For categorical domains of size $d=2^{H}$, the $k$-th value in the $h$-th hierarchy is given by:
$$V_{h,k}=\llbracket j\in [(k-1)\cdot 2^h : k\cdot 2^h] \rrbracket,$$
where $h\in [0:H-1]$, $k\in [1 : d/2^h]$, and $\llbracket \ \rrbracket$ is the Iverson bracket.

\begin{figure*}
\begin{center}
\centerline{\includegraphics[width=135mm]{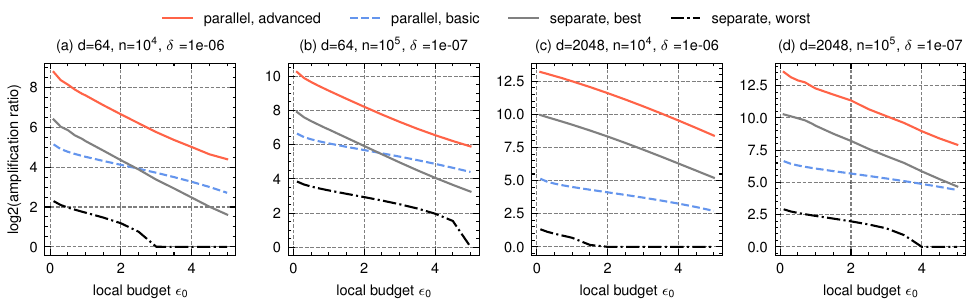}}
\vskip -0.18in
\caption{Numerical comparison of amplification effects (base $2$ logarithm of amplification ratio) of separated approach, basic parallel composition, and advanced parallel composition.}
\label{fig:paralleln45}
\end{center}
\vspace*{-1.0em}
\end{figure*}

Follow the approach suggested in \cite{cormode2019answering}, we assume each user uniformly selects one hierarchy level $h\in [0:H-1]$ and uses the generalized randomized response mechanism with full budget $\epsilon_0$ to report the one-hot vector $V_{h,*}$ (the generalized randomized response is optimal in the low local privacy regime \cite{ye2018optimal}). Using the basic parallel composition theorem, this parallel local randomizer with $H$ queries implies the following variation-ratio parameters: $p=e^{\epsilon_0}$, $\beta=\frac{e^{\epsilon_0}-1}{e^{\epsilon_0}+1}$, and $q=e^{\epsilon_0}$. In contrast, according to Table \ref{tab:parametersldp} and the advanced parallel composition in Theorem \ref{the:strongparallel}, the tighter variation-ratio parameters can be computed as follows: $p=e^{\epsilon_0}$, $q=e^{\epsilon_0}$, and $\beta=\sum_{h\in [0:H-1]}\frac{1}{H}\frac{e^{\epsilon_0}-1}{e^{\epsilon_0}+d/2^h-1}$.

\begin{table*}
\small
\caption{Amplification bounds (w.r.t. $\delta=0.01/n$) and running time comparison for general $\epsilon_0$-LDP randomizers with iterations $T=20$ or $T=10$ and varying number of users.}
\vspace*{-0.7em}
\label{tab:T20time}
\centering
\begin{tabular}{|c|c||c|c||c|c|}
\hline
     \multicolumn{2}{|c||}{} & \multicolumn{2}{|c||}{$T=20$}  & \multicolumn{2}{|c|}{$T=10$} \\
    \cline{1-6} 
    & \textbf{number of users $n$} & \textbf{amplified privacy } $\epsilon$ & \textbf{time (seconds)} & \textbf{amplified privacy } $\epsilon$ & \textbf{time (seconds)}  \\ \hline
\multirow{3}{*}{$\epsilon_0=1.0$} & $10^4$                    &   $0.0433$               &    $0.45$                 & $0.0440$ & $2.3$ \\ \cline{2-6}
    & $10^6$                & $0.00503$                      & $16.7$
    & $0.00586$ & $4.7$ \\ \cline{2-6}
    & $10^8$               & $0.000566$                      & $106.2$ 
    & $0.000977$ & $0.43$ \\
    \hline
\multirow{3}{*}{$\epsilon_0=3.0$} & $10^4$                   & $0.227$                 &   $3.8$                    & $0.229$ & $1.5$ \\ \cline{2-6} 
    & $10^6$                 &   $0.0255$               &  $13.9$                     
    & $0.0264$ & $2.8$ \\ \cline{2-6}
    & $10^8$                 &  $0.00283$                &  $\mathbf{117.5}$                     
    & $0.00293$ & $5.2$ \\ \hline
\multirow{3}{*}{$\epsilon_0=5.0$} & $10^4$                   & $0.743$                 &   $2.4$                    & $0.743$ & $0.9$\\ \cline{2-6} 
    & $10^6$                 &   $0.0778$               &  $9.3$                     
    & $0.0782$ & $3.1$ \\ \cline{2-6} 
    & $10^8$                 &  $0.00853$                &  $44.6$                     
    & $0.00977$ & $\mathbf{7.9}$\\ \hline
\multirow{3}{*}{$\epsilon_0=7.0$} & $10^4$                   & $6.99$                 &   $0.24$                    & $6.99$ & $0.31$\\ \cline{2-6} 
    & $10^6$                 &   $0.224$               &  $2.8$                     
    & $0.225$ & $1.2$ \\ \cline{2-6} 
    & $10^8$                 &  $0.0242$                &  $17.9$                     
    & $0.0273$ & $5.1$\\ \hline
\end{tabular}
\vspace*{0.0em}
\end{table*}

We present a comparative analysis of the numerical privacy amplification achieved through advanced parallel composition and basic parallel composition in Figure \ref{fig:paralleln45}, for the following settings: $d=64$ or $2048$, and $n=10^4$ or $10^5$. Our findings demonstrate that the use of advanced parallel composition results in a reduction of the privacy budget by approximately $75\%$. Moreover, we investigate the privacy amplification effects of the separate approach wherein non-overlapping users report each hierarchy level separately. We present the amplification results for $n/H$ users using the best-possible parameters: $p=e^{\epsilon_0}$, $q=e^{\epsilon_0}$, and $\beta=\frac{e^{\epsilon_0}-1}{e^{\epsilon_0}+d-1}$ (denoted as $separate, best$ in Figure \ref{fig:paralleln45}) or the worst-possible parameters: $p=e^{\epsilon_0}$, $q=e^{\epsilon_0}$, and $\beta=\frac{e^{\epsilon_0}-1}{e^{\epsilon_0}+1}$ (denoted as $separate, worst$). Our results show that the advanced parallel composition saves budget by $80\%$-$95\%$ when compared to the separated approach.

\balance

\subsection{Efficiency Evaluation}
\new{We implement Theorem \ref{the:dbound} and Algorithm \ref{alg:upperbound} in Python 3.8 and execute it on a desktop computer with Intel Core i7-10700KF @3.8GHz and $32$GB memory}. In Table \ref{tab:T20time}, we present the results for $T=20$ or $T=10$ ($T$ is the number of binary search iterations on amplified level $\epsilon$, see Section \ref{subsec:numericalupper}), where we vary the local budget $\epsilon_0$ and number of users $n$. The numerical privacy amplification bounds presented in the table correspond to general LDP mechanisms (refer to Table \ref{tab:parametersldp} for variation-ratio parameters). Our results demonstrate that the running time of the implementation is less sensitive to the local budget $\epsilon_0$, and mainly depends on the population size $n$ and number of iterations $T$. Furthermore, the running time grows linearly with $n$, and we can obtain tight numerical privacy amplification bounds within a dozen seconds even when $n$ is extremely large (e.g., $n=10^8$). Our results also show that one can trade tightness for computational efficiency by choosing a smaller value for $T$. Specifically, comparing the results obtained for $T=10$ with those for $T=20$, we observe that the latter provides slightly tighter bounds, but at the cost of much longer running time.

\section{Conclusion}\label{sec:conclusion}
This work has presented a unified, tight, easy-to-use, and efficient framework for analyzing privacy amplification within the emerging shuffle model. The framework has been shown to provide tight bounds for a wide range of state-of-the-art single-message/multi-message protocols, by translating two novel yet intuitive parameterizations (namely, total variation and probability ratio) about local randomizers into differential privacy levels. Additionally, the framework has induced an advanced parallel composition theorem in the shuffle model, widely applicable to count queries, data synthesis, data mining, and machine learning. Our framework also leads to a fast $\tilde{O}(n)$-complexity procedure for numerical privacy amplification analyses. Comprehensive experiments affirm the effectiveness and efficiency of the proposed framework.
   



\begin{acks}
\justifying{Shaowei Wang and Jin Li are also with the Guangdong Provincial Key Laboratory of Blockchain Security, China. Yun Peng is the corresponding author. This work is supported by National Key Project of China (No. 2020YFB1005700), National Natural Science Foundation of China (No.62372120, 62102108, U23A20307, U21A20463, 62172385, 62172117, U1936218, 61802383), Natural Science Foundation of Guangdong Province (No.2022A1515010061,  No.2023A1515030273), and Guangzhou Basic and Applied Basic Research Foundation (No.202201010194,  No.202201020139).}
\end{acks}

\clearpage

\balance
\bibliographystyle{ACM-Reference-Format}
\bibliography{refs}


\begin{thebibliography}{96}


\ifx \showCODEN    \undefined \def \showCODEN     #1{\unskip}     \fi
\ifx \showDOI      \undefined \def \showDOI       #1{#1}\fi
\ifx \showISBNx    \undefined \def \showISBNx     #1{\unskip}     \fi
\ifx \showISBNxiii \undefined \def \showISBNxiii  #1{\unskip}     \fi
\ifx \showISSN     \undefined \def \showISSN      #1{\unskip}     \fi
\ifx \showLCCN     \undefined \def \showLCCN      #1{\unskip}     \fi
\ifx \shownote     \undefined \def \shownote      #1{#1}          \fi
\ifx \showarticletitle \undefined \def \showarticletitle #1{#1}   \fi
\ifx \showURL      \undefined \def \showURL       {\relax}        \fi
\providecommand\bibfield[2]{#2}
\providecommand\bibinfo[2]{#2}
\providecommand\natexlab[1]{#1}
\providecommand\showeprint[2][]{arXiv:#2}

\bibitem[\protect\citeauthoryear{Acharya, Sun, and Zhang}{Acharya et~al\mbox{.}}{2018}]%
        {acharya2018hadamard}
\bibfield{author}{\bibinfo{person}{Jayadev Acharya}, \bibinfo{person}{Ziteng Sun}, {and} \bibinfo{person}{Huanyu Zhang}.} \bibinfo{year}{2018}\natexlab{}.
\newblock \showarticletitle{Hadamard Response: Estimating Distributions Privately, Efficiently, and with Little Communication}.
\newblock \bibinfo{journal}{\emph{arXiv preprint arXiv:1802.04705}} (\bibinfo{year}{2018}).
\newblock


\bibitem[\protect\citeauthoryear{Alvim, Chatzikokolakis, Palamidessi, and Pazii}{Alvim et~al\mbox{.}}{2018}]%
        {alvim2018local}
\bibfield{author}{\bibinfo{person}{M{\'a}rio Alvim}, \bibinfo{person}{Konstantinos Chatzikokolakis}, \bibinfo{person}{Catuscia Palamidessi}, {and} \bibinfo{person}{Anna Pazii}.} \bibinfo{year}{2018}\natexlab{}.
\newblock \showarticletitle{Local differential privacy on metric spaces: optimizing the trade-off with utility}. In \bibinfo{booktitle}{\emph{2018 IEEE 31st Computer Security Foundations Symposium (CSF)}}. IEEE, \bibinfo{pages}{262--267}.
\newblock


\bibitem[\protect\citeauthoryear{Andr{\'e}s, Bordenabe, Chatzikokolakis, and Palamidessi}{Andr{\'e}s et~al\mbox{.}}{2013}]%
        {andres2013geo}
\bibfield{author}{\bibinfo{person}{Miguel~E Andr{\'e}s}, \bibinfo{person}{Nicol{\'a}s~E Bordenabe}, \bibinfo{person}{Konstantinos Chatzikokolakis}, {and} \bibinfo{person}{Catuscia Palamidessi}.} \bibinfo{year}{2013}\natexlab{}.
\newblock \showarticletitle{Geo-indistinguishability: Differential privacy for location-based systems}. In \bibinfo{booktitle}{\emph{CCS. ACM}}.
\newblock


\bibitem[\protect\citeauthoryear{Balcer and Cheu}{Balcer and Cheu}{2020}]%
        {balcer2020separating}
\bibfield{author}{\bibinfo{person}{Victor Balcer} {and} \bibinfo{person}{Albert Cheu}.} \bibinfo{year}{2020}\natexlab{}.
\newblock \showarticletitle{Separating Local \& Shuffled Differential Privacy via Histograms}. In \bibinfo{booktitle}{\emph{1st Conference on Information-Theoretic Cryptography}}.
\newblock


\bibitem[\protect\citeauthoryear{Balcer, Cheu, Joseph, and Mao}{Balcer et~al\mbox{.}}{2021}]%
        {balcer2021connecting}
\bibfield{author}{\bibinfo{person}{Victor Balcer}, \bibinfo{person}{Albert Cheu}, \bibinfo{person}{Matthew Joseph}, {and} \bibinfo{person}{Jieming Mao}.} \bibinfo{year}{2021}\natexlab{}.
\newblock \showarticletitle{Connecting robust shuffle privacy and pan-privacy}. In \bibinfo{booktitle}{\emph{Proceedings of the 2021 ACM-SIAM Symposium on Discrete Algorithms (SODA)}}. SIAM, \bibinfo{pages}{2384--2403}.
\newblock


\bibitem[\protect\citeauthoryear{Balle, Barthe, and Gaboardi}{Balle et~al\mbox{.}}{2018}]%
        {balle2018subsampling}
\bibfield{author}{\bibinfo{person}{Borja Balle}, \bibinfo{person}{Gilles Barthe}, {and} \bibinfo{person}{Marco Gaboardi}.} \bibinfo{year}{2018}\natexlab{}.
\newblock \showarticletitle{Privacy amplification by subsampling: Tight analyses via couplings and divergences}.
\newblock \bibinfo{journal}{\emph{Advances in Neural Information Processing Systems}}  \bibinfo{volume}{31} (\bibinfo{year}{2018}).
\newblock


\bibitem[\protect\citeauthoryear{Balle, Bell, Gasc{\'o}n, and Nissim}{Balle et~al\mbox{.}}{2019}]%
        {balle2019privacy}
\bibfield{author}{\bibinfo{person}{Borja Balle}, \bibinfo{person}{James Bell}, \bibinfo{person}{Adria Gasc{\'o}n}, {and} \bibinfo{person}{Kobbi Nissim}.} \bibinfo{year}{2019}\natexlab{}.
\newblock \showarticletitle{The privacy blanket of the shuffle model}.
\newblock \bibinfo{journal}{\emph{CRYPTO}} (\bibinfo{year}{2019}).
\newblock


\bibitem[\protect\citeauthoryear{Balle, Bell, Gascon, and Nissim}{Balle et~al\mbox{.}}{2020}]%
        {balle2020private}
\bibfield{author}{\bibinfo{person}{Borja Balle}, \bibinfo{person}{James Bell}, \bibinfo{person}{Adria Gascon}, {and} \bibinfo{person}{Kobbi Nissim}.} \bibinfo{year}{2020}\natexlab{}.
\newblock \showarticletitle{Private summation in the multi-message shuffle model}.
\newblock \bibinfo{journal}{\emph{CCS}} (\bibinfo{year}{2020}).
\newblock


\bibitem[\protect\citeauthoryear{Bao, Yang, Xiao, and Ding}{Bao et~al\mbox{.}}{2021}]%
        {bao2021cgm}
\bibfield{author}{\bibinfo{person}{Ergute Bao}, \bibinfo{person}{Yin Yang}, \bibinfo{person}{Xiaokui Xiao}, {and} \bibinfo{person}{Bolin Ding}.} \bibinfo{year}{2021}\natexlab{}.
\newblock \showarticletitle{CGM: an enhanced mechanism for streaming data collection with local differential privacy}.
\newblock \bibinfo{journal}{\emph{Proceedings of the VLDB Endowment}} \bibinfo{volume}{14}, \bibinfo{number}{11} (\bibinfo{year}{2021}), \bibinfo{pages}{2258--2270}.
\newblock


\bibitem[\protect\citeauthoryear{Bassily, Nissim, Stemmer, and Guha~Thakurta}{Bassily et~al\mbox{.}}{2017}]%
        {bassily2017practical}
\bibfield{author}{\bibinfo{person}{Raef Bassily}, \bibinfo{person}{Kobbi Nissim}, \bibinfo{person}{Uri Stemmer}, {and} \bibinfo{person}{Abhradeep Guha~Thakurta}.} \bibinfo{year}{2017}\natexlab{}.
\newblock \showarticletitle{Practical locally private heavy hitters}.
\newblock \bibinfo{journal}{\emph{Advances in Neural Information Processing Systems}}  \bibinfo{volume}{30} (\bibinfo{year}{2017}).
\newblock


\bibitem[\protect\citeauthoryear{Bennett, Lanning, et~al\mbox{.}}{Bennett et~al\mbox{.}}{2007}]%
        {bennett2007netflix}
\bibfield{author}{\bibinfo{person}{James Bennett}, \bibinfo{person}{Stan Lanning}, {et~al\mbox{.}}} \bibinfo{year}{2007}\natexlab{}.
\newblock \showarticletitle{The netflix prize}. In \bibinfo{booktitle}{\emph{Proceedings of KDD cup and workshop}}, Vol.~\bibinfo{volume}{2007}. \bibinfo{pages}{35}.
\newblock


\bibitem[\protect\citeauthoryear{Bhowmick, Duchi, Freudiger, Kapoor, and Rogers}{Bhowmick et~al\mbox{.}}{2018}]%
        {bhowmick2018protection}
\bibfield{author}{\bibinfo{person}{Abhishek Bhowmick}, \bibinfo{person}{John Duchi}, \bibinfo{person}{Julien Freudiger}, \bibinfo{person}{Gaurav Kapoor}, {and} \bibinfo{person}{Ryan Rogers}.} \bibinfo{year}{2018}\natexlab{}.
\newblock \showarticletitle{Protection against reconstruction and its applications in private federated learning}.
\newblock \bibinfo{journal}{\emph{arXiv preprint arXiv:1812.00984}} (\bibinfo{year}{2018}).
\newblock


\bibitem[\protect\citeauthoryear{Bittau, Erlingsson, Maniatis, Mironov, Raghunathan, Lie, Rudominer, Kode, Tinnes, and Seefeld}{Bittau et~al\mbox{.}}{2017}]%
        {bittau2017prochlo}
\bibfield{author}{\bibinfo{person}{Andrea Bittau}, \bibinfo{person}{{\'U}lfar Erlingsson}, \bibinfo{person}{Petros Maniatis}, \bibinfo{person}{Ilya Mironov}, \bibinfo{person}{Ananth Raghunathan}, \bibinfo{person}{David Lie}, \bibinfo{person}{Mitch Rudominer}, \bibinfo{person}{Ushasree Kode}, \bibinfo{person}{Julien Tinnes}, {and} \bibinfo{person}{Bernhard Seefeld}.} \bibinfo{year}{2017}\natexlab{}.
\newblock \showarticletitle{Prochlo: Strong privacy for analytics in the crowd}. In \bibinfo{booktitle}{\emph{Proceedings of the 26th symposium on operating systems principles}}. \bibinfo{pages}{441--459}.
\newblock


\bibitem[\protect\citeauthoryear{Bonawitz, Ivanov, Kreuter, Marcedone, McMahan, Patel, Ramage, Segal, and Seth}{Bonawitz et~al\mbox{.}}{2017}]%
        {bonawitz2017practical}
\bibfield{author}{\bibinfo{person}{Keith Bonawitz}, \bibinfo{person}{Vladimir Ivanov}, \bibinfo{person}{Ben Kreuter}, \bibinfo{person}{Antonio Marcedone}, \bibinfo{person}{H~Brendan McMahan}, \bibinfo{person}{Sarvar Patel}, \bibinfo{person}{Daniel Ramage}, \bibinfo{person}{Aaron Segal}, {and} \bibinfo{person}{Karn Seth}.} \bibinfo{year}{2017}\natexlab{}.
\newblock \showarticletitle{Practical secure aggregation for privacy-preserving machine learning}. In \bibinfo{booktitle}{\emph{proceedings of the 2017 ACM SIGSAC Conference on Computer and Communications Security}}. \bibinfo{pages}{1175--1191}.
\newblock


\bibitem[\protect\citeauthoryear{Bonomi and Xiong}{Bonomi and Xiong}{2013}]%
        {bonomi2013mining}
\bibfield{author}{\bibinfo{person}{Luca Bonomi} {and} \bibinfo{person}{Li Xiong}.} \bibinfo{year}{2013}\natexlab{}.
\newblock \showarticletitle{Mining frequent patterns with differential privacy}.
\newblock \bibinfo{journal}{\emph{Proceedings of the VLDB Endowment}} \bibinfo{volume}{6}, \bibinfo{number}{12} (\bibinfo{year}{2013}), \bibinfo{pages}{1422--1427}.
\newblock


\bibitem[\protect\citeauthoryear{Calzada}{Calzada}{2022}]%
        {calzada2022citizens}
\bibfield{author}{\bibinfo{person}{Igor Calzada}.} \bibinfo{year}{2022}\natexlab{}.
\newblock \showarticletitle{Citizens’ data privacy in China: The state of the art of the Personal Information Protection Law (PIPL)}.
\newblock \bibinfo{journal}{\emph{Smart Cities}} \bibinfo{volume}{5}, \bibinfo{number}{3} (\bibinfo{year}{2022}), \bibinfo{pages}{1129--1150}.
\newblock


\bibitem[\protect\citeauthoryear{Chatzikokolakis, Andr{\'e}s, Bordenabe, and Palamidessi}{Chatzikokolakis et~al\mbox{.}}{2013}]%
        {chatzikokolakis2013broadening}
\bibfield{author}{\bibinfo{person}{Konstantinos Chatzikokolakis}, \bibinfo{person}{Miguel~E Andr{\'e}s}, \bibinfo{person}{Nicol{\'a}s~Emilio Bordenabe}, {and} \bibinfo{person}{Catuscia Palamidessi}.} \bibinfo{year}{2013}\natexlab{}.
\newblock \showarticletitle{Broadening the scope of Differential Privacy using metrics}. In \bibinfo{booktitle}{\emph{PETS. Springer}}.
\newblock


\bibitem[\protect\citeauthoryear{Chen, Fung, Yu, and Desai}{Chen et~al\mbox{.}}{2014}]%
        {chen2014correlated}
\bibfield{author}{\bibinfo{person}{Rui Chen}, \bibinfo{person}{Benjamin~CM Fung}, \bibinfo{person}{Philip~S Yu}, {and} \bibinfo{person}{Bipin~C Desai}.} \bibinfo{year}{2014}\natexlab{}.
\newblock \showarticletitle{Correlated network data publication via differential privacy}.
\newblock \bibinfo{journal}{\emph{The VLDB Journal}}  \bibinfo{volume}{23} (\bibinfo{year}{2014}), \bibinfo{pages}{653--676}.
\newblock


\bibitem[\protect\citeauthoryear{Chen, Mohammed, Fung, Desai, and Xiong}{Chen et~al\mbox{.}}{2011}]%
        {chen2011publishing}
\bibfield{author}{\bibinfo{person}{Rui Chen}, \bibinfo{person}{Noman Mohammed}, \bibinfo{person}{Benjamin~CM Fung}, \bibinfo{person}{Bipin~C Desai}, {and} \bibinfo{person}{Li Xiong}.} \bibinfo{year}{2011}\natexlab{}.
\newblock \showarticletitle{Publishing set-valued data via differential privacy}.
\newblock \bibinfo{journal}{\emph{Proceedings of the VLDB Endowment}} \bibinfo{volume}{4}, \bibinfo{number}{11} (\bibinfo{year}{2011}), \bibinfo{pages}{1087--1098}.
\newblock


\bibitem[\protect\citeauthoryear{Cheu, Smith, Ullman, Zeber, and Zhilyaev}{Cheu et~al\mbox{.}}{2019}]%
        {cheu2019distributed}
\bibfield{author}{\bibinfo{person}{Albert Cheu}, \bibinfo{person}{Adam Smith}, \bibinfo{person}{Jonathan Ullman}, \bibinfo{person}{David Zeber}, {and} \bibinfo{person}{Maxim Zhilyaev}.} \bibinfo{year}{2019}\natexlab{}.
\newblock \showarticletitle{Distributed differential privacy via shuffling}.
\newblock \bibinfo{journal}{\emph{EUROCRYPT}} (\bibinfo{year}{2019}).
\newblock


\bibitem[\protect\citeauthoryear{Cheu and Zhilyaev}{Cheu and Zhilyaev}{2022}]%
        {cheu2022differentially}
\bibfield{author}{\bibinfo{person}{Albert Cheu} {and} \bibinfo{person}{Maxim Zhilyaev}.} \bibinfo{year}{2022}\natexlab{}.
\newblock \showarticletitle{Differentially private histograms in the shuffle model from fake users}. In \bibinfo{booktitle}{\emph{2022 IEEE Symposium on Security and Privacy (SP)}}. IEEE, \bibinfo{pages}{440--457}.
\newblock


\bibitem[\protect\citeauthoryear{Cormode, Kulkarni, and Srivastava}{Cormode et~al\mbox{.}}{2018}]%
        {cormode2018marginal}
\bibfield{author}{\bibinfo{person}{Graham Cormode}, \bibinfo{person}{Tejas Kulkarni}, {and} \bibinfo{person}{Divesh Srivastava}.} \bibinfo{year}{2018}\natexlab{}.
\newblock \showarticletitle{Marginal release under local differential privacy}. In \bibinfo{booktitle}{\emph{Proceedings of the 2018 International Conference on Management of Data}}. ACM, \bibinfo{pages}{131--146}.
\newblock


\bibitem[\protect\citeauthoryear{Cormode, Kulkarni, and Srivastava}{Cormode et~al\mbox{.}}{2019a}]%
        {2019answering}
\bibfield{author}{\bibinfo{person}{Graham Cormode}, \bibinfo{person}{Tejas Kulkarni}, {and} \bibinfo{person}{Divesh Srivastava}.} \bibinfo{year}{2019}\natexlab{a}.
\newblock \showarticletitle{Answering range queries under local differential privacy}.
\newblock \bibinfo{journal}{\emph{VLDB}} (\bibinfo{year}{2019}).
\newblock


\bibitem[\protect\citeauthoryear{Cormode, Kulkarni, and Srivastava}{Cormode et~al\mbox{.}}{2019b}]%
        {cormode2019answering}
\bibfield{author}{\bibinfo{person}{Graham Cormode}, \bibinfo{person}{Tejas Kulkarni}, {and} \bibinfo{person}{Divesh Srivastava}.} \bibinfo{year}{2019}\natexlab{b}.
\newblock \showarticletitle{Answering range queries under local differential privacy}.
\newblock \bibinfo{journal}{\emph{Proceedings of the VLDB Endowment}} \bibinfo{volume}{12}, \bibinfo{number}{10} (\bibinfo{year}{2019}), \bibinfo{pages}{1126--1138}.
\newblock


\bibitem[\protect\citeauthoryear{Du, Zhang, Bai, Liu, Ji, Cheng, and Chen}{Du et~al\mbox{.}}{2021}]%
        {du2021ahead}
\bibfield{author}{\bibinfo{person}{Linkang Du}, \bibinfo{person}{Zhikun Zhang}, \bibinfo{person}{Shaojie Bai}, \bibinfo{person}{Changchang Liu}, \bibinfo{person}{Shouling Ji}, \bibinfo{person}{Peng Cheng}, {and} \bibinfo{person}{Jiming Chen}.} \bibinfo{year}{2021}\natexlab{}.
\newblock \showarticletitle{AHEAD: adaptive hierarchical decomposition for range query under local differential privacy}. In \bibinfo{booktitle}{\emph{Proceedings of the 2021 ACM SIGSAC Conference on Computer and Communications Security}}. \bibinfo{pages}{1266--1288}.
\newblock


\bibitem[\protect\citeauthoryear{Duchi, Jordan, and Wainwright}{Duchi et~al\mbox{.}}{2013}]%
        {duchi2013local}
\bibfield{author}{\bibinfo{person}{John~C Duchi}, \bibinfo{person}{Michael~I Jordan}, {and} \bibinfo{person}{Martin~J Wainwright}.} \bibinfo{year}{2013}\natexlab{}.
\newblock \showarticletitle{Local privacy and statistical minimax rates}. In \bibinfo{booktitle}{\emph{Foundations of Computer Science (FOCS), 2013 IEEE 54th Annual Symposium on}}. IEEE, \bibinfo{pages}{429--438}.
\newblock


\bibitem[\protect\citeauthoryear{Dwork}{Dwork}{2006}]%
        {dwork2006differential}
\bibfield{author}{\bibinfo{person}{Cynthia Dwork}.} \bibinfo{year}{2006}\natexlab{}.
\newblock \showarticletitle{Differential privacy}. In \bibinfo{booktitle}{\emph{Automata, Languages and Programming: 33rd International Colloquium, ICALP 2006, Venice, Italy, July 10-14, 2006, Proceedings, Part II 33}}. Springer, \bibinfo{pages}{1--12}.
\newblock


\bibitem[\protect\citeauthoryear{Dwork}{Dwork}{2008}]%
        {dwork2008differential}
\bibfield{author}{\bibinfo{person}{Cynthia Dwork}.} \bibinfo{year}{2008}\natexlab{}.
\newblock \showarticletitle{Differential privacy: A survey of results}. In \bibinfo{booktitle}{\emph{Theory and Applications of Models of Computation: 5th International Conference, TAMC 2008, Xi’an, China, April 25-29, 2008. Proceedings 5}}. Springer, \bibinfo{pages}{1--19}.
\newblock


\bibitem[\protect\citeauthoryear{Dwork, McSherry, Nissim, and Smith}{Dwork et~al\mbox{.}}{2006}]%
        {dwork2006calibrating}
\bibfield{author}{\bibinfo{person}{Cynthia Dwork}, \bibinfo{person}{Frank McSherry}, \bibinfo{person}{Kobbi Nissim}, {and} \bibinfo{person}{Adam Smith}.} \bibinfo{year}{2006}\natexlab{}.
\newblock \showarticletitle{Calibrating noise to sensitivity in private data analysis}. In \bibinfo{booktitle}{\emph{Theory of Cryptography: Third Theory of Cryptography Conference, TCC 2006, New York, NY, USA, March 4-7, 2006. Proceedings 3}}. Springer, \bibinfo{pages}{265--284}.
\newblock


\bibitem[\protect\citeauthoryear{Erlingsson, Feldman, Mironov, Raghunathan, Talwar, and Thakurta}{Erlingsson et~al\mbox{.}}{2019}]%
        {erlingsson2019amplification}
\bibfield{author}{\bibinfo{person}{{\'U}lfar Erlingsson}, \bibinfo{person}{Vitaly Feldman}, \bibinfo{person}{Ilya Mironov}, \bibinfo{person}{Ananth Raghunathan}, \bibinfo{person}{Kunal Talwar}, {and} \bibinfo{person}{Abhradeep Thakurta}.} \bibinfo{year}{2019}\natexlab{}.
\newblock \showarticletitle{Amplification by shuffling: From local to central differential privacy via anonymity}.
\newblock \bibinfo{journal}{\emph{SODA}} (\bibinfo{year}{2019}).
\newblock


\bibitem[\protect\citeauthoryear{Erlingsson, Pihur, and Korolova}{Erlingsson et~al\mbox{.}}{2014}]%
        {erlingsson2014rappor}
\bibfield{author}{\bibinfo{person}{{\'U}lfar Erlingsson}, \bibinfo{person}{Vasyl Pihur}, {and} \bibinfo{person}{Aleksandra Korolova}.} \bibinfo{year}{2014}\natexlab{}.
\newblock \showarticletitle{Rappor: Randomized aggregatable privacy-preserving ordinal response}. In \bibinfo{booktitle}{\emph{Proceedings of the 2014 ACM SIGSAC conference on computer and communications security}}. ACM, \bibinfo{pages}{1054--1067}.
\newblock


\bibitem[\protect\citeauthoryear{Feldman, McMillan, and Talwar}{Feldman et~al\mbox{.}}{2022}]%
        {feldman2022hiding}
\bibfield{author}{\bibinfo{person}{Vitaly Feldman}, \bibinfo{person}{Audra McMillan}, {and} \bibinfo{person}{Kunal Talwar}.} \bibinfo{year}{2022}\natexlab{}.
\newblock \showarticletitle{Hiding among the clones: A simple and nearly optimal analysis of privacy amplification by shuffling}. In \bibinfo{booktitle}{\emph{2021 IEEE 62nd Annual Symposium on Foundations of Computer Science (FOCS)}}. IEEE, \bibinfo{pages}{954--964}.
\newblock


\bibitem[\protect\citeauthoryear{Feldman, McMillan, and Talwar}{Feldman et~al\mbox{.}}{2023}]%
        {feldman2023stronger}
\bibfield{author}{\bibinfo{person}{Vitaly Feldman}, \bibinfo{person}{Audra McMillan}, {and} \bibinfo{person}{Kunal Talwar}.} \bibinfo{year}{2023}\natexlab{}.
\newblock \showarticletitle{Stronger privacy amplification by shuffling for R{\'e}nyi and approximate differential privacy}. In \bibinfo{booktitle}{\emph{Proceedings of the 2023 Annual ACM-SIAM Symposium on Discrete Algorithms (SODA)}}. SIAM, \bibinfo{pages}{4966--4981}.
\newblock


\bibitem[\protect\citeauthoryear{Fletcher and Islam}{Fletcher and Islam}{2019}]%
        {fletcher2019decision}
\bibfield{author}{\bibinfo{person}{Sam Fletcher} {and} \bibinfo{person}{Md~Zahidul Islam}.} \bibinfo{year}{2019}\natexlab{}.
\newblock \showarticletitle{Decision tree classification with differential privacy: A survey}.
\newblock \bibinfo{journal}{\emph{ACM Computing Surveys (CSUR)}} \bibinfo{volume}{52}, \bibinfo{number}{4} (\bibinfo{year}{2019}), \bibinfo{pages}{1--33}.
\newblock


\bibitem[\protect\citeauthoryear{Garcelon, Chaudhuri, Perchet, and Pirotta}{Garcelon et~al\mbox{.}}{2022}]%
        {garcelon2022privacy}
\bibfield{author}{\bibinfo{person}{Evrard Garcelon}, \bibinfo{person}{Kamalika Chaudhuri}, \bibinfo{person}{Vianney Perchet}, {and} \bibinfo{person}{Matteo Pirotta}.} \bibinfo{year}{2022}\natexlab{}.
\newblock \showarticletitle{Privacy amplification via shuffling for linear contextual bandits}. In \bibinfo{booktitle}{\emph{International Conference on Algorithmic Learning Theory}}. PMLR, \bibinfo{pages}{381--407}.
\newblock


\bibitem[\protect\citeauthoryear{Ghazi, Golowich, Kumar, Pagh, and Velingker}{Ghazi et~al\mbox{.}}{2021a}]%
        {ghazi2021power}
\bibfield{author}{\bibinfo{person}{Badih Ghazi}, \bibinfo{person}{Noah Golowich}, \bibinfo{person}{Ravi Kumar}, \bibinfo{person}{Rasmus Pagh}, {and} \bibinfo{person}{Ameya Velingker}.} \bibinfo{year}{2021}\natexlab{a}.
\newblock \showarticletitle{On the power of multiple anonymous messages: Frequency estimation and selection in the shuffle model of differential privacy}. In \bibinfo{booktitle}{\emph{Annual International Conference on the Theory and Applications of Cryptographic Techniques}}. Springer, \bibinfo{pages}{463--488}.
\newblock


\bibitem[\protect\citeauthoryear{Ghazi, Kumar, Manurangsi, and Pagh}{Ghazi et~al\mbox{.}}{2020a}]%
        {ghazi2020private}
\bibfield{author}{\bibinfo{person}{Badih Ghazi}, \bibinfo{person}{Ravi Kumar}, \bibinfo{person}{Pasin Manurangsi}, {and} \bibinfo{person}{Rasmus Pagh}.} \bibinfo{year}{2020}\natexlab{a}.
\newblock \showarticletitle{Private counting from anonymous messages: Near-optimal accuracy with vanishing communication overhead}. In \bibinfo{booktitle}{\emph{International Conference on Machine Learning}}. PMLR, \bibinfo{pages}{3505--3514}.
\newblock


\bibitem[\protect\citeauthoryear{Ghazi, Kumar, Manurangsi, Pagh, and Sinha}{Ghazi et~al\mbox{.}}{2021b}]%
        {ghazi2021differentially}
\bibfield{author}{\bibinfo{person}{Badih Ghazi}, \bibinfo{person}{Ravi Kumar}, \bibinfo{person}{Pasin Manurangsi}, \bibinfo{person}{Rasmus Pagh}, {and} \bibinfo{person}{Amer Sinha}.} \bibinfo{year}{2021}\natexlab{b}.
\newblock \showarticletitle{Differentially private aggregation in the shuffle model: Almost central accuracy in almost a single message}. In \bibinfo{booktitle}{\emph{International Conference on Machine Learning}}. PMLR, \bibinfo{pages}{3692--3701}.
\newblock


\bibitem[\protect\citeauthoryear{Ghazi, Manurangsi, Pagh, and Velingker}{Ghazi et~al\mbox{.}}{2020b}]%
        {ghazi2020fewer}
\bibfield{author}{\bibinfo{person}{Badih Ghazi}, \bibinfo{person}{Pasin Manurangsi}, \bibinfo{person}{Rasmus Pagh}, {and} \bibinfo{person}{Ameya Velingker}.} \bibinfo{year}{2020}\natexlab{b}.
\newblock \showarticletitle{Private aggregation from fewer anonymous messages}. In \bibinfo{booktitle}{\emph{Advances in Cryptology--EUROCRYPT 2020: 39th Annual International Conference on the Theory and Applications of Cryptographic Techniques, Zagreb, Croatia, May 10--14, 2020, Proceedings, Part II 30}}. Springer, \bibinfo{pages}{798--827}.
\newblock


\bibitem[\protect\citeauthoryear{Girgis, Data, and Diggavi}{Girgis et~al\mbox{.}}{2021a}]%
        {girgis2021subsampled}
\bibfield{author}{\bibinfo{person}{Antonious Girgis}, \bibinfo{person}{Deepesh Data}, {and} \bibinfo{person}{Suhas Diggavi}.} \bibinfo{year}{2021}\natexlab{a}.
\newblock \showarticletitle{Renyi differential privacy of the subsampled shuffle model in distributed learning}.
\newblock \bibinfo{journal}{\emph{Advances in Neural Information Processing Systems}}  \bibinfo{volume}{34} (\bibinfo{year}{2021}), \bibinfo{pages}{29181--29192}.
\newblock


\bibitem[\protect\citeauthoryear{Girgis, Data, Diggavi, Kairouz, and Suresh}{Girgis et~al\mbox{.}}{2021b}]%
        {girgis2021shuffled}
\bibfield{author}{\bibinfo{person}{Antonious Girgis}, \bibinfo{person}{Deepesh Data}, \bibinfo{person}{Suhas Diggavi}, \bibinfo{person}{Peter Kairouz}, {and} \bibinfo{person}{Ananda~Theertha Suresh}.} \bibinfo{year}{2021}\natexlab{b}.
\newblock \showarticletitle{Shuffled Model of Differential Privacy in Federated Learning}. In \bibinfo{booktitle}{\emph{International Conference on Artificial Intelligence and Statistics}}. PMLR, \bibinfo{pages}{2521--2529}.
\newblock


\bibitem[\protect\citeauthoryear{Girgis, Data, Diggavi, Kairouz, and Suresh}{Girgis et~al\mbox{.}}{2021c}]%
        {girgis2021shuffledtradeoffs}
\bibfield{author}{\bibinfo{person}{Antonious~M Girgis}, \bibinfo{person}{Deepesh Data}, \bibinfo{person}{Suhas Diggavi}, \bibinfo{person}{Peter Kairouz}, {and} \bibinfo{person}{Ananda~Theertha Suresh}.} \bibinfo{year}{2021}\natexlab{c}.
\newblock \showarticletitle{Shuffled model of federated learning: Privacy, accuracy and communication trade-offs}.
\newblock \bibinfo{journal}{\emph{IEEE journal on selected areas in information theory}} \bibinfo{volume}{2}, \bibinfo{number}{1} (\bibinfo{year}{2021}), \bibinfo{pages}{464--478}.
\newblock


\bibitem[\protect\citeauthoryear{Girgis, Data, Diggavi, Suresh, and Kairouz}{Girgis et~al\mbox{.}}{2021d}]%
        {girgis2021renyi}
\bibfield{author}{\bibinfo{person}{Antonious~M Girgis}, \bibinfo{person}{Deepesh Data}, \bibinfo{person}{Suhas Diggavi}, \bibinfo{person}{Ananda~Theertha Suresh}, {and} \bibinfo{person}{Peter Kairouz}.} \bibinfo{year}{2021}\natexlab{d}.
\newblock \showarticletitle{On the r{\'e}nyi differential privacy of the shuffle model}. In \bibinfo{booktitle}{\emph{Proceedings of the 2021 ACM SIGSAC Conference on Computer and Communications Security}}. \bibinfo{pages}{2321--2341}.
\newblock


\bibitem[\protect\citeauthoryear{Girgis and Diggavi}{Girgis and Diggavi}{2023}]%
        {girgis2023distributed}
\bibfield{author}{\bibinfo{person}{Antonious~M Girgis} {and} \bibinfo{person}{Suhas Diggavi}.} \bibinfo{year}{2023}\natexlab{}.
\newblock \showarticletitle{Distributed Mean Estimation for Multi-Message Shuffled Privacy}. In \bibinfo{booktitle}{\emph{Federated Learning and Analytics in Practice: Algorithms, Systems, Applications, and Opportunities}}.
\newblock


\bibitem[\protect\citeauthoryear{Gu, Li, Cheng, Xiong, and Cao}{Gu et~al\mbox{.}}{2020}]%
        {gu2019pckv}
\bibfield{author}{\bibinfo{person}{Xiaolan Gu}, \bibinfo{person}{Ming Li}, \bibinfo{person}{Yueqiang Cheng}, \bibinfo{person}{Li Xiong}, {and} \bibinfo{person}{Yang Cao}.} \bibinfo{year}{2020}\natexlab{}.
\newblock \showarticletitle{{PCKV}: Locally Differentially Private Correlated Key-Value Data Collection with Optimized Utility}.
\newblock \bibinfo{journal}{\emph{{USENIX} Security}} (\bibinfo{year}{2020}).
\newblock


\bibitem[\protect\citeauthoryear{Hirschman, Alba, and Farley}{Hirschman et~al\mbox{.}}{2000}]%
        {hirschman2000meaning}
\bibfield{author}{\bibinfo{person}{Charles Hirschman}, \bibinfo{person}{Richard Alba}, {and} \bibinfo{person}{Reynolds Farley}.} \bibinfo{year}{2000}\natexlab{}.
\newblock \showarticletitle{The meaning and measurement of race in the US census: Glimpses into the future}.
\newblock \bibinfo{journal}{\emph{Demography}} \bibinfo{volume}{37}, \bibinfo{number}{3} (\bibinfo{year}{2000}), \bibinfo{pages}{381--393}.
\newblock


\bibitem[\protect\citeauthoryear{Hu, Yuan, Yao, Deng, Chen, Yang, Guan, and Zeng}{Hu et~al\mbox{.}}{2015}]%
        {hu2015differential}
\bibfield{author}{\bibinfo{person}{Xueyang Hu}, \bibinfo{person}{Mingxuan Yuan}, \bibinfo{person}{Jianguo Yao}, \bibinfo{person}{Yu Deng}, \bibinfo{person}{Lei Chen}, \bibinfo{person}{Qiang Yang}, \bibinfo{person}{Haibing Guan}, {and} \bibinfo{person}{Jia Zeng}.} \bibinfo{year}{2015}\natexlab{}.
\newblock \showarticletitle{Differential privacy in telco big data platform}.
\newblock \bibinfo{journal}{\emph{Proceedings of the VLDB Endowment}} \bibinfo{volume}{8}, \bibinfo{number}{12} (\bibinfo{year}{2015}), \bibinfo{pages}{1692--1703}.
\newblock


\bibitem[\protect\citeauthoryear{Jagannathan, Pillaipakkamnatt, and Wright}{Jagannathan et~al\mbox{.}}{2009}]%
        {jagannathan2009practical}
\bibfield{author}{\bibinfo{person}{Geetha Jagannathan}, \bibinfo{person}{Krishnan Pillaipakkamnatt}, {and} \bibinfo{person}{Rebecca~N Wright}.} \bibinfo{year}{2009}\natexlab{}.
\newblock \showarticletitle{A practical differentially private random decision tree classifier}. In \bibinfo{booktitle}{\emph{2009 IEEE International Conference on Data Mining Workshops}}. IEEE, \bibinfo{pages}{114--121}.
\newblock


\bibitem[\protect\citeauthoryear{Johnson, Near, and Song}{Johnson et~al\mbox{.}}{2018}]%
        {johnson2018towards}
\bibfield{author}{\bibinfo{person}{Noah Johnson}, \bibinfo{person}{Joseph~P Near}, {and} \bibinfo{person}{Dawn Song}.} \bibinfo{year}{2018}\natexlab{}.
\newblock \showarticletitle{Towards practical differential privacy for SQL queries}.
\newblock \bibinfo{journal}{\emph{Proceedings of the VLDB Endowment}} \bibinfo{volume}{11}, \bibinfo{number}{5} (\bibinfo{year}{2018}), \bibinfo{pages}{526--539}.
\newblock


\bibitem[\protect\citeauthoryear{Kairouz, Bonawitz, and Ramage}{Kairouz et~al\mbox{.}}{2016}]%
        {kairouz2016discrete}
\bibfield{author}{\bibinfo{person}{Peter Kairouz}, \bibinfo{person}{Keith Bonawitz}, {and} \bibinfo{person}{Daniel Ramage}.} \bibinfo{year}{2016}\natexlab{}.
\newblock \showarticletitle{Discrete Distribution Estimation under Local Privacy}. In \bibinfo{booktitle}{\emph{International Conference on Machine Learning}}. \bibinfo{pages}{2436--2444}.
\newblock


\bibitem[\protect\citeauthoryear{Kairouz, Oh, and Viswanath}{Kairouz et~al\mbox{.}}{2014}]%
        {kairouz2014extremal}
\bibfield{author}{\bibinfo{person}{Peter Kairouz}, \bibinfo{person}{Sewoong Oh}, {and} \bibinfo{person}{Pramod Viswanath}.} \bibinfo{year}{2014}\natexlab{}.
\newblock \showarticletitle{Extremal mechanisms for local differential privacy}.
\newblock \bibinfo{journal}{\emph{Advances in neural information processing systems}}  \bibinfo{volume}{27} (\bibinfo{year}{2014}).
\newblock


\bibitem[\protect\citeauthoryear{Kasiviswanathan, Lee, Nissim, Raskhodnikova, and Smith}{Kasiviswanathan et~al\mbox{.}}{2011}]%
        {kasiviswanathan2011can}
\bibfield{author}{\bibinfo{person}{Shiva~Prasad Kasiviswanathan}, \bibinfo{person}{Homin~K Lee}, \bibinfo{person}{Kobbi Nissim}, \bibinfo{person}{Sofya Raskhodnikova}, {and} \bibinfo{person}{Adam Smith}.} \bibinfo{year}{2011}\natexlab{}.
\newblock \showarticletitle{What can we learn privately?}
\newblock \bibinfo{journal}{\emph{SIAM J. Comput.}} \bibinfo{volume}{40}, \bibinfo{number}{3} (\bibinfo{year}{2011}), \bibinfo{pages}{793--826}.
\newblock


\bibitem[\protect\citeauthoryear{Koskela, Heikkil{\"a}, and Honkela}{Koskela et~al\mbox{.}}{2022}]%
        {koskela2021tight}
\bibfield{author}{\bibinfo{person}{Antti Koskela}, \bibinfo{person}{Mikko~A Heikkil{\"a}}, {and} \bibinfo{person}{Antti Honkela}.} \bibinfo{year}{2022}\natexlab{}.
\newblock \showarticletitle{Numerical Accounting in the Shuffle Model of Differential Privacy}.
\newblock \bibinfo{journal}{\emph{Transactions on Machine Learning Research}} (\bibinfo{year}{2022}).
\newblock


\bibitem[\protect\citeauthoryear{Koskela, J{\"a}lk{\"o}, and Honkela}{Koskela et~al\mbox{.}}{2020}]%
        {koskela2020computing}
\bibfield{author}{\bibinfo{person}{Antti Koskela}, \bibinfo{person}{Joonas J{\"a}lk{\"o}}, {and} \bibinfo{person}{Antti Honkela}.} \bibinfo{year}{2020}\natexlab{}.
\newblock \showarticletitle{Computing tight differential privacy guarantees using fft}. In \bibinfo{booktitle}{\emph{International Conference on Artificial Intelligence and Statistics}}. PMLR, \bibinfo{pages}{2560--2569}.
\newblock


\bibitem[\protect\citeauthoryear{Kotsogiannis, Tao, He, Fanaeepour, Machanavajjhala, Hay, and Miklau}{Kotsogiannis et~al\mbox{.}}{2019}]%
        {kotsogiannis2019privatesql}
\bibfield{author}{\bibinfo{person}{Ios Kotsogiannis}, \bibinfo{person}{Yuchao Tao}, \bibinfo{person}{Xi He}, \bibinfo{person}{Maryam Fanaeepour}, \bibinfo{person}{Ashwin Machanavajjhala}, \bibinfo{person}{Michael Hay}, {and} \bibinfo{person}{Gerome Miklau}.} \bibinfo{year}{2019}\natexlab{}.
\newblock \showarticletitle{Privatesql: a differentially private sql query engine}.
\newblock \bibinfo{journal}{\emph{Proceedings of the VLDB Endowment}} \bibinfo{volume}{12}, \bibinfo{number}{11} (\bibinfo{year}{2019}), \bibinfo{pages}{1371--1384}.
\newblock


\bibitem[\protect\citeauthoryear{Li, Hay, Miklau, and Wang}{Li et~al\mbox{.}}{2014}]%
        {li2014data}
\bibfield{author}{\bibinfo{person}{Chao Li}, \bibinfo{person}{Michael Hay}, \bibinfo{person}{Gerome Miklau}, {and} \bibinfo{person}{Yue Wang}.} \bibinfo{year}{2014}\natexlab{}.
\newblock \showarticletitle{A Data-and Workload-Aware Algorithm for Range Queries Under Differential Privacy}.
\newblock \bibinfo{journal}{\emph{Proceedings of the VLDB Endowment}} \bibinfo{volume}{7}, \bibinfo{number}{5} (\bibinfo{year}{2014}).
\newblock


\bibitem[\protect\citeauthoryear{Li, Miklau, Hay, McGregor, and Rastogi}{Li et~al\mbox{.}}{2015}]%
        {li2015matrix}
\bibfield{author}{\bibinfo{person}{Chao Li}, \bibinfo{person}{Gerome Miklau}, \bibinfo{person}{Michael Hay}, \bibinfo{person}{Andrew McGregor}, {and} \bibinfo{person}{Vibhor Rastogi}.} \bibinfo{year}{2015}\natexlab{}.
\newblock \showarticletitle{The matrix mechanism: optimizing linear counting queries under differential privacy}.
\newblock \bibinfo{journal}{\emph{The VLDB journal}}  \bibinfo{volume}{24} (\bibinfo{year}{2015}), \bibinfo{pages}{757--781}.
\newblock


\bibitem[\protect\citeauthoryear{Li, Liu, Feng, Huang, Hu, Liu, Ren, and Qin}{Li et~al\mbox{.}}{2023}]%
        {li2023privacy}
\bibfield{author}{\bibinfo{person}{Xiaochen Li}, \bibinfo{person}{Weiran Liu}, \bibinfo{person}{Hanwen Feng}, \bibinfo{person}{Kunzhe Huang}, \bibinfo{person}{Yuke Hu}, \bibinfo{person}{Jinfei Liu}, \bibinfo{person}{Kui Ren}, {and} \bibinfo{person}{Zhan Qin}.} \bibinfo{year}{2023}\natexlab{}.
\newblock \showarticletitle{Privacy enhancement via dummy points in the shuffle model}.
\newblock \bibinfo{journal}{\emph{IEEE Transactions on Dependable and Secure Computing}} (\bibinfo{year}{2023}).
\newblock


\bibitem[\protect\citeauthoryear{Liu, Lou, Xiong, Liu, and Meng}{Liu et~al\mbox{.}}{2021}]%
        {liu2021projected}
\bibfield{author}{\bibinfo{person}{Junxu Liu}, \bibinfo{person}{Jian Lou}, \bibinfo{person}{Li Xiong}, \bibinfo{person}{Jinfei Liu}, {and} \bibinfo{person}{Xiaofeng Meng}.} \bibinfo{year}{2021}\natexlab{}.
\newblock \showarticletitle{Projected federated averaging with heterogeneous differential privacy}.
\newblock \bibinfo{journal}{\emph{Proceedings of the VLDB Endowment}} \bibinfo{volume}{15}, \bibinfo{number}{4} (\bibinfo{year}{2021}), \bibinfo{pages}{828--840}.
\newblock


\bibitem[\protect\citeauthoryear{Liu, Musen, and Chou}{Liu et~al\mbox{.}}{2015}]%
        {liu2015data}
\bibfield{author}{\bibinfo{person}{Vincent Liu}, \bibinfo{person}{Mark~A Musen}, {and} \bibinfo{person}{Timothy Chou}.} \bibinfo{year}{2015}\natexlab{}.
\newblock \showarticletitle{Data breaches of protected health information in the United States}.
\newblock \bibinfo{journal}{\emph{JAMA}} \bibinfo{volume}{313}, \bibinfo{number}{14} (\bibinfo{year}{2015}), \bibinfo{pages}{1471--1473}.
\newblock


\bibitem[\protect\citeauthoryear{Luo, Wang, and Yi}{Luo et~al\mbox{.}}{2022}]%
        {luo2022frequency}
\bibfield{author}{\bibinfo{person}{Qiyao Luo}, \bibinfo{person}{Yilei Wang}, {and} \bibinfo{person}{Ke Yi}.} \bibinfo{year}{2022}\natexlab{}.
\newblock \showarticletitle{Frequency Estimation in the Shuffle Model with Almost a Single Message}. In \bibinfo{booktitle}{\emph{Proceedings of the 2022 ACM SIGSAC Conference on Computer and Communications Security}}. \bibinfo{pages}{2219--2232}.
\newblock


\bibitem[\protect\citeauthoryear{McSherry and Talwar}{McSherry and Talwar}{2007}]%
        {mcsherry2007mechanism}
\bibfield{author}{\bibinfo{person}{Frank McSherry} {and} \bibinfo{person}{Kunal Talwar}.} \bibinfo{year}{2007}\natexlab{}.
\newblock \showarticletitle{Mechanism Design via Differential Privacy.}
\newblock \bibinfo{journal}{\emph{FOCS}} (\bibinfo{year}{2007}).
\newblock


\bibitem[\protect\citeauthoryear{Mironov}{Mironov}{2017}]%
        {mironov2017renyi}
\bibfield{author}{\bibinfo{person}{Ilya Mironov}.} \bibinfo{year}{2017}\natexlab{}.
\newblock \showarticletitle{R{\'e}nyi differential privacy}. In \bibinfo{booktitle}{\emph{2017 IEEE 30th computer security foundations symposium (CSF)}}. IEEE, \bibinfo{pages}{263--275}.
\newblock


\bibitem[\protect\citeauthoryear{Mironov, Pandey, Reingold, and Vadhan}{Mironov et~al\mbox{.}}{2009}]%
        {mironov2009computational}
\bibfield{author}{\bibinfo{person}{Ilya Mironov}, \bibinfo{person}{Omkant Pandey}, \bibinfo{person}{Omer Reingold}, {and} \bibinfo{person}{Salil Vadhan}.} \bibinfo{year}{2009}\natexlab{}.
\newblock \showarticletitle{Computational differential privacy}. In \bibinfo{booktitle}{\emph{Advances in Cryptology-CRYPTO 2009: 29th Annual International Cryptology Conference, Santa Barbara, CA, USA, August 16-20, 2009. Proceedings}}. Springer, \bibinfo{pages}{126--142}.
\newblock


\bibitem[\protect\citeauthoryear{Nguy{\^e}n, Xiao, Yang, Hui, Shin, and Shin}{Nguy{\^e}n et~al\mbox{.}}{2016}]%
        {nguyen2016collecting}
\bibfield{author}{\bibinfo{person}{Th{\^o}ng~T Nguy{\^e}n}, \bibinfo{person}{Xiaokui Xiao}, \bibinfo{person}{Yin Yang}, \bibinfo{person}{Siu~Cheung Hui}, \bibinfo{person}{Hyejin Shin}, {and} \bibinfo{person}{Junbum Shin}.} \bibinfo{year}{2016}\natexlab{}.
\newblock \showarticletitle{Collecting and analyzing data from smart device users with local differential privacy}.
\newblock \bibinfo{journal}{\emph{arXiv preprint arXiv:1606.05053}} (\bibinfo{year}{2016}).
\newblock


\bibitem[\protect\citeauthoryear{Paris}{Paris}{2010}]%
        {paris2010incomplete}
\bibfield{author}{\bibinfo{person}{RB Paris}.} \bibinfo{year}{2010}\natexlab{}.
\newblock \showarticletitle{Incomplete beta functions}.
\newblock \bibinfo{journal}{\emph{NIST Handbook of Mathematical Functions,, ed. Frank WJ Olver, Daniel M. Lozier, Ronald F. Boisvert and Charles W. Clark. Cambridge University Press. http://dlmf. nist. gov/8.17}} (\bibinfo{year}{2010}).
\newblock


\bibitem[\protect\citeauthoryear{Qin, Yang, Yu, Khalil, Xiao, and Ren}{Qin et~al\mbox{.}}{2016}]%
        {qin2016heavy}
\bibfield{author}{\bibinfo{person}{Zhan Qin}, \bibinfo{person}{Yin Yang}, \bibinfo{person}{Ting Yu}, \bibinfo{person}{Issa Khalil}, \bibinfo{person}{Xiaokui Xiao}, {and} \bibinfo{person}{Kui Ren}.} \bibinfo{year}{2016}\natexlab{}.
\newblock \showarticletitle{Heavy hitter estimation over set-valued data with local differential privacy}. In \bibinfo{booktitle}{\emph{Proceedings of the 2016 ACM SIGSAC Conference on Computer and Communications Security}}. ACM, \bibinfo{pages}{192--203}.
\newblock


\bibitem[\protect\citeauthoryear{Raghupathi and Raghupathi}{Raghupathi and Raghupathi}{2014}]%
        {raghupathi2014big}
\bibfield{author}{\bibinfo{person}{Wullianallur Raghupathi} {and} \bibinfo{person}{Viju Raghupathi}.} \bibinfo{year}{2014}\natexlab{}.
\newblock \showarticletitle{Big data analytics in healthcare: Promise and potential}.
\newblock \bibinfo{journal}{\emph{Health information science and systems}}  \bibinfo{volume}{2} (\bibinfo{year}{2014}), \bibinfo{pages}{1--10}.
\newblock


\bibitem[\protect\citeauthoryear{Ren, Shi, Yu, Yang, Zhao, and Xu}{Ren et~al\mbox{.}}{2022}]%
        {ren2022ldp}
\bibfield{author}{\bibinfo{person}{Xuebin Ren}, \bibinfo{person}{Liang Shi}, \bibinfo{person}{Weiren Yu}, \bibinfo{person}{Shusen Yang}, \bibinfo{person}{Cong Zhao}, {and} \bibinfo{person}{Zongben Xu}.} \bibinfo{year}{2022}\natexlab{}.
\newblock \showarticletitle{LDP-IDS: Local differential privacy for infinite data streams}. In \bibinfo{booktitle}{\emph{Proceedings of the 2022 international conference on management of data}}. \bibinfo{pages}{1064--1077}.
\newblock


\bibitem[\protect\citeauthoryear{Ren, Yu, Yu, Yang, Yang, McCann, and Philip}{Ren et~al\mbox{.}}{2018}]%
        {ren2018lopub}
\bibfield{author}{\bibinfo{person}{Xuebin Ren}, \bibinfo{person}{Chia-Mu Yu}, \bibinfo{person}{Weiren Yu}, \bibinfo{person}{Shusen Yang}, \bibinfo{person}{Xinyu Yang}, \bibinfo{person}{Julie~A McCann}, {and} \bibinfo{person}{S~Yu Philip}.} \bibinfo{year}{2018}\natexlab{}.
\newblock \showarticletitle{LoPub: high-dimensional crowdsourced data publication with local differential privacy}.
\newblock \bibinfo{journal}{\emph{IEEE Transactions on Information Forensics and Security}} \bibinfo{volume}{13}, \bibinfo{number}{9} (\bibinfo{year}{2018}), \bibinfo{pages}{2151--2166}.
\newblock


\bibitem[\protect\citeauthoryear{Tenenbaum, Kaplan, Mansour, and Stemmer}{Tenenbaum et~al\mbox{.}}{2021}]%
        {tenenbaum2021differentially}
\bibfield{author}{\bibinfo{person}{Jay Tenenbaum}, \bibinfo{person}{Haim Kaplan}, \bibinfo{person}{Yishay Mansour}, {and} \bibinfo{person}{Uri Stemmer}.} \bibinfo{year}{2021}\natexlab{}.
\newblock \showarticletitle{Differentially private multi-armed bandits in the shuffle model}.
\newblock \bibinfo{journal}{\emph{Advances in Neural Information Processing Systems}}  \bibinfo{volume}{34} (\bibinfo{year}{2021}), \bibinfo{pages}{24956--24967}.
\newblock


\bibitem[\protect\citeauthoryear{Truex, Liu, Chow, Gursoy, and Wei}{Truex et~al\mbox{.}}{2020}]%
        {truex2020ldp}
\bibfield{author}{\bibinfo{person}{Stacey Truex}, \bibinfo{person}{Ling Liu}, \bibinfo{person}{Ka-Ho Chow}, \bibinfo{person}{Mehmet~Emre Gursoy}, {and} \bibinfo{person}{Wenqi Wei}.} \bibinfo{year}{2020}\natexlab{}.
\newblock \showarticletitle{LDP-Fed: Federated learning with local differential privacy}. In \bibinfo{booktitle}{\emph{Proceedings of the Third ACM International Workshop on Edge Systems, Analytics and Networking}}. \bibinfo{pages}{61--66}.
\newblock


\bibitem[\protect\citeauthoryear{Voigt and Von~dem Bussche}{Voigt and Von~dem Bussche}{2017}]%
        {voigt2017eu}
\bibfield{author}{\bibinfo{person}{Paul Voigt} {and} \bibinfo{person}{Axel Von~dem Bussche}.} \bibinfo{year}{2017}\natexlab{}.
\newblock \showarticletitle{The eu general data protection regulation (gdpr)}.
\newblock \bibinfo{journal}{\emph{A Practical Guide, 1st Ed., Cham: Springer International Publishing}} \bibinfo{volume}{10}, \bibinfo{number}{3152676} (\bibinfo{year}{2017}), \bibinfo{pages}{10--5555}.
\newblock


\bibitem[\protect\citeauthoryear{Wang, Gaboardi, and Xu}{Wang et~al\mbox{.}}{2018a}]%
        {wang2018empirical}
\bibfield{author}{\bibinfo{person}{Di Wang}, \bibinfo{person}{Marco Gaboardi}, {and} \bibinfo{person}{Jinhui Xu}.} \bibinfo{year}{2018}\natexlab{a}.
\newblock \showarticletitle{Empirical risk minimization in non-interactive local differential privacy revisited}.
\newblock \bibinfo{journal}{\emph{Advances in Neural Information Processing Systems}}  \bibinfo{volume}{31} (\bibinfo{year}{2018}).
\newblock


\bibitem[\protect\citeauthoryear{Wang and Xu}{Wang and Xu}{2019}]%
        {wang2019sparse}
\bibfield{author}{\bibinfo{person}{Di Wang} {and} \bibinfo{person}{Jinhui Xu}.} \bibinfo{year}{2019}\natexlab{}.
\newblock \showarticletitle{On sparse linear regression in the local differential privacy model}. In \bibinfo{booktitle}{\emph{International Conference on Machine Learning}}. PMLR, \bibinfo{pages}{6628--6637}.
\newblock


\bibitem[\protect\citeauthoryear{Wang, Xiao, Yang, Zhao, Hui, Shin, Shin, and Yu}{Wang et~al\mbox{.}}{2019d}]%
        {wang2019collecting}
\bibfield{author}{\bibinfo{person}{Ning Wang}, \bibinfo{person}{Xiaokui Xiao}, \bibinfo{person}{Yin Yang}, \bibinfo{person}{Jun Zhao}, \bibinfo{person}{Siu~Cheung Hui}, \bibinfo{person}{Hyejin Shin}, \bibinfo{person}{Junbum Shin}, {and} \bibinfo{person}{Ge Yu}.} \bibinfo{year}{2019}\natexlab{d}.
\newblock \showarticletitle{Collecting and analyzing multidimensional data with local differential privacy}. In \bibinfo{booktitle}{\emph{2019 IEEE 35th International Conference on Data Engineering (ICDE)}}. IEEE, \bibinfo{pages}{638--649}.
\newblock


\bibitem[\protect\citeauthoryear{Wang, Huang, Nie, Wang, Xu, and Yang}{Wang et~al\mbox{.}}{2018b}]%
        {wang2018privset}
\bibfield{author}{\bibinfo{person}{Shaowei Wang}, \bibinfo{person}{Liusheng Huang}, \bibinfo{person}{Yiwen Nie}, \bibinfo{person}{Pengzhan Wang}, \bibinfo{person}{Hongli Xu}, {and} \bibinfo{person}{Wei Yang}.} \bibinfo{year}{2018}\natexlab{b}.
\newblock \showarticletitle{PrivSet: Set-Valued Data Analyses with Locale Differential Privacy}. In \bibinfo{booktitle}{\emph{IEEE INFOCOM 2018-IEEE Conference on Computer Communications}}. IEEE, \bibinfo{pages}{1088--1096}.
\newblock


\bibitem[\protect\citeauthoryear{Wang, Huang, Nie, Zhang, Wang, Xu, and Yang}{Wang et~al\mbox{.}}{2019b}]%
        {wang2019local}
\bibfield{author}{\bibinfo{person}{Shaowei Wang}, \bibinfo{person}{Liusheng Huang}, \bibinfo{person}{Yiwen Nie}, \bibinfo{person}{Xinyuan Zhang}, \bibinfo{person}{Pengzhan Wang}, \bibinfo{person}{Hongli Xu}, {and} \bibinfo{person}{Wei Yang}.} \bibinfo{year}{2019}\natexlab{b}.
\newblock \showarticletitle{Local differential private data aggregation for discrete distribution estimation}.
\newblock \bibinfo{journal}{\emph{IEEE Transactions on Parallel and Distributed Systems}} \bibinfo{volume}{30}, \bibinfo{number}{9} (\bibinfo{year}{2019}), \bibinfo{pages}{2046--2059}.
\newblock


\bibitem[\protect\citeauthoryear{Wang, Li, Li, Luo, Yang, Yan, and Dong}{Wang et~al\mbox{.}}{2023a}]%
        {wang2023SMDP}
\bibfield{author}{\bibinfo{person}{Shaowei Wang}, \bibinfo{person}{Jin Li}, \bibinfo{person}{Yuntong Li}, \bibinfo{person}{Xuandi Luo}, \bibinfo{person}{Wei Yang}, \bibinfo{person}{Hongyang Yan}, {and} \bibinfo{person}{Changyu Dong}.} \bibinfo{year}{2023}\natexlab{a}.
\newblock \showarticletitle{Metric differential privacy in the shuffle model}.
\newblock \bibinfo{journal}{\emph{(to be released)}} (\bibinfo{year}{2023}).
\newblock


\bibitem[\protect\citeauthoryear{Wang, Li, Qian, Du, Lin, and Yang}{Wang et~al\mbox{.}}{2021}]%
        {wang2021hiding}
\bibfield{author}{\bibinfo{person}{Shaowei Wang}, \bibinfo{person}{Jin Li}, \bibinfo{person}{Yuqiu Qian}, \bibinfo{person}{Jiachun Du}, \bibinfo{person}{Wenqing Lin}, {and} \bibinfo{person}{Wei Yang}.} \bibinfo{year}{2021}\natexlab{}.
\newblock \showarticletitle{Hiding Numerical Vectors in Local Private and Shuffled Messages.}. In \bibinfo{booktitle}{\emph{IJCAI}}. \bibinfo{pages}{3706--3712}.
\newblock


\bibitem[\protect\citeauthoryear{Wang, Luo, Qian, Du, Lin, and Yang}{Wang et~al\mbox{.}}{2022}]%
        {wang2022analyzing}
\bibfield{author}{\bibinfo{person}{Shaowei Wang}, \bibinfo{person}{Xuandi Luo}, \bibinfo{person}{Yuqiu Qian}, \bibinfo{person}{Jiachun Du}, \bibinfo{person}{Wenqing Lin}, {and} \bibinfo{person}{Wei Yang}.} \bibinfo{year}{2022}\natexlab{}.
\newblock \showarticletitle{Analyzing Preference Data With Local Privacy: Optimal Utility and Enhanced Robustness}.
\newblock \bibinfo{journal}{\emph{IEEE Transactions on Knowledge and Data Engineering}} (\bibinfo{year}{2022}).
\newblock


\bibitem[\protect\citeauthoryear{Wang, Luo, Qian, Zhu, Chen, Chen, Xin, and Yang}{Wang et~al\mbox{.}}{2023b}]%
        {wang2023shuffle}
\bibfield{author}{\bibinfo{person}{Shaowei Wang}, \bibinfo{person}{Xuandi Luo}, \bibinfo{person}{Yuqiu Qian}, \bibinfo{person}{Youwen Zhu}, \bibinfo{person}{Kongyang Chen}, \bibinfo{person}{Qi Chen}, \bibinfo{person}{Bangzhou Xin}, {and} \bibinfo{person}{Wei Yang}.} \bibinfo{year}{2023}\natexlab{b}.
\newblock \showarticletitle{Shuffle Differential Private Data Aggregation for Random Population}.
\newblock \bibinfo{journal}{\emph{IEEE Transactions on Parallel and Distributed Systems}} (\bibinfo{year}{2023}).
\newblock


\bibitem[\protect\citeauthoryear{Wang, Qian, Du, Yang, Huang, and Xu}{Wang et~al\mbox{.}}{2020}]%
        {wang2020set}
\bibfield{author}{\bibinfo{person}{Shaowei Wang}, \bibinfo{person}{Yuqiu Qian}, \bibinfo{person}{Jiachun Du}, \bibinfo{person}{Wei Yang}, \bibinfo{person}{Liusheng Huang}, {and} \bibinfo{person}{Hongli Xu}.} \bibinfo{year}{2020}\natexlab{}.
\newblock \showarticletitle{Set-valued data publication with local privacy: tight error bounds and efficient mechanisms}.
\newblock \bibinfo{journal}{\emph{Proceedings of the VLDB Endowment}} \bibinfo{volume}{13}, \bibinfo{number}{8} (\bibinfo{year}{2020}), \bibinfo{pages}{1234--1247}.
\newblock


\bibitem[\protect\citeauthoryear{Wang, Blocki, Li, and Jha}{Wang et~al\mbox{.}}{2017}]%
        {wang2017locally}
\bibfield{author}{\bibinfo{person}{Tianhao Wang}, \bibinfo{person}{Jeremiah Blocki}, \bibinfo{person}{Ninghui Li}, {and} \bibinfo{person}{Somesh Jha}.} \bibinfo{year}{2017}\natexlab{}.
\newblock \showarticletitle{Locally differentially private protocols for frequency estimation}. In \bibinfo{booktitle}{\emph{Proc. of the 26th USENIX Security Symposium}}. \bibinfo{pages}{729--745}.
\newblock


\bibitem[\protect\citeauthoryear{Wang, Ding, Zhou, Hong, Huang, Li, and Jha}{Wang et~al\mbox{.}}{2019a}]%
        {wang2019answering}
\bibfield{author}{\bibinfo{person}{Tianhao Wang}, \bibinfo{person}{Bolin Ding}, \bibinfo{person}{Jingren Zhou}, \bibinfo{person}{Cheng Hong}, \bibinfo{person}{Zhicong Huang}, \bibinfo{person}{Ninghui Li}, {and} \bibinfo{person}{Somesh Jha}.} \bibinfo{year}{2019}\natexlab{a}.
\newblock \showarticletitle{Answering multi-dimensional analytical queries under local differential privacy}. In \bibinfo{booktitle}{\emph{Proceedings of the 2019 International Conference on Management of Data}}. \bibinfo{pages}{159--176}.
\newblock


\bibitem[\protect\citeauthoryear{Wang, Li, and Jha}{Wang et~al\mbox{.}}{2018c}]%
        {wang2018locally}
\bibfield{author}{\bibinfo{person}{Tianhao Wang}, \bibinfo{person}{Ninghui Li}, {and} \bibinfo{person}{Somesh Jha}.} \bibinfo{year}{2018}\natexlab{c}.
\newblock \showarticletitle{Locally differentially private frequent itemset mining}. In \bibinfo{booktitle}{\emph{2018 IEEE Symposium on Security and Privacy (SP)}}. IEEE, \bibinfo{pages}{127--143}.
\newblock


\bibitem[\protect\citeauthoryear{Wang, Li, and Jha}{Wang et~al\mbox{.}}{2019c}]%
        {wang2019locally}
\bibfield{author}{\bibinfo{person}{Tianhao Wang}, \bibinfo{person}{Ninghui Li}, {and} \bibinfo{person}{Somesh Jha}.} \bibinfo{year}{2019}\natexlab{c}.
\newblock \showarticletitle{Locally differentially private heavy hitter identification}.
\newblock \bibinfo{journal}{\emph{IEEE Transactions on Dependable and Secure Computing}} \bibinfo{volume}{18}, \bibinfo{number}{2} (\bibinfo{year}{2019}), \bibinfo{pages}{982--993}.
\newblock


\bibitem[\protect\citeauthoryear{Wang, Yang, Ren, Yu, and Yang}{Wang et~al\mbox{.}}{2019e}]%
        {wangT2019locally}
\bibfield{author}{\bibinfo{person}{Teng Wang}, \bibinfo{person}{Xinyu Yang}, \bibinfo{person}{Xuebin Ren}, \bibinfo{person}{Wei Yu}, {and} \bibinfo{person}{Shusen Yang}.} \bibinfo{year}{2019}\natexlab{e}.
\newblock \showarticletitle{Locally private high-dimensional crowdsourced data release based on copula functions}.
\newblock \bibinfo{journal}{\emph{IEEE Transactions on Services Computing}} \bibinfo{volume}{15}, \bibinfo{number}{2} (\bibinfo{year}{2019}), \bibinfo{pages}{778--792}.
\newblock


\bibitem[\protect\citeauthoryear{Warner}{Warner}{1965}]%
        {warner1965randomized}
\bibfield{author}{\bibinfo{person}{Stanley~L Warner}.} \bibinfo{year}{1965}\natexlab{}.
\newblock \showarticletitle{Randomized response: A survey technique for eliminating evasive answer bias}.
\newblock \bibinfo{journal}{\emph{J. Amer. Statist. Assoc.}} \bibinfo{volume}{60}, \bibinfo{number}{309} (\bibinfo{year}{1965}), \bibinfo{pages}{63--69}.
\newblock


\bibitem[\protect\citeauthoryear{Xin, Yang, Wang, and Huang}{Xin et~al\mbox{.}}{2019}]%
        {xin2019differentially}
\bibfield{author}{\bibinfo{person}{Bangzhou Xin}, \bibinfo{person}{Wei Yang}, \bibinfo{person}{Shaowei Wang}, {and} \bibinfo{person}{Liusheng Huang}.} \bibinfo{year}{2019}\natexlab{}.
\newblock \showarticletitle{Differentially private greedy decision forest}. In \bibinfo{booktitle}{\emph{ICASSP 2019-2019 IEEE International Conference on Acoustics, Speech and Signal Processing (ICASSP)}}. IEEE, \bibinfo{pages}{2672--2676}.
\newblock


\bibitem[\protect\citeauthoryear{Xu, Ding, Wang, and Zhou}{Xu et~al\mbox{.}}{2020}]%
        {xu2020collecting}
\bibfield{author}{\bibinfo{person}{Min Xu}, \bibinfo{person}{Bolin Ding}, \bibinfo{person}{Tianhao Wang}, {and} \bibinfo{person}{Jingren Zhou}.} \bibinfo{year}{2020}\natexlab{}.
\newblock \showarticletitle{Collecting and analyzing data jointly from multiple services under local differential privacy}.
\newblock \bibinfo{journal}{\emph{Proceedings of the VLDB Endowment}} \bibinfo{volume}{13}, \bibinfo{number}{12} (\bibinfo{year}{2020}), \bibinfo{pages}{2760--2772}.
\newblock


\bibitem[\protect\citeauthoryear{Xu, Wang, Ding, Zhou, Hong, and Huang}{Xu et~al\mbox{.}}{2019}]%
        {xu2019dpsaas}
\bibfield{author}{\bibinfo{person}{Min Xu}, \bibinfo{person}{Tianhao Wang}, \bibinfo{person}{Bolin Ding}, \bibinfo{person}{Jingren Zhou}, \bibinfo{person}{Cheng Hong}, {and} \bibinfo{person}{Zhicong Huang}.} \bibinfo{year}{2019}\natexlab{}.
\newblock \showarticletitle{Dpsaas: Multi-dimensional data sharing and analytics as services under local differential privacy}.
\newblock \bibinfo{journal}{\emph{Proceedings of the VLDB Endowment}} \bibinfo{volume}{12}, \bibinfo{number}{12} (\bibinfo{year}{2019}), \bibinfo{pages}{1862--1865}.
\newblock


\bibitem[\protect\citeauthoryear{Ye and Barg}{Ye and Barg}{2018}]%
        {ye2018optimal}
\bibfield{author}{\bibinfo{person}{Min Ye} {and} \bibinfo{person}{Alexander Barg}.} \bibinfo{year}{2018}\natexlab{}.
\newblock \showarticletitle{Optimal schemes for discrete distribution estimation under locally differential privacy}.
\newblock \bibinfo{journal}{\emph{IEEE Transactions on Information Theory}} (\bibinfo{year}{2018}).
\newblock


\bibitem[\protect\citeauthoryear{Zhang, Wang, Li, He, and Chen}{Zhang et~al\mbox{.}}{2018}]%
        {zhang2018calm}
\bibfield{author}{\bibinfo{person}{Zhikun Zhang}, \bibinfo{person}{Tianhao Wang}, \bibinfo{person}{Ninghui Li}, \bibinfo{person}{Shibo He}, {and} \bibinfo{person}{Jiming Chen}.} \bibinfo{year}{2018}\natexlab{}.
\newblock \showarticletitle{Calm: Consistent adaptive local marginal for marginal release under local differential privacy}. In \bibinfo{booktitle}{\emph{Proceedings of the 2018 ACM SIGSAC Conference on Computer and Communications Security}}. \bibinfo{pages}{212--229}.
\newblock


\bibitem[\protect\citeauthoryear{Zhao and Chen}{Zhao and Chen}{2022}]%
        {zhao2022survey}
\bibfield{author}{\bibinfo{person}{Ying Zhao} {and} \bibinfo{person}{Jinjun Chen}.} \bibinfo{year}{2022}\natexlab{}.
\newblock \showarticletitle{A survey on differential privacy for unstructured data content}.
\newblock \bibinfo{journal}{\emph{ACM Computing Surveys (CSUR)}} \bibinfo{volume}{54}, \bibinfo{number}{10s} (\bibinfo{year}{2022}), \bibinfo{pages}{1--28}.
\newblock


\bibitem[\protect\citeauthoryear{Zhu, Dong, and Wang}{Zhu et~al\mbox{.}}{2022}]%
        {zhu2022optimal}
\bibfield{author}{\bibinfo{person}{Yuqing Zhu}, \bibinfo{person}{Jinshuo Dong}, {and} \bibinfo{person}{Yu-Xiang Wang}.} \bibinfo{year}{2022}\natexlab{}.
\newblock \showarticletitle{Optimal accounting of differential privacy via characteristic function}. In \bibinfo{booktitle}{\emph{International Conference on Artificial Intelligence and Statistics}}. PMLR, \bibinfo{pages}{4782--4817}.
\newblock


\end{thebibliography}

\newpage
\onecolumn
\appendix
\setcounter{page}{1}

\section{Proof of Lemma \ref{lemma:mixture}}
Let $R_1^0$, $R_1^1$, and $R_i$ denote the probability distributions of the random variable $\mech{R}_1(x_1^0)$, $\mech{R}_1(x_1^1)$, and $\mech{R}_i(x_i)$, respectively. We establish the existence of mixture distributions through constructive proof. Specifically, for any $y\in \dom{Y}$ in the output domain, we define the probability distributions of $\mathcal{Q}_1^0$, $\mathcal{Q}_1^1$, $\mathcal{Q}_1$, $\mathcal{Q}_2$, $\ldots$, and $\mathcal{Q}_n$ as follows:
\begin{align*}
&& \mathcal{Q}_1^0[y] &=\left\{
    \begin{array}{@{}lr@{}}
        \frac{R_1^0[y]-R_1^1[y]}{(p-1)\alpha}, & \text{if } R_1^0[y]> R_1^1[y];\\
        0, & \text{else,}
    \end{array}
    \right. & \\
&&\mathcal{Q}_1^1[y] &=\left\{
    \begin{array}{@{}lr@{}}
        \frac{R_1^1[y]-R_1^0[y]}{(p-1)\alpha}, & \text{if } R_1^0[y]< R_1^1[y];\\
        0, & \text{else,}
    \end{array}
    \right. & \\ 
&&\mathcal{Q}_1[y] &= 
        \frac{\min\{R_1^1[y],R_1^0[y]\}}{1-\alpha-p\alpha}-\frac{ |R_1^0[y]-R_1^1[y]|}{(p-1)(1-\alpha-p\alpha)}, \\
&&\mathcal{Q}_i[y] &=
        \frac{R_i[y]-r(\mathcal{Q}_1^0[y]+\mathcal{Q}_1^1[y])}{1-2r}. &
\end{align*}
We first demonstrate the validity of probability distributions $\mathcal{Q}_1^0$ and $\mathcal{Q}_1^1$. Specifically, for $\mathcal{Q}_1^0$, it can be observed that $\mathcal{Q}_1^0[y]$ is non-negative for all $y\in\dom{Y}$, and $\sum_{y\in \dom{Y}} \mathcal{Q}_1^0[y] = \frac{D_{1}(\mech{R}_1(x_1^0)| \mech{R}_1(x_1^1))}{(p-1)\alpha}=\frac{\beta'}{\beta'}=1$, thereby verifying it as a valid probability distribution. Likewise, we establish $\mathcal{Q}_1^1$ as a valid distribution.

Furthermore, we prove that $\mathcal{Q}_1$ is a valid distribution, and we demonstrate that Equations \ref{eq:mix1} and \ref{eq:mix2} hold. Applying the $(p, \beta')$-variation property, we obtain $\max\{R_1^0[y], R_1^1[y]\}\leq p \cdot\min\{R_1^0[y], R_1^1[y]\}$, thus verifying that $\mathcal{Q}_1[y]$ is non-negative. Additionally, since $p \alpha+\alpha+(1-\alpha-p \alpha)=1$, it follows that $$\sum_{y\in \dom{Y}} \mathcal{Q}_1[y]= \frac{\sum_{y\in \dom{Y}}R_1^0[y]-p \alpha \mathcal{Q}_1^0[y] - \alpha \mathcal{Q}_1^1[y]}{(1-\alpha-p \alpha)}=1,$$ indicating that $\mathcal{Q}_1$ is indeed a valid distribution.

Next, we demonstrate that Equation \ref{eq:mix1} holds. When $R_1^0[y]\geq R_1^1[y]$, we have $p \alpha \mathcal{Q}_1^0[y] + \alpha \mathcal{Q}_1^1[y]+(1-\alpha-p \alpha)\mathcal{Q}_1[y]=(R_1^0[y]-R_1^1[y])p/(p-1)+R_1^1[y]-(R_1^0[y]-R_1^1[y])/(p-1)=R_1^0[y]$. When $R_1^0[y]< R_1^1[y]$, we have $p \alpha \mathcal{Q}_1^0[y] + \alpha \mathcal{Q}_1^1[y]+(1-\alpha-p \alpha)\mathcal{Q}_1[y]=(R_1^1[y]-R_1^0[y])/(p-1)+R_1^0[y]-(R_1^1[y]-R_1^0[y])/(p-1)=R_1^0[y]$. Combining these two cases, we obtain Equation \ref{eq:mix1}. Similarly, we can show that Equation \ref{eq:mix2} holds by symmetry.

Finally, we demonstrate the validity of $\mathcal{Q}_i$ and the satisfaction of Equation \ref{eq:mix3}. By invoking the $q$-ratio property of $R_i$, we obtain $R_i[y]\geq \max\{R_1^0[y], R_1^1[y]\}/q$. Consequently, we obtain $R_i[y]\geq \max\{p\alpha\mathcal{Q}_1^0[y], p\alpha\mathcal{Q}_1^1[y]\}/q$. Notably, we can infer that $\mathcal{Q}_1^0[y]$ and $\mathcal{Q}_1^1[y]$ are never simultaneously greater than 0. As such, we conclude that $R_i[y]\geq p\alpha(\mathcal{Q}_1^0[y]+\mathcal{Q}_1^1[y])/q\geq r(\mathcal{Q}_1^0[y]+\mathcal{Q}_1^1[y])$, thereby leading to $\mathcal{Q}_i[y]$ being non-negative. Furthermore, we utilize the fact that $r+r+(1-2r)=1$, which enables us to establish that $\mathcal{Q}_i$ is a valid distribution. For the right-hand side of Equation \ref{eq:mix3}, if $\mathcal{Q}_1^0[y]\geq \mathcal{Q}_1^1[y]$, we obtain $r\mathcal{Q}_1^0[y] + r \mathcal{Q}_1^1[y]+(1-2r)\mathcal{Q}_i[y]=r\mathcal{Q}_1^0[y]+R_i[y]-r\mathcal{Q}_1^0[y]=R_i[y]$. Alternatively, if $\mathcal{Q}_1^0[y]< \mathcal{Q}_1^1[y]$, we obtain $r\mathcal{Q}_1^0[y] + r \mathcal{Q}_1^1[y]+(1-2r)\mathcal{Q}_i[y]=r\mathcal{Q}_1^1[y]+R_i[y]-r\mathcal{Q}_1^1[y]=R_i[y]$. Combining these two conditional results establishes the satisfaction of Equation \ref{eq:mix3}.

Let $c=a+b$, the condition \small$$\frac{\mathbb{P}[P^{q}_{p,\beta}=(a,b)]}{\mathbb{P}[Q^{q}_{p,\beta}=(a,b)]}=\frac{p\alpha {a}+\alpha{b}+(1-\alpha-\alpha p)(n-a-b)\cdot \frac{r}{1-2r}}{\alpha {a}+p \alpha {b}+(1-\alpha-\alpha p)(n-a-b)\cdot \frac{r}{1-2r}}> e^{\epsilon'}$$\normalsize holds \emph{if and only if} when $a> low_c=\frac{(e^{\epsilon'}p-1)\alpha c+(e^{\epsilon'}-1)f}{\alpha(e^{\epsilon'}+1)(p-1)}$. Let $P=(A, C-A)$, $P_0=(A+1, C-A)$ and $P_1=(A, C-A+1)$, then the Hockey-stick divergence becomes:
\begin{alignat*}{2}
&D_{e^\epsilon}(P_{p,\beta}^{q} \| Q_{p,\beta}^{q})\\
=&\sum_{a,b\in [0,n]^2}\max\big\{0, \mathbb{P}[P_{\beta,p}^{q}=(a,b)]-e^\epsilon\mathbb{P}[Q_{\beta,p}^{q}=(a,b)]\big\}\\
=&\sum_{c\in [0,n]}\sum_{a\in [\lceil low_c \rceil,c]}\mathbb{P}[P_{p,\beta}^{q}=(a,c-a)]-e^\epsilon\mathbb{P}[Q_{p,\beta}^{q}=(a,c-a)]\\
=&\sum_{c\in [0,n]}\sum_{a\in [\lceil low_c \rceil,c]} (p-e^{\epsilon})\alpha\mathbb{P}[P_{0}=(a,c-a)] \\
&+\sum_{c\in [0,n]}\sum_{a\in [\lceil low_c \rceil,c]} (1-p e^{\epsilon})\alpha\mathbb{P}[P_{1}=(a,c-a)] \\
&+\sum_{c\in [0,n]}\sum_{a\in [\lceil low_c \rceil,c]} (1-\alpha-\alpha p)(1-e^{\epsilon})\mathbb{P}[P=(a,c-a)].
\end{alignat*}
Pluging into the probability formulas of $\mathbb{P}[P_{0}=(a,c-a)], \mathbb{P}[P_{1}=(a,c-a)]$ and $\mathbb{P}[P=(a,c-a)]$, we arrive at the conclusion. 

\section{Proof of Theorem \ref{the:lowerbound}}
Given $x_2=x^*,....,x_n=x^*$, we consider the shuffled messages obtained from applying the mechanism $\mech{S}$ to $\mech{R}_1(x_1^0),\mech{R}_2(x^*),\\...,\mech{R}_2(x^*)$ and $\mech{R}_1(x_1^1),\mech{R}_2(x^*),...,\mech{R}_2(x^*)$. Let $g:\dom{Y}\mapsto \{0,1\}^2$ be a post-processing function on a single message, where $\dom{Y}$ denotes the message space. For any $y\in \dom{Y}$, we define $g(y)$ as follows: \begin{equation*}
    g(y)=\left\{
    \begin{array}{@{}lr@{}}
        (1, 0), & \text{if } \mathbb{P}[R_1(x_1^0)=y]> \mathbb{P}[R_1(x_1^1)=y];\\
        (0, 1), & \text{if } \mathbb{P}[R_1(x_1^0)=y]< \mathbb{P}[R_1(x_1^1)=y];\\
        (0, 0), & \text{else.}
    \end{array}
    \right. 
\end{equation*} 
We also define $g_n:\dom{Y}^n\mapsto \mathbb{N}^2$ as a function on $n$ shuffled messages $S$, where $g_n(S)=\sum_{s\in \mathcal{S}} g(s)$ is the vector summation of $g(s)$. It is observed that 
$g_n({\mech{R}_1(x_1^0),\mech{R}_2(x^*),...,}$ ${\mech{R}_2(x^*)})\stackrel{d}{=} P_{p_0,\beta}^{q_0,q_1}$ and $g_n({\mech{R}_1(x_1^1),\mech{R}_2(x^*),...,\mech{R}_2(x^*)})\stackrel{d}{=} Q_{p_0,\beta}^{q_0,q_1}$. By applying the data-processing inequality property of $D$, we arrive at our conclusion.

\section{Proof of Lemma \ref{lemma:clone}}\label{app:clone}
This lemma generalizes the stronger clone reduction presented in \cite[Lemma 3.2]{feldman2023stronger}, where all local randomizers satisfy LDP and $r$ must equals to $\alpha$, albeit with a similar underlying proof. Let $x_1^0,\ldots,x_n$ and $x_1^1,\ldots,x_n$ be two neighboring datasets. We define the following distributions for $i\in [1,n]$:\begin{equation}\label{eq:intermediater}
{Y}_i = \begin{cases} 0 & \text{w.p.   } r\\ 1 & \text{w.p.   } r \\ i+1 & \text{w.p.   } 1-2r \end{cases}\text{,\ \ \  }\;\;\;\;{Y_1^0} = \begin{cases} 0 & \text{w.p.   } p\alpha\\ 1 & \text{w.p.   } \alpha \\ 2 & \text{w.p.   } 1-p\alpha-\alpha \end{cases}\text{,\ \ \  }\;\;\;\;{Y_1^1} = \begin{cases} 0 & \text{w.p.   } \alpha\\ 1 & \text{w.p.   } p\alpha \\ 2 & \text{w.p.   } 1-p\alpha-\alpha \end{cases}.
\end{equation}
We now consider two independent sampling processes:
(i) we draw one sample from $Y_1^0$, and one sample from every $Y_i$ (for $i\in [2,n]$);
(ii) we draw one sample from $Y_1^1$, and one sample from every $Y_i$ (for $i\in [1,n]$).
Using the mixture property of $\mech{R}_1$ and $\mech{R}_i$ (for $i\in [2,n]$), we obtain $\mech{S}(\mech{R}_1(x_1^0),\ldots,\mech{R}_n(x_n))$ and $\mech{S}(\mech{R}_1(x_1^1),\ldots,\mech{R}_n(x_n))$ by post-processing $\mech{S}(Y_1^0, Y_2,\ldots, Y_n)$ and $\mech{S}(Y_1^1, Y_2,\ldots, Y_n)$ through some function $G:[0,n+1]^n\mapsto \dom{Y}^n$. 

We now focus on the statistical divergence between $\mech{S}(Y_1^0, Y_2,\ldots, Y_n)$ and $\mech{S}(Y_1^1, Y_2,\ldots, Y_n)$. Since they differ only in random variables $Y_1^0, Y_1^1$, and the two distributions differ only in the probability distribution over ${0,1}$, the $2,3,...,n+1$ terms in $\mech{S}(Y_1^0, Y_2,\ldots, Y_n)$ and $\mech{S}(Y_1^1, Y_2,\ldots, Y_n)$ can be omitted. Formally, for any distance measure $D$ that satisfies the data processing inequality, our goal is to prove that the following inequalities hold:
\begin{alignat*}{2}
D(\mech{S}(Y_1^0, Y_2,\ldots, Y_n)\|\mech{S}(Y_1^1, Y_2,\ldots, Y_n))\leq D(P' \| Q'),\\
D(\mech{S}(Y_1^0, Y_2,\ldots, Y_n)\|\mech{S}(Y_1^1, Y_2,\ldots, Y_n))\geq D(P' \| Q').
\end{alignat*}
For the first inequality, we define the following post-processing function over $P'$ or $Q'$: (1) assume two numbers in $P'$ or $Q'$ are $(a,b)$, uniformly sample $n-a-b$ elements from $[3,n+1]$ (denoted as $E$); (2) with probability of $\frac{(n-a-b)(1-p\alpha-\alpha)(2r)}{(n-a-b)(1-p\alpha-\alpha)(2r)+(a+b)(p\alpha+\alpha)(1-2r)}$, replace one uniform-random element in $E$ with $2$; (3) initialize a list of $a$-repeat $0$s and $b$-repeat $1$s, then append $E$ to the list, finally uniform-randomly shuffle the list. It can be observed that the post-processing result of $P'$ distributionally equals to $\mech{S}(Y_1^0, Y_2,\ldots, Y_n)$ and the post-processing result of $Q'$ distributionally equals to $\mech{S}(Y_1^1, Y_2,\ldots, Y_n)$. By applying the data processing inequality, we obtain $D(\mech{S}(Y_1^0, Y_2,\ldots, Y_n)\|\mech{S}(Y_1^1, Y_2,\ldots, Y_n))\leq D(P' \| Q')$.
For the second inequality, we define the following post-processing function over the output of $\mech{S}(Y_1, Y_2,\ldots, Y_n)$ (either $Y_1=Y_1^0$ or $Y_1=Y_1^1$): remove all $2,\dots,n+1$ from the output. Now observe that the number of $0$s and $1$s in $\mech{S}(Y_1^0, Y_2,\ldots, Y_n)$ follows the distribution $P'=(A+\Delta_1, C-A+\Delta_2)$, while the number of $0$s and $1$s in $\mech{S}(Y_1^1, Y_2,\ldots, Y_n)$ follows the distribution $Q'=(A+\Delta_2, C-A+\Delta_1)$. By applying the data processing inequality, we obtain $D(\mech{S}(Y_1^0, Y_2,\ldots, Y_n)\|\mech{S}(Y_1^1, Y_2,\ldots, Y_n))\geq D(P' \| Q')$.

Combining the previous two paragraphs, we get $D(\mech{S}(\mech{R}_1(x_1^0),..,\mech{R}_n(x_n)) \| \mech{S}(\mech{R}_1(x_1^1),..,\mech{R}_n(x_n)))
\leq D(\mech{S}(Y_1^0, Y_2,\ldots, Y_n)\|\mech{S}(Y_1^1, Y_2,\ldots, Y_n))\leq D(P' \| Q')$.

\section{Proof of Lemma \ref{lemma:nondecrease}}\label{app:nondecrease}
Let us define a post-processing function $g: \mathbb{N}^2\mapsto \mathbb{N}^2$ as $g(d, e)=(\operatorname{Binomial}(d, \beta'/\beta), \operatorname{Binomial}(e, \beta'/\beta))$. To prove the post-processing inequality of distance measure $D$, we need to show that $g(P^{q}_{p,\beta})\stackrel{d}{=}P^{q}_{\beta',p}$ and $g(Q^{q}_{p,\beta})\stackrel{d}{=}Q^{q}_{\beta',p}$.

We shall use an equivalent sampling process for $P^{q}_{p,\beta}$ and $Q^{q}_{p,\beta}$ similar to Equation \ref{eq:intermediater}. Let us define the following random variables:
\begin{equation*}
{G}_i = \begin{cases} (1,0) & \text{w.p.   } r\\ (0,1) & \text{w.p.   } r \\ (0,0) & \text{w.p.   } 1-2r \end{cases}\text{,\ \ }\;\;{G_1^{0}} = \begin{cases} (1,0) & \text{w.p.   } p\alpha\\ (0,1) & \text{w.p.   } \alpha \\ (0,0) & \text{w.p.   } 1-p\alpha-\alpha \end{cases}\text{,\ \ }\;\;{G_1^{1}} = \begin{cases} (1,0) & \text{w.p.   } \alpha\\ (0,1) & \text{w.p.   } p\alpha \\ (0,0) & \text{w.p.   } 1-p\alpha-\alpha \end{cases},
\end{equation*}
then $P^{q}_{p,\beta}\stackrel{d}{=}G_1^{0}+\sum_{i=2}^n G_i$ and $P^{q}_{p,\beta}\stackrel{d}{=} G_1^{1}+\sum_{i=2}^n G_i$, where the same $G_i$ appearing in the two equations are independently sampled.

Let ${G}_1^{0'}=g(G_1^0)$, ${G}_1^{1'}=g(G_1^0)$, and $G'_i=g(G_i)$, they follow distributions:
\begin{equation*}
{G}'_i = \begin{cases} (1,0) & \text{w.p.   } r'\\ (0,1) & \text{w.p.   } r' \\ (0,0) & \text{w.p.   } 1-2r' \end{cases}\text{,\ \ }\;\;{G_1^{0'}} = \begin{cases} (1,0) & \text{w.p.   } p\alpha'\\ (0,1) & \text{w.p.   } \alpha' \\ (0,0) & \text{w.p.   } 1-p\alpha'-\alpha' \end{cases}\text{,\ \ }\;\;{G_1^{1'}} = \begin{cases} (1,0) & \text{w.p.   } \alpha'\\ (0,1) & \text{w.p.   } p\alpha' \\ (0,0) & \text{w.p.   } 1-p\alpha'-\alpha' \end{cases},
\end{equation*}
where $\alpha'=\frac{\beta'}{p-1}$ and $r'=\frac{\alpha' p}{q}$. Therefore, we have $g(P^{q}_{p,\beta})\stackrel{d}{=} G_1^{0'}+\sum_{i=2}^n G'_i \stackrel{d}{=}  P^{q}_{\beta',p}$ and $g(Q^{q}_{p,\beta})\stackrel{d}{=} G_1^{1'}+\sum_{i=2}^n G'_i \stackrel{d}{=} Q^{q}_{\beta',p}$, by applying the post-processing inequality of distance measure $D$.

\section{Detailed Proof of Theorem \ref{the:hsdnumerical}}\label{app:hsdnumerical}
According to the definition of Hockey-stick divergence, we have:
\begin{align}\label{eq:hsformula}
D_{e^\epsilon}(P_{p,\beta}^{q} \| Q_{p,\beta}^{q})=\sum_{a,b\in [0,n]^2}\max\{0, \mathbb{P}[P_{p,\beta}^{q}=(a,b)]-e^\epsilon\mathbb{P}[Q_{p,\beta}^{q}=(a,b)]\}.
\end{align}
Recall that $C\sim Binomial(n-1, 2r)$, $A\sim Binomial(C, 1/2)$, and $\Delta_1=Bernoulli(p \alpha)$ and $\Delta_2=Bernoulli(1-\Delta_1, \alpha/(1-p \alpha))$. First consider $P=(A, C-A)$, $P_0=(A+1, C-A)$ and $P_1=(A, C-A+1)$. Based on the above sampling process, it is derived that
\begin{equation}
\begin{aligned}
&& \frac{\mathbb{P}[P_0=(a,b)]}{\mathbb{P}[P_1=(a,b)]} &= \frac{{n-1\choose a+b-1}(2r)^{a+b-1}(1-2r)^{n-a-b}{a+b-1\choose a-1}(1/2)^{a+b-1}}{{n-1\choose a+b-1}(2r)^{a+b-1}(1-2r)^{n-a-b}{a+b-1\choose a}(1/2)^{a+b-1}}&\\
&& &=\frac{a}{b}.&\\
\end{aligned}
\end{equation}
Similarly, we have:
\begin{equation}
\begin{aligned}
&& \frac{\mathbb{P}[P=(a,b)]}{\mathbb{P}[P_1=(a,b)]} &= \frac{{n-1\choose a+b}(2r)^{a+b}(1-2r)^{n-a-b-1}{a+b\choose a}(1/2)^{a+b}}{{n-1\choose a+b-1}(2r)^{a+b-1}(1-2r)^{n-a-b}{a+b-1\choose a}(1/2)^{a+b-1}}&\\
&& &=\frac{n-a-b}{a+b}\cdot \frac{2r}{1-2r}\cdot \frac{a+b}{2b}.&\\
&& &=\frac{n-a-b}{b}\cdot \frac{r}{1-2r}.&\\
\end{aligned}
\end{equation}

Now consider two variables $P^{q}_{p,\beta}$ and  $Q^{q}_{p,\beta}$, we have the probability ratio as:
\begin{equation}\label{eq:PQ1}
\begin{aligned}
&& \frac{\mathbb{P}[P^{q}_{p,\beta}=(a,b)]}{\mathbb{P}[Q^{q}_{p,\beta}=(a,b)]} &= \frac{p\alpha {\mathbb{P}[P_0=(a,b)]}+\alpha {\mathbb{P}[P_1=(a,b)]}+(1-\alpha-\alpha p){\mathbb{P}[P=(a,b)]}}
{\alpha {\mathbb{P}[P_0=(a,b)]}+p \alpha {\mathbb{P}[P_1=(a,b)]}+(1-\alpha-\alpha p) {\mathbb{P}[P=(a,b)]}}&\\
&& &=\frac{(p-1)\alpha {\mathbb{P}[P_0=(a,b)]}+(1-p)\alpha {\mathbb{P}[P_1=(a,b)]}}{\alpha {\mathbb{P}[P_0=(a,b)]}+p \alpha {\mathbb{P}[P_1=(a,b)]}+(1-\alpha-\alpha p){\mathbb{P}[P=(a,b)]}}&\\
&& &= 1+\frac{(p-1)\alpha {a}+(1-p)\alpha {b}}{\alpha {a}+p \alpha {b}+(1-\alpha-\alpha p)\cdot (n-a-b)\cdot \frac{r}{1-2r}}.&\\
&& &= 1+\frac{(p-1)\alpha (a-b)}{\alpha {(a+b)}+(p-1) \alpha {b}+(1-\alpha-\alpha p)\cdot (n-(a+b))\cdot \frac{r}{1-2r}}.&\\
\end{aligned}
\end{equation}
A key observation is that the above formula of $\frac{\mathbb{P}[P^{q}_{p,\beta}=(a,b)]}{\mathbb{P}[Q^{q}_{p,\beta}=(a,b)]}$ monotonically increases with $a$ when $a+b$ is fixed. Therefore, we let $c=a+b$, then the condition $\frac{\mathbb{P}[P^{q}_{p,\beta}=(a,b)]}{\mathbb{P}[Q^{q}_{p,\beta}=(a,b)]}=\frac{p\alpha {a}+\alpha{b}+(1-\alpha-\alpha p)(n-a-b)\cdot \frac{r}{1-2r}}{\alpha {a}+p \alpha {b}+(1-\alpha-\alpha p)(n-a-b)\cdot \frac{r}{1-2r}}> e^{\epsilon'}$ holds \emph{if and only if} when $a> low_c=\frac{(e^{\epsilon'}p-1)\alpha c+(e^{\epsilon'}-1)f}{\alpha(e^{\epsilon'}+1)(p-1)}$. Therefore, the Equation \ref{eq:hsformula} becomes:
\begin{align*}\label{eq:hsformula}
&&D_{e^\epsilon}(P_{p,\beta}^{q} \| Q_{p,\beta}^{q})&=\sum_{c\in [0,n]}\sum_{a\in [\lceil low_c \rceil,c]}\mathbb{P}[P_{p,\beta}^{q}=(a,c-a)]-e^\epsilon\mathbb{P}[Q_{p,\beta}^{q}=(a,c-a)]&\\
&& &=\sum_{c\in [0,n]}\sum_{a\in [\lceil low_c \rceil,c]} (p-e^{\epsilon'})\alpha\mathbb{P}[P_{0}=(a,c-a)] &\\
&& &\ \ \ \ +\sum_{c\in [0,n]}\sum_{a\in [\lceil low_c \rceil,c]} (1-p e^{\epsilon'})\alpha\mathbb{P}[P_{1}=(a,c-a)] &\\
&& &\ \ \ \ +\sum_{c\in [0,n]}\sum_{a\in [\lceil low_c \rceil,c]} (1-\alpha-\alpha p)(1-e^{\epsilon'})\mathbb{P}[P=(a,c-a)]& \\
&& &=(p-e^{\epsilon'})\alpha\sum_{c\in [0,n]}{n-1\choose c-1}(2r)^{c-1}(1-2r)^{n-c}\Big(\sum_{a\in [\lceil low_c \rceil,c]} {c-1\choose a-1}(1/2)^{c-1}\Big) &\\
&& &\ \ \ \ +(1-p e^{\epsilon'})\alpha\sum_{c\in [0,n]}{n-1\choose c-1}(2r)^{c-1}(1-2r)^{n-c}\Big(\sum_{a\in [\lceil low_c \rceil,c]} {c-1\choose a}(1/2)^{c-1}\Big) &\\
&& &\ \ \ \ +(1-\alpha-\alpha p)(1-e^{\epsilon'})\sum_{c\in [0,n]}{n-1\choose c}(2r)^{c}(1-2r)^{n-c-1}\Big(\sum_{a\in [\lceil low_c \rceil,c]} {c\choose a}(1/2)^{c}\Big).&
\end{align*}
Notice that the formula $\sum_{c\in [0,n]}{n-1\choose c-1}(2r)^{c-1}(1-2r)^{n-c}\Big(\sum_{a\in [\lceil low_c \rceil,c]} {c-1\choose a-1}(1/2)^{c-1}\Big)$ equals to:
$$\small\sum_{c\in [0,n-1]}{n-1\choose c}(2r)^{c}(1-2r)^{n-c-1}\Big(\sum_{a\in [\lceil low_{c+1} -1\rceil,c]} {c\choose a}(1/2)^{c}\Big)=\mathop{\mathbb{E}}\limits_{c\sim Binom(n-1,2r)}\mathop{\mathsf{CDF}}\limits_{c,1/2}[\lceil low_{c+1}-1\rceil, c];$$ 
 the formula  $\sum_{c\in [0,n]}{n-1\choose c-1}(2r)^{c-1}(1-2r)^{n-c}\Big(\sum_{a\in [\lceil low_c \rceil,c]} {c-1\choose a}(1/2)^{c-1}\Big)$ equals to:
$$\small\sum_{c\in [0,n-1]}{n-1\choose c}(2r)^{c}(1-2r)^{n-c-1}\Big(\sum_{a\in [\lceil low_{c+1} \rceil,c]} {c\choose a}(1/2)^{c}\Big)=\mathop{\mathbb{E}}\limits_{c\sim Binom(n-1,2r)}\mathop{\mathsf{CDF}}\limits_{c,1/2}[\lceil low_{c+1}\rceil, c];$$
 the formula $\sum_{c\in [0,n]}{n-1\choose c}(2r)^{c}(1-2r)^{n-c-1}\Big(\sum_{a\in [\lceil low_c \rceil,c]} {c\choose a}(1/2)^{c}\Big)$ equals to:
$$\mathop{\mathbb{E}}\limits_{c\sim Binom(n-1,2r)}\mathop{\mathsf{CDF}}\limits_{c,1/2}[\lceil low_c\rceil, c].$$
Combining these three equations, we have proved the equation about $D_{e^\epsilon}(P_{p,\beta}^{q} \| Q_{p,\beta}^{q})$.

As for $D_{e^\epsilon}(Q_{p,\beta}^{q} \| P_{p,\beta}^{q})$, a key observation is that $\frac{\mathbb{P}[P^{q}_{p,\beta}=(a,b)]}{\mathbb{P}[Q^{q}_{p,\beta}=(a,b)]}$ monotonically decreases with $a$ when $a+b$ is fixed, and the condition $\frac{\mathbb{P}[P^{q}_{p,\beta}=(a,b)]}{\mathbb{P}[Q^{q}_{p,\beta}=(a,b)]}=\frac{p\alpha {a}+\alpha{b}+(1-\alpha-\alpha p)(n-a-b)\cdot \frac{r}{1-2r}}{\alpha {a}+p \alpha {b}+(1-\alpha-\alpha p)(n-a-b)\cdot \frac{r}{1-2r}}< e^{-\epsilon'}$ holds \emph{if and only if} when $a< high_c=\frac{(e^{-\epsilon'}p-1)\alpha c+(e^{-\epsilon'}-1)f}{\alpha(e^{-\epsilon'}+1)(p-1)}$. As with previous procedures, the present analysis derives the equation governing $D_{e^\epsilon}(Q_{p,\beta}^{q} \| P_{p,\beta}^{q})$.

\section{Proof of Theorem \ref{the:tightlevel}}\label{app:analyticbound}
Recall that $C\sim Binomial(n-1, 2r)$, $A\sim Binomial(C, 1/2)$, and $\Delta_1=Bernoulli(p \alpha)$ and $\Delta_2=Bernoulli(1-\Delta_1, \alpha/(1-p \alpha))$. Similar to proof for Theorem \ref{the:hsdnumerical}, we first consider $P=(A, C-A)$, $P_0=(A+1, C-A)$ and $P_1=(A, C-A+1)$. Use the fact that $C=a+b$ or $C=a+b-1$ and $|A-C/2|< \sqrt{C/2\log(4/\delta)}$ holds with probability $1-\delta/2$ (by the Hoeffding's inequality), we have the probability ratio upper bounded as:
\begin{equation*}
\begin{aligned}
&& \frac{\mathbb{P}[P^{q}_{p,\beta}=(a,b)]}{\mathbb{P}[Q^{q}_{p,\beta}=(a,b)]} &= \frac{p\alpha {\mathbb{P}[P_0=(a,b)]}+\alpha {\mathbb{P}[P_1=(a,b)]}+(1-\alpha-\alpha p){\mathbb{P}[P=(a,b)]}}
{\alpha {\mathbb{P}[P_0=(a,b)]}+p \alpha {\mathbb{P}[P_1=(a,b)]}+(1-\alpha-\alpha p) {\mathbb{P}[P=(a,b)]}}&\\
&& &=1+\frac{(p-1)\alpha {\mathbb{P}[P_0=(a,b)]}+(1-p)\alpha {\mathbb{P}[P_1=(a,b)]}}{\alpha {\mathbb{P}[P_0=(a,b)]}+p \alpha {\mathbb{P}[P_1=(a,b)]}+(1-\alpha-\alpha p){\mathbb{P}[P=(a,b)]}}&\\
&& &= 1+\frac{(p-1)\alpha ((a+b)-2b)}{\alpha {(a+b)}+(p-1) \alpha {b}+(1-\alpha-\alpha p)\cdot (n-(a+b))\cdot \frac{r}{1-2r}}.&\\
&& &\leq 1+\frac{(p-1)\alpha (C+1-2b)}{\alpha C+(p-1) \alpha {b}+(1-\alpha-\alpha p)\cdot (n-1-C)\cdot \frac{r}{1-2r}}&\\
&& &\leq 1+\frac{(p-1)\alpha (2\sqrt{C/2\log(4/\delta)}+1)}{\alpha C+(p-1) \alpha (C/2-\sqrt{C/2\log(4/\delta)})+(1-\alpha-\alpha p)\cdot (n-1-C)\cdot \frac{r}{1-2r}}.&\\
\end{aligned}
\end{equation*}
Based on the induced formula in the last line (denoted as $1+F(C)$), if the coefficient $\alpha+(p-1)\alpha/2-(1-\alpha-\alpha p)r/(1-2r)$ of $C$ in the denominator is no less than $0$, the derivative $\frac{\mathrm{d} F}{\mathrm{d} C}$ is lower than $0$ when $C>\frac{2p(\beta+1+(\beta-1)p)(n-1)+\beta}{q+p(\beta-1+(\beta+1)p)-pq}$. Now focus on the variable $C$, according to the multiplicative Chernoff bound and Hoeffding's inequality, it is derived that $C\geq (n-1)2r-\sqrt{\min\{6r,1/2\}(n-1)\log(4/\delta)}$ holds with probability at least $1-\delta/2$. Therefore, if $\Omega=(n-1)2r-\sqrt{\min\{6r,1/2\}(n-1)\log(4/\delta)}\geq \frac{2p(\beta+1+(\beta-1)p)(n-1)+\beta}{q+p(\beta-1+(\beta+1)p)-pq}$, then with probability at least $1-\delta/2$, the $F(C)\leq F(\Omega)$ and hence $\frac{\mathbb{P}[P^{q}_{p,\beta}=(a,b)]}{\mathbb{P}[Q^{q}_{p,\beta}=(a,b)]}\leq e^{\epsilon}$ holds. 

Similarly, we have the probability ratio lower bounded as:
\begin{alignat*}{5}
&& \frac{\mathbb{P}[P^{q}_{p,\beta}=(a,b)]}{\mathbb{P}[Q^{q}_{p,\beta}=(a,b)]} &= \frac{p\alpha  {a}+\alpha b +(1-\alpha-\alpha p)\cdot (n-a-b)\cdot \frac{r}{1-2r}}{\alpha {a}+p \alpha b+(1-\alpha-\alpha p)\cdot (n-a-b)\cdot \frac{r}{1-2r}}&\\
&& &={1}/({\frac{\alpha {a}+p \alpha b+(1-\alpha-\alpha p)\cdot (n-a-b)\cdot \frac{r}{1-2r}}{p\alpha {a}+\alpha b+(1-\alpha-\alpha p)\cdot (n-a-b)\cdot \frac{r}{1-2r}}}) &\\
&& &={1}/(1+{\frac{(p-1) \alpha ((a+b)-2a)}{\alpha (a+b)+(p-1)\alpha a+(1-\alpha-\alpha p)\cdot (n-(a+b))\cdot \frac{r}{1-2r}}}) &\\
&& &\geq 1/(1+\frac{(p-1)\alpha (C+1-2a)}{\alpha C+(p-1) \alpha {a}+(1-\alpha-\alpha p)\cdot (n-1-C)\cdot \frac{r}{1-2r}})&\\
&& &\geq 1/(1+\frac{(p-1)\alpha (2\sqrt{C/2\log(4/\delta)}+1)}{\alpha C+(p-1) \alpha (C/2-\sqrt{C/2\log(4/\delta)})+(1-\alpha-\alpha p)\cdot (n-1-C)\cdot \frac{r}{1-2r}}).&
\end{alignat*}
Then under the same condition about $C$, we have $\frac{\mathbb{P}[P^{q}_{p,\beta}=(a,b)]}{\mathbb{P}[Q^{q}_{p,\beta}=(a,b)]}\geq e^{-\epsilon}$ holds.

\section{Proof of Theorem \ref{the:level}}\label{app:asymptoticbound}
According to multiplicative Chernoff bound and Hoeffding's inequality, we have $|C-(n-1)2r|< \sqrt{\min\{6r,1/2\}(n-1)\log(4/\delta)}$ holds with probability at least $1-\delta/2$; according to Hoeffding's inequality, $|A-C/2|< \sqrt{C/2\log(4/\delta)}$ holds with probability $1-\delta/2$. Specifically, when $n\geq \frac{8\log(2/\delta)}{r}$, both $\sqrt{\min\{6r,1/2\}(n-1)\log(4/\delta)}< \min\{r, \sqrt{r}/4, 1-2r\}(n-1)$ and $\sqrt{C/2\log(4\delta)}< C/4$ hold. The remaining proof conditions on these events.

It is observed that:
\begin{equation}
\begin{aligned}
&& \frac{n-a-b}{b}\cdot \frac{r}{1-2r} &=\frac{n-C}{C-A}\cdot \frac{r}{1-2r} &\\
&& &\geq\frac{n-C}{C}\cdot \frac{4r}{3(1-2r)} &\\
&& &\geq\frac{n-(n-1)2r-\min\{r, \sqrt{r}/4, 1-2r\}(n-1)}{(n-1)2r+\min\{r, \sqrt{r}/4, 1-2r\}(n-1)}\cdot \frac{4r}{3(1-2r)} &\\
&& &\geq\frac{1-2r-\min\{r, \sqrt{r}/4, 1-2r\}}{2r+\min\{r, \sqrt{r}/4, 1-2r\}}\cdot \frac{4r}{3(1-2r)} &\\
&& &\geq\frac{4(1-3r)}{9(1-2r)}. &
\end{aligned}
\end{equation}
By symmetry of $a$ and $b$, we also have $\frac{n-a-b}{a}\cdot \frac{r}{1-2r}\geq \frac{4(1-3r)}{9(1-2r)}$. We use $c_r$ to denote $\max\{0, \frac{4(1-3r)}{9(1-2r)}\}$.

Based on the classical clone reduction \cite[Lemma A.3]{feldman2023stronger}, when $n\geq \frac{8\log(2/\delta)}{r}$, it is derived that the following two inequalities hold with $\epsilon'=\log(1+\sqrt{\frac{32\log(4/\delta)}{r (n-1)}}+\frac{4}{r n})$: 
\[e^{-\epsilon'}\leq \frac{\mathbb{P}[P_0=(a,b)]}{\mathbb{P}[P_1=(a,b)]}\leq e^{\epsilon'},\]
Now notice that $P^{q}_{p,\beta}=p \alpha P_0+\alpha P_1+(1-\alpha-\alpha p)P$ and $Q^{q}_{p,\beta}=\alpha P_0+p \alpha P_1+(1-\alpha-\alpha p)P$, we have:
\begin{equation}\label{eq:PQ}
\begin{aligned}
&& \frac{\mathbb{P}[P^{q}_{p,\beta}=(a,b)]}{\mathbb{P}[Q^{q}_{p,\beta}=(a,b)]} &= \frac{p\alpha \frac{\mathbb{P}[P_0=(a,b)]}{\mathbb{P}[P_1=(a,b)]}+\alpha+(1-\alpha-\alpha p)\frac{\mathbb{P}[P=(a,b)]}{\mathbb{P}[P_1=(a,b)]}}
{\alpha \frac{\mathbb{P}[P_0=(a,b)]}{\mathbb{P}[P_1=(a,b)]}+p \alpha+(1-\alpha-\alpha p)\frac{\mathbb{P}[P=(a,b)]}{\mathbb{P}[P_1=(a,b)]}}&\\
&& &=1+\frac{(p-1)\alpha \frac{\mathbb{P}[P_0=(a,b)]}{\mathbb{P}[P_1=(a,b)]}+(1-p)\alpha}{\alpha \frac{\mathbb{P}[P_0=(a,b)]}{\mathbb{P}[P_1=(a,b)]}+p \alpha+(1-\alpha-\alpha p)\frac{\mathbb{P}[P=(a,b)]}{\mathbb{P}[P_1=(a,b)]}}&\\
&& &\leq 1+\frac{(p-1)\alpha\frac{a}{b}+(1-p)\alpha}{\alpha \frac{a}{b}+p \alpha+(1-\alpha-\alpha p)\cdot\frac{n-a-b}{b}\cdot \frac{r}{1-2r}}&\\
&& &\leq 1+\frac{(p-1)\alpha\frac{a}{b}+(1-p)\alpha}{\alpha \frac{a}{b}+p \alpha+(1-\alpha-\alpha p)\cdot c_r}&\\
&& &\leq 1+\frac{(p-1)\alpha e^{\epsilon'}+(1-p)\alpha}{\alpha e^{\epsilon'}+p \alpha+(1-\alpha-\alpha p)\cdot c_r}&\\
&& &\leq 1+\frac{(p-1)\alpha}{\alpha+p \alpha+(1-\alpha-\alpha p)\cdot c_r}\cdot (e^{\epsilon'}-1)&\\
&& &\leq 1+\frac{\beta}{\alpha+p \alpha+(1-\alpha-\alpha p)\cdot c_r}\cdot (e^{\epsilon'}-1).&
\end{aligned}
\end{equation}
Besides, we have:
\begin{equation}\label{eq:QP}
\begin{aligned}
&& \frac{\mathbb{P}[P^{q}_{p,\beta}=(a,b)]}{\mathbb{P}[Q^{q}_{p,\beta}=(a,b)]} &= \frac{p\alpha \frac{a}{b}+\alpha+(1-\alpha-\alpha p)\cdot\frac{n-a-b}{b}\cdot \frac{r}{1-2r}}{\alpha \frac{a}{b}+p \alpha+(1-\alpha-\alpha p)\cdot\frac{n-a-b}{b}\cdot \frac{r}{1-2r}}&\\
&& &={1}/({\frac{\alpha \frac{a}{b}+p \alpha+(1-\alpha-\alpha p)\cdot\frac{n-a-b}{b}\cdot \frac{r}{1-2r}}{p\alpha \frac{a}{b}+\alpha+(1-\alpha-\alpha p)\cdot\frac{n-a-b}{b}\cdot \frac{r}{1-2r}}}) &\\
&& &={1}/({\frac{\alpha +p \alpha \frac{b}{a}+(1-\alpha-\alpha p)\cdot\frac{n-a-b}{b}\cdot \frac{r}{1-2r}\frac{b}{a}}{p\alpha+\alpha \frac{b}{a} +(1-\alpha-\alpha p)\cdot\frac{n-a-b}{b}\cdot \frac{r}{1-2r}\frac{b}{a}}}) &\\
&& &\geq {1}/({\frac{\alpha +p \alpha e^{\epsilon'}+(1-\alpha-\alpha p)\cdot\frac{n-a-b}{a}\cdot \frac{r}{1-2r}}{p\alpha+\alpha e^{\epsilon'} +(1-\alpha-\alpha p)\cdot\frac{n-a-b}{a}\cdot \frac{r}{1-2r}}}) &\\
&& &\geq {1}/({1+\frac{(p-1)\alpha}{p\alpha+\alpha+(1-\alpha-\alpha p)c_r}\cdot (e^{\epsilon'}-1)}) &\\
&& &\geq {1}/({1+\frac{\beta}{{p\alpha+\alpha+(1-\alpha-\alpha p)c_r}}\cdot (e^{\epsilon'}-1)}).&
\end{aligned}
\end{equation}

By combining Equations \ref{eq:PQ} and \ref{eq:QP}, it follows that, with a probability of $1-\delta$, the inequality $\frac{\mathbb{P}[P^{q}_{p,\beta}=(a,b)]}{\mathbb{P}[Q^{q}_{p,\beta}=(a,b)]}\in [e^{-\epsilon}, e^{\epsilon}]$ holds. The value of $\epsilon$ is defined as $\log(1+\frac{\beta}{\beta(1+p)/(p-1)+(1-\beta(1+p)/(p-1))c_r})(\sqrt{\frac{32\log(4/\delta)}{r (n-1)}}+\frac{4}{r n}))$.

\section{Probability ratio of $P^{q_0,q_1}_{p_0,\beta}$ and $Q^{q_0,q_1}_{p_0,\beta}$}\label{app:ratiolower}

Recall that for $p_0> 1, \beta\in [0, \frac{p_0-1}{p_0+1}],  q_0,q_1\in [1,+\infty)$ such that $q_0\leq p_0 q_1$ and $q_1\leq p_0 q_0$, we define $\alpha$ as $\frac{\beta}{(p_0-1)}$, $r_0$ as $\frac{\alpha p_0}{q_0}$ and $r_1$ as $\frac{\alpha p_0}{q_1}$, and $C\sim Binom(n-1, r_0+r_1)$, $A\sim Binom(C, r_0/(r_0+r_1))$, and $\Delta_1=Bernoulli(p_0 \alpha')$ and $\Delta_2=Bernoulli(1-\Delta_1, \alpha/(1-p_0 \alpha))$. The random variable $P_{p_0,\beta}^{q_0,q_1}$ corresponds to $(A+\Delta_1, C-A+\Delta_2)$ and the $Q_{p_0,\beta}^{q_0,q_1}$ corresponds to $(A+\Delta_2, C-A+\Delta_1)$ from two independent samplings.

To see how the algorithm works, we put forth following analyses on the two variables. According to the sampling process, we have: 
\begin{align*}
&& \mathbb{P}[P_{p_0,\beta}^{q_0,q_1} = (a,b)]=&p\alpha {n-1\choose a+b-1}(r_0+r_1)^{a+b-1}(1-r_0-r_1)^{n-a-b}{a+b-1\choose a-1}\frac{(r_0)^{a-1}(r_1)^{b}}{(r_0+r_1)^{a+b-1}} &\\
&& &+\alpha {n-1\choose a+b-1}(r_0+r_1)^{a+b-1}(1-r_0-r_1)^{n-a-b}{a+b-1\choose a}\frac{(r_0)^a(r_1)^{b-1}}{(r_0+r_1)^{a+b-1}} &\\
&& &+(1-\alpha-p_0\alpha) {n-1\choose a+b}(r_0+r_1)^{a+b}(1-r_0-r_1)^{n-a-b-1}{a+b\choose a}\frac{(r_0)^a(r_1)^{b}}{(r_0+r_1)^{a+b}}.&
\end{align*}
Similarly, we have: 
\begin{align*}
&& \mathbb{P}[Q_{p_0,\beta}^{q_0,q_1} = (a,b)] =& \alpha {n-1\choose a+b-1}(r_0+r_1)^{a+b-1}(1-r_0-r_1)^{n-a-b}{a+b-1\choose a-1}\frac{(r_0)^{a-1}(r_1)^{b}}{(r_0+r_1)^{a+b-1}} &\\
&& &+p_0\alpha {n-1\choose a+b-1}(r_0+r_1)^{a+b-1}(1-r_0-r_1)^{n-a-b}{a+b-1\choose a}\frac{(r_0)^a(r_1)^{b-1}}{(r_0+r_1)^{a+b-1}} &\\
&& &+(1-\alpha-p_0\alpha) {n-1\choose a+b}(r_0+r_1)^{a+b}(1-r_0-r_1)^{n-a-b-1}{a+b\choose a}\frac{(r_0)^a(r_1)^{b}}{(r_0+r_1)^{a+b}}.&
\end{align*}
Then, it is observed that:
$$\frac{{n-1\choose a+b-1}(r_0+r_1)^{a+b-1}(1-r_0-r_1)^{n-a-b}{a+b-1\choose a}\frac{(r_0)^a(r_1)^{b-1}}{(r_0+r_1)^{a+b-1}}}{{n-1\choose a+b-1}(r_0+r_1)^{a+b-1}(1-r_0-r_1)^{n-a-b}{a+b-1\choose a-1}\frac{(r_0)^{a-1}(r_1)^{b}}{(r_0+r_1)^{a+b-1}}}=\frac{r_0 b}{r_1 a},$$
$$\frac{{n-1\choose a+b}(r_0+r_1)^{a+b}(1-r_0-r_1)^{n-a-b-1}{a+b\choose a}\frac{(r_0)^a(r_1)^{b}}{(r_0+r_1)^{a+b}}}{{n-1\choose a+b-1}(r_0+r_1)^{a+b-1}(1-r_0-r_1)^{n-a-b}{a+b-1\choose a-1}\frac{(r_0)^{a-1}(r_1)^{b}}{(r_0+r_1)^{a+b-1}}}=\frac{r_0(n-a-b)}{(1-r_0-r_1)a}.$$
Consequently, we get:
\begin{align}\label{eq:PQ11}
&& \frac{\mathbb{P}[P_{p_0,\beta}^{q_0,q_1} = (a,b)]}{\mathbb{P}[Q_{p_0,\beta}^{q_0,q_1} = (a,b)]} = \frac{p_0\alpha a/r_0+\alpha b/r_1+(1-\alpha-p_0 \alpha)(n-a-b)/(1-r_0-r_1))}{\alpha a/r_0+p_0\alpha b/r_1+(1-\alpha-p_0 \alpha)(n-a-b)/(1-r_0-r_1))}.&
\end{align}

\section{Numerical Lower Bounds}
We proceed to numerically compute the lower bound for privacy amplification. While a naive approach would involve enumerating the entire output space of $P_{p,\beta}^{q_0,q_1}$ and $Q_{p,\beta}^{q_0,q_1}$ with $O(T n^2)$ complexities, we propose a more efficient implementation. Assuming $q_0/q_1\in [1/p,p]$, and with $a+b$ fixed, we observe that the ratio $\frac{\mathbb{P}[P_{p,\beta}^{q_0,q_1} = (a,b)]}{\mathbb{P}[Q_{p,\beta}^{q_0,q_1} = (a,b)]}$ monotonically increases with $a$ (see Appendix \ref{app:ratiolower} for details). Specifically, let $g$ denote $(1-\alpha-\alpha p)(n-c)\cdot \frac{1}{1-r_0-r_1}$, then if $a>low$, where $low=\frac{(e^{\epsilon'}p-1)\alpha c/r_1+(e^{\epsilon'}-1)g}{\alpha(p/r_0-1/r_1+e^{\epsilon'}(p/r_1-1/r_0))}$, the ratio exceeds $e^{\epsilon}$. If $a<high$, where $high=\frac{(e^{-\epsilon'}p-1)\alpha c/r_1+(e^{-\epsilon'}-1)g}{\alpha(p/r_0-1/r_1+e^{-\epsilon'}(p/r_1-1/r_0))}$, the ratio is lower than $e^{-\epsilon}$. We present the divergence in an expectation form in Proposition \ref{pro:hsdnumerical2} and provide an efficient implementation in Algorithm \ref{alg:biasedlowerbound} with $\tilde{O}(T n)$ complexities (see Appendix \ref{app:lower}).

\begin{proposition}[Divergence bound as an expectation]\label{pro:hsdnumerical2}
For $p> 1, \beta\in [0, \frac{p-1}{p+1}],  q_0,q_1\in [1,\infty)$ that $q_0/q_1\in [1/p,p]$, let $\alpha=\frac{\beta}{p-1}$, $r_0=\frac{\alpha p}{q_0}$, and $r_1=\frac{\alpha p}{q_1}$, let $low_c=\frac{(e^{\epsilon}p-1)\alpha c/r_1+(e^{\epsilon}-1)(1-\alpha-\alpha p)(n-c)/(1-r_0-r_1)}{\alpha(p/r_0-1/r_1+e^{\epsilon}(p/r_1-1/r_0))}$ and $high_c=\frac{(e^{-\epsilon}p-1)\alpha c/r_1+(e^{-\epsilon}-1)(1-\alpha-\alpha p)(n-c)/(1-r_0-r_1)}{\alpha(p/r_0-1/r_1+e^{-\epsilon}(p/r_1-1/r_0))}$, then for any $\epsilon \in \mathbb{R}$:
\begin{align*}
&&& D_{e^\epsilon}(P_{p,\beta}^{q_0,q_1} \| Q_{p,\beta}^{q_0,q_1})= \mathop{\mathbb{E}}\limits_{c\sim Binomial(n-1,r_0+r_1)}\Big[ (p-e^{\epsilon})\alpha\cdot \mathop{\mathsf{CDF}}\limits_{c,\frac{r_0}{r_0+r_1}}[\lceil low_{c+1}-1\rceil, c] &\\
&&&\ \ \ \ \ \ \ \ \ \ \ \ \ \ \ \ \ \ \ \ \ \ \ \ \ \ \ \ \ \ \ \ \ +(1-p e^\epsilon)\alpha\cdot\mathop{\mathsf{CDF}}\limits_{c,\frac{r_0}{r_0+r_1}}[\lceil low_{c+1}\rceil, c] +(1-e^\epsilon)(1-\alpha-p\alpha)\cdot \mathop{\mathsf{CDF}}\limits_{c,\frac{r_0}{r_0+r_1}}[\lceil low_c\rceil, c]\Big], &\\
&&& D_{e^\epsilon}(Q_{p,\beta}^{q_0,q_1} \| P_{p,\beta}^{q_0,q_1})= \mathop{\mathbb{E}}\limits_{c\sim Binomial(n-1,r_0+r_1)}\Big[(1-pe^{\epsilon})\alpha\cdot \mathop{\mathsf{CDF}}\limits_{c,\frac{r_0}{r_0+r_1}}[0,\lfloor high_{c+1}-1\rceil] &\\
&&&\ \ \ \ \ \ \ \ \ \ \ \ \ \ \ \ \ \ \ \ \ \ \ \ \ \ \ \ \ \ \ \ \ +(p- e^\epsilon)\alpha\cdot\mathop{\mathsf{CDF}}\limits_{c,\frac{r_0}{r_0+r_1}}[0,\lfloor high_{c+1}\rfloor] +(1-e^\epsilon)(1-\alpha-p\alpha)\cdot \mathop{\mathsf{CDF}}\limits_{c,\frac{r_0}{r_0+r_1}}[0,\lfloor high_c\rfloor]\Big].&
\end{align*}
\end{proposition}

It is worth noting that the upper bound of indistinguishability between $P^{q_0,q_1}_{p,\beta}$ and $Q^{q_0,q_1}_{p,\beta}$, which refers to the minimum value of $\epsilon'$ that satisfies $\max[D_{e^{\epsilon'}}(P^{q_0,q_1}_{p,\beta}\| Q^{q_0,q_1}_{p,\beta}), D_{e^{\epsilon'}}(Q^{q_0,q_1}_{p,\beta}\| P^{q_0,q_1}_{p,\beta})]\leq \delta$, can be obtained in a similar manner as that of the lower bound. The only difference is that, in Algorithm \ref{alg:biasedlowerbound}, we return $\epsilon_H$ instead of $\epsilon_L$ at the last line. This would be useful to derive precise amplification upper bounds for specific randomizers that are not tight in Theorem \ref{the:reduction}, such as randomized response on $2$ options \cite{warner1965randomized}, local hash with length $l=2$ \cite{wang2017locally}, and exponential mechanism for metric LDP on $3$ options \cite{mcsherry2007mechanism}.

\section{Efficient search of the indistinguishable lower bound}\label{app:lower}
The efficient implementation in Algorithm \ref{alg:biasedlowerbound} relies on Proposition \ref{pro:hsdnumerical2}, which expresses the divergence as an expectation.
\begin{algorithm}[htbp]
    \caption{Efficient Search of the  indistinguishable lower bound of $P^{q_0,q_1}_{p_0,\beta}$ and $Q^{q_0,q_1}_{p_0,\beta}$}
    \label{alg:biasedlowerbound}
    \KwIn{privacy parameter $\delta$, number of clients $n$, property parameters $p_0>1, \beta\in [0, \frac{p_0-1}{p_0+1}]$ and $q_0,q_1\geq 1$ that $q_0/q_1\in [1/p_0,p_0]$, number of iterations $T$.}
    \KwOut{A lower bound of $\epsilon'_c$ that  $\max[D_{\epsilon'_c}(P^{q_0,q_1}_{p_0,\beta}\| Q^{q_0,q_1}_{p_0,\beta}), D_{\epsilon'_c}(Q^{q_0,q_1}_{p_0,\beta}\| P^{q_0,q_1}_{p_0,\beta})]\geq \delta$ holds.}

    {$\alpha=\frac{\beta}{p_0-1}$,\ \ $r_0=\frac{\alpha p_0}{q_0}$,\ \ $r_1=\frac{\alpha p_0}{q_1}$}

    \SetKwProg{myproc}{Procedure}{}{}
    \SetKwFunction{proc}{{Delta}}
    \myproc{\proc{$\epsilon'$}}{
        {$\delta'_0 \leftarrow 0$,}\ \ \ \ {$\delta'_1 \leftarrow 0$}

        {$w_c=\binom{n-1}{c}(2r)^c(1-2r)^{n-1-c}$}

        {$low_c = \frac{(e^{\epsilon'}p_0-1)\alpha c/r_1+(e^{\epsilon'}-1)g}{\alpha(p/r_0-1/r_1+e^{\epsilon'}(p_0/r_1-1/r_0))}$}

        {$high_c = \frac{(e^{-\epsilon'}p_0-1)\alpha c/r_1+(e^{-\epsilon'}-1)g}{\alpha(p_0/r_0-1/r_1+e^{-\epsilon'}(p_0/r_1-1/r_0))}$}
        
        \For{$c \in [0,n]$}{
            
            {$\delta'_0 \leftarrow \delta'_0+w_c\Big((p_0-e^{\epsilon})\alpha\cdot \mathop{\mathsf{CDF}}\limits_{c,\frac{r_0}{r_0+r_1}}[\lceil low_{c+1}-1\rceil, c]+(1-p_0 e^\epsilon)\alpha\cdot\mathop{\mathsf{CDF}}\limits_{c,\frac{r_0}{r_0+r_1}}[\lceil low_{c+1}\rceil, c] +(1-e^\epsilon)(1-\alpha-p_0\alpha)\cdot \mathop{\mathsf{CDF}}\limits_{c,\frac{r_0}{r_0+r_1}}[\lceil low_c\rceil, c]\Big)$}

            {$\delta'_1 \leftarrow \delta'_1+w_c\Big((1-p_0e^{\epsilon})\alpha\cdot \mathop{\mathsf{CDF}}\limits_{c,\frac{r_0}{r_0+r_1}}[0,\lfloor high_{c+1}-1\rceil]+(p_0- e^\epsilon)\alpha\cdot\mathop{\mathsf{CDF}}\limits_{c,\frac{r_0}{r_0+r_1}}[0,\lfloor high_{c+1}\rfloor]+(1-e^\epsilon)(1-\alpha-p\alpha)\cdot \mathop{\mathsf{CDF}}\limits_{c,\frac{r_0}{r_0+r_1}}[0,\lfloor high_c\rfloor]\Big)$}
    
        }
        \KwRet{$\max(\delta'_0, \delta'_1)$}
    }
    
    $\epsilon_{L}\leftarrow 0$,\ \ \ \ $\epsilon_{H}\leftarrow \log(p_0)$
    
    \For{$t \in [y]$}{
    {$\epsilon_t\leftarrow \frac{\epsilon_{L}+\epsilon_{H}}{2}$}
    
        \eIf{$\proc(\epsilon_t)> \delta$}
        {$\epsilon_{L}\leftarrow \epsilon_t$}
        {$\epsilon_{H}\leftarrow \epsilon_t$}
    }
    
    \KwRet{$\epsilon_L$}
\end{algorithm}

\section{Complementary amplification parameters of $\epsilon$-LDP mechanisms}\label{app:ldp}
Zhu \textit{et al.}\cite[Proposition 8]{zhu2022optimal} proved that for any given privatization mechanism, there always exists a \emph{tightly} dominating pair of distributions. Therefore, the tight upper bound on pairwise total variation (i.e., hockey-stick divergence with $e^{\epsilon}=1$) can be  computed from the tightly dominating pair. 

\begin{table}[h!]
\setlength{\tabcolsep}{0.30em}
\renewcommand{\arraystretch}{1.5}
\caption{Additional amplification parameters of $\epsilon$-LDP randomizers.}
\label{tab:parametersldp2}
\centering
\begin{tabular}{|c|c|c|c|}
\hline
\bfseries randomizer & \bfseries param. $p$ & \bfseries param. $\beta$ & \bfseries param. $q$ \\  
\hline
general mechanisms & $e^{\epsilon}$ & $\frac{e^{\epsilon}-1}{e^{\epsilon}+1}$ & $e^{\epsilon}$ \\
\hline

Duchi \textit{et al.} \cite{duchi2013local}  for $[-1,1]^d$  & $e^{\epsilon}$ & $\frac{e^{\epsilon}-1}{e^{\epsilon}+1}$ & $e^{\epsilon}$ \\

Harmony mechanism for $[-1,1]^d$ \cite{nguyen2016collecting} & $e^{\epsilon}$ & $\frac{e^{\epsilon}-1}{e^{\epsilon}+1}$ & $e^{\epsilon}$ \\



$k$-subset exponential on $s$ in $d$ options \cite{wang2018privset}  & $e^{\epsilon}$ & $\frac{(e^{\epsilon}-1)({d-s\choose k}-{d-2s\choose k})}{e^{\epsilon}({d\choose k}-{d-s\choose k})+{d-s\choose k}}$ & $e^{\epsilon}$ \\



\hline
\end{tabular}
\end{table}

\end{document}